\documentclass[11pt,a4paper]{article}
\usepackage{float}
\makeatletter
\gdef\@fpheader{Prepared for submission to JHEP}
\gdef\@journal{jhep}

\newif\ifnotoc\notocfalse
\newif\ifemailadd\emailaddfalse
\newif\iftoccontinuous\toccontinuousfalse
\newif\ifnatbibsort\natbibsorttrue

\RequirePackage{amsmath}
\RequirePackage{amssymb}
\RequirePackage{epsfig}
\RequirePackage{graphicx}
\ifnatbibsort\RequirePackage[numbers,sort&compress]{natbib}\else\RequirePackage[numbers,compress]{natbib}\fi
\RequirePackage{color}
\RequirePackage[colorlinks=true
,urlcolor=blue
,anchorcolor=blue
,citecolor=blue
,filecolor=blue
,linkcolor=blue
,menucolor=blue
,pagecolor=blue
,linktocpage=true
,pdfproducer=medialab
,pdfa=true
]{hyperref}

\def\@subheader{\@empty}
\def\@keywords{\@empty}
\def\@abstract{\@empty}
\def\@xtum{\@empty}
\def\@dedicated{\@empty}
\def\@arxivnumber{\@empty}
\def\@collaboration{\@empty}
\def\@collaborationImg{\@empty}
\def\@proceeding{\@empty}
\def\@preprint{\@empty}

\newcommand{\subheader}[1]{\gdef\@subheader{#1}}
\newcommand{\keywords}[1]{\if!\@keywords!\gdef\@keywords{#1}\else%
\PackageWarningNoLine{\jname}{Keywords already defined.\MessageBreak Ignoring last definition.}\fi}
\renewcommand{\abstract}[1]{\gdef\@abstract{#1}}
\newcommand{\dedicated}[1]{\gdef\@dedicated{#1}}
\newcommand{\arxivnumber}[1]{\gdef\@arxivnumber{#1}}
\newcommand{\proceeding}[1]{\gdef\@proceeding{#1}}
\newcommand{\xtumfont}[1]{\textsc{#1}}
\newcommand{\correctionref}[3]{\gdef\@xtum{\xtumfont{#1} \href{#2}{#3}}}
\newcommand\jname{JHEP}

\newcommand\preprint[1]{\gdef\@preprint{\hfill #1}}

\newcommand\note[2][]{%
\if!#1!%
\stepcounter{footnote}\footnotetext{#2}%
\else%
{\renewcommand\thefootnote{#1}%
\footnotetext{#2}}%
\fi}

\newtoks\auth@toks
\renewcommand{\author}[2][]{%
  \if!#1!%
    \auth@toks=\expandafter{\the\auth@toks#2\ }%
  \else
    \auth@toks=\expandafter{\the\auth@toks#2$^{#1}$\ }%
  \fi
}

\newtoks\affil@toks\newif\ifaffil\affilfalse
\newcommand{\affiliation}[2][]{%
\affiltrue
  \if!#1!%
    \affil@toks=\expandafter{\the\affil@toks{\item[]#2}}%
  \else
    \affil@toks=\expandafter{\the\affil@toks{\item[$^{#1}$]#2}}%
  \fi
}

\newtoks\email@toks\newcounter{email@counter}%
\newcommand{\emailAdd}[1]{%
\emailaddtrue%
\ifnum\theemail@counter>0\email@toks=\expandafter{\the\email@toks, \@email{#1}}%
\else\email@toks=\expandafter{\the\email@toks\@email{#1}}%
\fi\stepcounter{email@counter}}
\newcommand{\@email}[1]{\href{mailto:#1}{\tt #1}}

\newcommand*\collaboration[1]{\gdef\@collaboration{#1}}
\newcommand*\collaborationImg[2][]{\gdef\@collaborationImg{#2}}

\newcommand\afterLogoSpace{\smallskip}
\newcommand\afterSubheaderSpace{\vskip3pt plus 2pt minus 1pt}
\newcommand\afterProceedingsSpace{\vskip21pt plus0.4fil minus15pt}
\newcommand\afterTitleSpace{\vskip23pt plus0.06fil minus13pt}
\newcommand\afterRuleSpace{\vskip23pt plus0.06fil minus13pt}
\newcommand\afterCollaborationSpace{\vskip3pt plus 2pt minus 1pt}
\newcommand\afterCollaborationImgSpace{\vskip3pt plus 2pt minus 1pt}
\newcommand\afterAuthorSpace{\vskip5pt plus4pt minus4pt}
\newcommand\afterAffiliationSpace{\vskip3pt plus3pt}
\newcommand\afterEmailSpace{\vskip16pt plus9pt minus10pt\filbreak}
\newcommand\afterXtumSpace{\par\bigskip}
\newcommand\afterAbstractSpace{\vskip16pt plus9pt minus13pt}
\newcommand\afterKeywordsSpace{\vskip16pt plus9pt minus13pt}
\newcommand\afterArxivSpace{\vskip3pt plus0.01fil minus10pt}
\newcommand\afterDedicatedSpace{\vskip0pt plus0.01fil}
\newcommand\afterTocSpace{\bigskip\medskip}
\newcommand\afterTocRuleSpace{\bigskip\bigskip}
\newlength{\affiliationsSep}
\newcommand\beforetochook{\pagestyle{myplain}\pagenumbering{roman}}

\DeclareFixedFont\trfont{OT1}{phv}{b}{sc}{11}

\renewcommand\maketitle{
\pagestyle{empty}
\thispagestyle{titlepage}
\setcounter{page}{0}
\noindent{\small\scshape\@fpheader}\@preprint\par
\afterLogoSpace
\if!\@subheader!\else\noindent{\trfont{\@subheader}}\fi
\afterSubheaderSpace
\if!\@proceeding!\else\noindent{\sc\@proceeding}\fi
\afterProceedingsSpace
{\LARGE\flushleft\sffamily\bfseries\@title\par}
\afterTitleSpace
\hrule height 1.5\p@%
\afterRuleSpace
\if!\@collaboration!\else
{\Large\bfseries\sffamily\raggedright\@collaboration}\par
\afterCollaborationSpace
\fi
\if!\@collaborationImg!\else
{\normalsize\bfseries\sffamily\raggedright\@collaborationImg}\par
\afterCollaborationImgSpace
\fi
{\bfseries\raggedright\sffamily\the\auth@toks\par}
\afterAuthorSpace
\ifaffil\begin{list}{}{%
\setlength{\leftmargin}{0.28cm}%
\setlength{\labelsep}{0pt}%
\setlength{\itemsep}{\affiliationsSep}%
\setlength{\topsep}{-\parskip}}
\itshape\small%
\the\affil@toks
\end{list}\fi
\afterAffiliationSpace
\ifemailadd 
\noindent\hspace{0.28cm}\begin{minipage}[l]{.9\textwidth}
\begin{flushleft}
\textit{E-mail:} \the\email@toks
\end{flushleft}
\end{minipage}
\else 
\PackageWarningNoLine{\jname}{E-mails are missing.\MessageBreak Plese use \protect\emailAdd\space macro to provide e-mails.}
\fi
\afterEmailSpace
\if!\@xtum!\else\noindent{\@xtum}\afterXtumSpace\fi
\if!\@abstract!\else\noindent{\renewcommand\baselinestretch{.9}\textsc{Abstract:}}\ \@abstract\afterAbstractSpace\fi
\if!\@keywords!\else\noindent{\textsc{Keywords:}} \@keywords\afterKeywordsSpace\fi
\if!\@arxivnumber!\else\noindent{\textsc{ArXiv ePrint:}} \href{http://arxiv.org/abs/\@arxivnumber}{\@arxivnumber}\afterArxivSpace\fi
\if!\@dedicated!\else\vbox{\small\it\raggedleft\@dedicated}\afterDedicatedSpace\fi
\ifnotoc\else
\iftoccontinuous\else\newpage\fi
\beforetochook\hrule
\tableofcontents
\afterTocSpace
\hrule
\afterTocRuleSpace
\fi
\setcounter{footnote}{0}
\pagestyle{myplain}\pagenumbering{arabic}
} 

\renewcommand{\baselinestretch}{1.1}\normalsize
\divide\@tempdima\baselineskip
\@tempcnta=\@tempdima
\skip\@mpfootins = \skip\footins
\renewcommand{\@dotsep}{10000}

\newcommand\ps@myplain{
\pagenumbering{arabic}
\renewcommand\@oddfoot{\hfill-- \thepage\ --\hfill}
\renewcommand\@oddhead{}}
\let\ps@plain=\ps@myplain

\newcommand\ps@titlepage{\renewcommand\@oddfoot{}\renewcommand\@oddhead{}}

\numberwithin{equation}{section}

\renewcommand\section{\@startsection{section}{1}{\z@}%
                                   {-3.5ex \@plus -1.3ex \@minus -.7ex}%
                                   {2.3ex \@plus.4ex \@minus .4ex}%
                                   {\normalfont\large\bfseries}}
\renewcommand\subsection{\@startsection{subsection}{2}{\z@}%
                                   {-2.3ex\@plus -1ex \@minus -.5ex}%
                                   {1.2ex \@plus .3ex \@minus .3ex}%
                                   {\normalfont\normalsize\bfseries}}
\renewcommand\subsubsection{\@startsection{subsubsection}{3}{\z@}%
                                   {-2.3ex\@plus -1ex \@minus -.5ex}%
                                   {1ex \@plus .2ex \@minus .2ex}%
                                   {\normalfont\normalsize\bfseries}}
\renewcommand\paragraph{\@startsection{paragraph}{4}{\z@}%
                                   {1.75ex \@plus1ex \@minus.2ex}%
                                   {-1em}%
                                   {\normalfont\normalsize\bfseries}}
\renewcommand\subparagraph{\@startsection{subparagraph}{5}{\parindent}%
                                   {1.75ex \@plus1ex \@minus .2ex}%
                                   {-1em}%
                                   {\normalfont\normalsize\bfseries}}

\def\fnum@figure{\textbf{\figurename\nobreakspace\thefigure}}
\def\fnum@table{\textbf{\tablename\nobreakspace\thetable}}

\long\def\@makecaption#1#2{%
  \vskip\abovecaptionskip
  \sbox\@tempboxa{\small #1. #2}%
  \ifdim \wd\@tempboxa >\hsize
    \small #1. #2\par
  \else
    \global \@minipagefalse
    \hb@xt@\hsize{\hfil\box\@tempboxa\hfil}%
  \fi
  \vskip\belowcaptionskip}

\renewenvironment{thebibliography}[1]{%
\begin{oldthebibliography}{#1}%
\small%
\raggedright%
\setlength{\itemsep}{5pt plus 0.2ex minus 0.05ex}%
}%
{%
\end{oldthebibliography}%
}
\makeatother

\usepackage{comment}
\usepackage{amsthm,booktabs,mathtools,tikz,array,longtable}
\usetikzlibrary{arrows.meta,positioning,calc}
\newtheorem{theorem}{Theorem}[section]

\theoremstyle{definition}

\newcommand{\Tri}{\operatorname{Tri}}
\newcommand{\Span}{\operatorname{span}}

\newcommand{\Rz}{\mathcal R}
\newcommand{\Pz}{\mathcal P}
\newcommand{\Cz}{\mathcal C}
\newcommand{\Xz}{\mathcal X}
\newcommand{\Dz}{\mathcal D}
\newcommand{\K}{\mathcal K}

\let\originaltableofcontents\tableofcontents
\renewcommand{\tableofcontents}{\begingroup\small\originaltableofcontents\endgroup}
\renewcommand{\afterTocRuleSpace}{\clearpage}
\title{From Hidden Zeros to Exact Splittings and New Zeros}
\author[a,b]{Yang Li}
\author[c]{Laurentiu Rodina}
\affiliation[a]{Van Swinderen Institute for Particle Physics and Gravity, University of Groningen, 9747 AG Groningen, The Netherlands}
\affiliation[b]{Department of Physics, Technion -- Israel Institute of Technology, Haifa 3200003, Israel}
\affiliation[c]{Beijing Institute of Mathematical Sciences and Applications (BIMSA), Beijing 101408, China}

\emailAdd{leey0342@gmail.com}
\emailAdd{laurentiu.rodina@gmail.com}

\abstract{
Hidden zeros impose constraints on scattering amplitudes
beyond their familiar factorization properties.
We explore their structure in two related directions.
We first promote hidden zeros to explicit splitting formulas
that organize full amplitudes at generic kinematics.
Ordered $\operatorname{Tr}(\phi^3)$ amplitudes, nonlinear sigma model (NLSM) amplitudes,
and cosmological wavefunctions share
the same defining exact splitting structure
$\mathcal A=\sum_\alpha R_\alpha\,
 \mathcal J_{\alpha,L}\mathcal J_{\alpha,R}
 \mathcal J_{\alpha,\mathrm{out}}$, adapted to any hidden zero. Each prefactor $R_\alpha$ is a condition of the hidden zero locus, so this decomposition manifests the hidden vanishing term by term. The remaining factors are lower point amplitudes or currents. Using these representations,
we can systematically find new zero loci beyond the usual ones. In particular, we obtain five new all-multiplicity zero families. We also find continuous families of  zeros:
the coefficients of their defining equations can vary
continuously while the amplitude remains zero.
}
\keywords{Scattering Amplitudes, Scalar Field Theories, Amplitude zeros, New representations of amplitudes}
\hypersetup{pdftitle={From Hidden Zeros to Exact Splittings and New Zeros},pdfsubject={Hidden zeros as generic splitting formulas and new maximal linear zero families in cubic theory, the nonlinear sigma model, and tree wavefunctions}}
\begin{document}
\maketitle
\raggedbottom

\section{Introduction}

Locality and unitarity constrain scattering amplitudes through their poles:
when an intermediate particle goes on shell, the residue factorizes into
lower-point amplitudes. Hidden zeros provide complementary information.
They specify kinematic loci on which the full amplitude vanishes even
though every propagator remains finite. Their discovery in particle and
string amplitudes \cite{HiddenZeros} has prompted a broader study of
kinematic cancellations and the structures they reveal.

Positive geometry and related algebraic constructions provide a
geometric setting for studying amplitude poles, zeros and factorization
\cite{Arkani-Hamed:2013jha,Arkani-Hamed:2017tmz,ABHY,Gao:2017dek,Arkani-Hamed:2019vag,Arkani-Hamed:2019mrd,Arkani-Hamed:2019plo,CosmologicalPolytopes,Arkani-Hamed:2024jbp,Glew:2025otn,Ardila-Mantilla:2026cbo,Bartsch:2025mvy,Telen:2025ajb,Bruser:2025ezv}.
Complementary explanations and constructions of hidden zeros use
diagrammatic methods, amplitude relations and recursion
\cite{Zhou:2024ddy,PionFactorization,Zhou:2026fkg,Bartsch:2024amu,Li:2024qfp,Huang:2025blb,Feng:2025ofq,NLSMZeroRecursion,Zhou:2026isc,Gomez:2025tqx}.
Extensions to other theories and observables, together with related
surface constructions and relations between theories, have broadened
these questions across field theory, string theory and cosmology
\cite{NLSMShift,Arkani-Hamed:2024yvu,Paranjape:2025wjk,Dong:2024klq,Brauner:2025rzv,Cao:2026bog,Bartsch:2025loa,Bartsch:2026mnl,SurfaceCounting,SurfaceMultiplicity,Arkani-Hamed:2024vna,De:2024wsy,Arkani-Hamed:2024pzc,Arkani-Hamed:2023jry,Arkani-Hamed:2024tzl,Backus:2025orx,Backus:2025njt,Dong:2025lrf,Carrolo:2025jca,Laddha:2024qtn,Zhou:2025tvq,Zhang:2026dcm,CarrilloGonzalez:2026lnu,Zhang:2026wye,Cao:2025hio,Zhang:2025zjx,Chang:2025cqe,Cao:2025lzv,Arkani-Hamed:2024nzc,Figueiredo:2025fnr,Gomez:2026exq,HiddenCosmologicalZeros,Cao:2024bky,He:2026dql}. S-matrix bootstrap studies use zeros, splitting and
complementary constraints to restrict spectra and effective couplings
\cite{Cheung:2024uhn,Berman:2025owb,Cheung:2025tbr,Saha:2026ftv,Wan:2026pjq,Shao:2026akl,Basile:2026gnd,Cheung:2025nhw,Cheung:2026lpv,Elvang:2026pmc,Berman:2026ezk,Sung:2026euf}.
Hidden zeros were also found to strongly constrain amplitudes, leading to uniqueness theorems and the emergence of both locality and unitarity
\cite{RodinaUV,Rodina:2026yaq,Backus:2025hpn}.
For local amplitudes and cosmological wavefunction, they are equivalent
to dual shuffle factorization: a decomposition into order-preserving
shuffles of lower-point local building blocks
\cite{CosmologicalWavefunctions}. These results suggest that hidden zeros are much more fundamental and powerful than naicely would expect.

Of particular relevance for this current work are splitting formulas 
\cite{SmoothSplitting,TwoSplit,UniversalSplittings,Guevara:2024nxd,Zhang:2024efe,Feng:2025dci,Azevedo:2025vxo,AllOrderSplits,Zhou:2025xly,Bartsch:2024ofb,Zhou:2026ukg,GimenezUmbert:2025ech,Early:2025ivr}.
These formulas usually describe amplitudes on special kinematic
loci, including partially relaxed hidden zeros, or their behavior
in  soft limits. On the splitting loci, the
product identities are exact, although kinematic conditions remain
imposed. A notable exception is the exact splitting
construction \cite{JonesParanjape}, which
systematically relaxes hidden-zero conditions through on-shell
recursion. Their cubic boundary-row formula already holds at
generic kinematics for skinny rectangles; for wider rectangles,
it retains the conditions in the other rows.

In this work, we promote hidden zeros to a universal exact
splitting structure shared by ordered $\operatorname{Tr}(\phi^3)$
amplitudes, nonlinear sigma model (NLSM) amplitudes and cosmological wavefunctions. In all three settings, the
full function takes the form
\begin{equation}\label{eq:intro-universal-hosf}
 \boxed{\mathcal F=\sum_\alpha R_\alpha\,
 \mathcal J_{\alpha,L}\mathcal J_{\alpha,R}
 \mathcal J_{\alpha,\mathrm{out}}.}
\end{equation}
The $R_\alpha$ are linear conditions defining the chosen zero, and
each current contains a complete lower-point amplitude or lower tree
function, together with its attaching propagator or its NLSM
projection. These identities hold at generic kinematics, with every zero condition relaxed simultaneously.  We therefore obtain a new definition for amplitudes in these theories, built around manifesting each hidden zero. This turns the hidden zero from a simple yet mysterious property of amplitudes, to a universal constructive principle.

For cubic amplitudes, this result generalizes the formula found in \cite{JonesParanjape}, which applies only to ''skinny zeros". Writing
$X_{ij}=(p_i+\cdots+p_{j-1})^2$, define
\begin{equation}\label{eq:intro-cubic-current}
 \mathcal J^{\phi^3}[S]
 :=\frac{A^{\phi^3}[S]}{X_{s_1s_m}},\qquad
 S=(s_1,\ldots,s_m),\quad m\geq3,
 \qquad \mathcal J^{\phi^3}[n,1]:=1.
\end{equation}
For every $3\leq j\leq n-1$, the exact splitting formula
reads
\begin{equation}\label{eq:intro-hosf}
 A_n^{\phi^3}
 =\sum_{a=1}^{j-2}\sum_{k=j+1}^{n}R_{ak}\,
 \mathcal J^{\phi^3}[a,\ldots,j]\,
 \mathcal J^{\phi^3}[j-1,\ldots,k]\,
 \mathcal J^{\phi^3}[k,\ldots,n,1,\ldots,a],
\end{equation}
where $R_{ak}=X_{aj}+X_{j-1,k}-X_{a,k}$ and $X_{1n}=0$.
Its $(j-2)(n-j)$ numerators form a basis of the chosen rectangular
zero conditions: each independent condition multiplies one product
of complete lower amplitudes.

For even $n$, the formula gives $A_n^{\mathrm{NLSM}}$ by replacing
$\mathcal J^{\phi^3}$ with $\mathcal J^{\mathrm{NLSM}}$:
\begin{equation}\label{eq:intro-nlsm-projection}
 \mathcal J^{\phi^3}[S]
 \longmapsto\mathcal J^{\mathrm{NLSM}}[S]=
 \begin{cases}
 \displaystyle\frac{A^{\mathrm{NLSM}}[S]}{X_{s_1s_m}},
       &m\ \text{even},\\[6pt]
 (-1)^{s_1}A^{\mathrm{ext}}[S],&m\ \text{odd}.
 \end{cases}
\end{equation}
 The mixed amplitude $A^{\mathrm{ext}}$ has three scalars
and otherwise pions \cite{ExtendedNLSM}. Thus the cubic and NLSM formulas
have identical sums and zero numerators, with the theory dependence
carried entirely by the currents.

For cosmological wavefunctions, we derive
an exact splitting formula of the form~\eqref{eq:intro-universal-hosf} for rooted
blob zeros on arbitrary trees. Its three currents contain complete
functions on two subtrees meeting at the root and on the contracted
outside tree. The chain specialization recovers
eq.~\eqref{eq:intro-hosf}, extending the same construction from
ordered amplitudes to more general graphs.

The lower factors in each term allow provide a systematic way to discover new hidden zero loci. Setting every $R_\alpha$ to zero recovers the
original zero. Alternatively, a zero of a lower function removes
every term containing it, leaving fewer prefactors to constrain.
We use this mechanism to obtain five new all-multiplicity zero families
in cubic theory and determine their NLSM
counterparts. We also find non-blob zeros on branched trees. Finite brute force searches for zeros provide further examples whose general patterns remain to be understood.

We also uncover a drastically different kind of zeros: the coefficients of the defining linear equations
can vary continuously while the amplitude remains zero. For example, we find
$A_5^{\phi^3}$ vanishes on
\begin{equation}\label{eq:intro-continuous-five}
 \boxed{\begin{gathered}
 X_{24}=-\lambda X_{13},\qquad X_{25}=X_{14},\\
 (2-\lambda)X_{35}=\lambda(X_{13}+X_{14}).
 \end{gathered}}
\end{equation}
The exact splitting formula carries then this continuous family to higher points. 

Hidden zeros thus organize the full amplitude at generic kinematics, extending their significance beyond the special loci where the amplitude vanishes. The resulting splitting formulas reveal further families of zeros, showing that the zero structure of amplitudes is far richer than the familiar rectangular patterns suggest.

Section~\ref{sec:warmup} reviews the kinematics, hidden zeros and
earlier splitting formulas. Sections~\ref{sec:seed},
\ref{sec:wavefunctions} and \ref{sec:nlsm} develop the cubic, tree
and NLSM exact splitting formulas. Sections~\ref{sec:newzeros} and
\ref{sec:continuouszeros} present the new zero families and the
continuous examples. Section~\ref{sec: conclusion} discusses open
questions. Proofs and splitting formulas adapted to some of the
new zeros appear in appendix~\ref{app:proofs}. Further checks of
regularity and the distinction from ordinary rectangular zeros
appear in appendix~\ref{app:cubic-regularity}.

\section{Kinematics, hidden zeros and splitting formulas: a review}
\label{sec:warmup}

We review the kinematic conventions, known hidden zeros and the
on-shell splitting construction of ref.~\cite{JonesParanjape}.

\subsection{Hidden zeros of cubic amplitudes, the NLSM and tree wavefunctions}
\label{subsec:warmup-zeros}

\paragraph{Ordered cubic amplitudes.}
For massless momenta with $\sum_{i=1}^n p_i=0$, define
\begin{equation}
 X_{ij}=(p_i+\cdots+p_{j-1})^2,\qquad
 c_{ij}=-2p_i\cdot p_j.
 \label{eq:warmup-kinematics}
\end{equation}
The indices are cyclic,  and $X_{i,i}=X_{i,i+1}=0$.
 The remaining $d_n=n(n-3)/2$
planar channels are independent coordinates. 

For a cyclically ordered polygon $P$, the cubic amplitude is
\begin{equation}\label{eq:amplitude}
 A^{\phi^3}[P]=\sum_{T\in\Tri(P)}\prod_{e\in T}\frac1{X_e},
 \qquad A_3^{\phi^3}=1.
\end{equation}
Here $\Tri(P)$ is the set of all triangulations of $P$, and
each internal diagonal $e$ supplies one propagator. We abbreviate
$A^{\phi^3}[P]$ as $A[P]$ in the cubic sections and write
$A_n=A[1,\ldots,n]$.

Every lower factor is the complete amplitude of its indicated
ordered subpolygon. For $P=(i_1,\ldots,i_m)$, an internal
diagonal joining $i_a$ to $i_b$ carries the parent channel
$X_{i_a i_b}$. A boundary edge of $P$ supplies no propagator,
even when it is a diagonal of the parent polygon.

The elementary rectangle relation is
\begin{equation}
 c_{ij}=X_{ij}+X_{i+1,j+1}-X_{i,j+1}-X_{i+1,j}.
 \label{eq:warmup-rectangle}
\end{equation}
In the kinematic mesh, $X$'s label vertices and $c$'s label elementary
plaquettes \cite{ABHY,HiddenZeros}. Summing
eq.~\eqref{eq:warmup-rectangle} over a rectangle
cancels the interior $X$'s and leaves its four corners.
For the maximal rectangle based at $X_{1j}$, the two side corners
are boundary variables and vanish. Consequently,
\begin{equation}
 \sum_{u=1}^{j-2}\sum_{v=j}^{n-1}c_{uv}
 =X_{1j}+X_{j-1,n}.
 \label{eq:warmup-full-rectangle}
\end{equation}
The associated hidden zero is
\begin{equation}
 \left.A_n^{\phi^3}\right|_{\substack{c_{uv}=0\\
  1\leq u\leq j-2,\ j\leq v\leq n-1}}=0,
 \qquad 3\leq j\leq n-1.
 \label{eq:warmup-cubic-zero}
\end{equation}
These are the ordinary rectangular hidden zeros
\cite{HiddenZeros}. 

Equivalently, write the cyclic ordering as $(A,j-1,B,n)$, with
$A=\{1,\ldots,j-2\}$ and $B=\{j,\ldots,n-1\}$.
The distinguished legs are $j-1$ and $n$, and the zero sets
all cross contractions $p_a\cdot p_b=0$, $a\in A$, $b\in B$.
A \emph{skinny} rectangle has one row or one column.
After cyclic relabelling it can be represented by $j=3$.
A non-skinny rectangle has both $j-2>1$ and $n-j>1$.
For example, the six-point representatives are
\begin{equation}
 \begin{array}{ll}
 \text{skinny:}&c_{13}=c_{14}=c_{15}=0,\\[2pt]
 \text{non-skinny:}&c_{14}=c_{15}=c_{24}=c_{25}=0.
 \end{array}
 \label{eq:warmup-six-rectangles}
\end{equation}
These zeros arise from cancellations away from factorization poles.
The simplest example appears for the four-point amplitude
\begin{equation}
 A_4^{\phi^3}
 =\frac1{X_{13}}+\frac1{X_{24}}
 =\frac{c_{13}}{X_{13}X_{24}},
 \qquad c_{13}=X_{13}+X_{24}.
 \label{eq:warmup-four-cubic}
\end{equation}
so that setting $c_{13}=0$ forces $A_4^{\phi^3}=0$.

\paragraph{The NLSM from the $\delta$-shift.}
At even multiplicity the NLSM is obtained from the cubic amplitude
by the large-$\delta$ shift \cite{NLSMShift}. Shifting every $X_{ij}$ by
\begin{equation}
 \begin{aligned}
 X_{ij}^{(\delta)}&=X_{ij}+\delta\eta_{ij},\qquad
 \eta_{ij}=
 \begin{cases}
  +1,&i,j\ \text{even},\\
  -1,&i,j\ \text{odd},\\
  0,&i,j\ \text{of opposite parity},
 \end{cases}
 \end{aligned}
 \label{eq:warmup-delta}
\end{equation}
we obtain the NLSM amplitude as the leading order 
\begin{equation}
 A_n^{\mathrm{NLSM}}(X)=\lim_{\delta\to\infty}\delta^{n-2}
     A_n^{\phi^3}(X^{(\delta)}).
 \end{equation}
We write $A^{\mathrm{NLSM}}[P]$ for the complete ordered pion
amplitude, using the same inherited-channel convention. Our
normalization for cyclically ordered labels is
\begin{equation}\label{eq:nlsmfour}
 A_4^{\mathrm{NLSM}}[a,b,c,d]=-(X_{ac}+X_{bd}),
\end{equation}
where channel endpoints are unordered. In particular, the shift gives
\begin{equation}
 A_4^{\mathrm{NLSM}}
 =\lim_{\delta\to\infty}\delta^2
 \left(\frac1{X_{13}-\delta}+\frac1{X_{24}+\delta}\right)
 =-(X_{13}+X_{24})=-c_{13}.
 \label{eq:warmup-four-nlsm}
\end{equation}
Every $c_{uv}$ inside the rectangle in
eq.~\eqref{eq:warmup-cubic-zero} is invariant under this deformation:
the shifts of its four corner variables cancel in
eq.~\eqref{eq:warmup-rectangle}. Thus, if those $c$'s vanish before
the shift, they vanish for every $\delta$. The shifted cubic
amplitude is identically zero, and its leading large-$\delta$
coefficient gives
\begin{equation}
 \left.A_n^{\mathrm{NLSM}}\right|_{\substack{c_{uv}=0\\
  1\leq u\leq j-2,\ j\leq v\leq n-1}}=0,
 \qquad n\ \text{even}.
 \label{eq:warmup-nlsm-zero}
\end{equation}
The projection of a splitting formula also depends on its lower
factors. Even-point factors give NLSM amplitudes. The odd-point
factors used here give amplitudes $A^{\mathrm{ext}}$ with three
biadjoint scalars and the remaining particles pions
\cite{ExtendedNLSM,NLSMShift}. The normalization is
$A_3^{\mathrm{ext}}=1$.

\paragraph{Cosmological tree wavefunctions.}
Cosmological tree wavefunctions provide a natural extension of
our question: do hidden zeros and their splitting formulas persist
when a chain is replaced by a branched tree? For our purposes,
the scalar tree functions entering these wavefunctions generalize
ordered $\operatorname{Tr}(\phi^3)$ amplitudes through the kinematic
map described below. Detailed accounts of the wavefunctions,
their graph representations and their zeros can be found in
refs.~\cite{CosmologicalPolytopes,Arkani-Hamed:2024jbp,HiddenCosmologicalZeros,CosmologicalWavefunctions}.
We recall only the ingredients needed here.

Start with an arbitrary tree $G$ representing a cosmological
Feynman diagram, with its external legs suppressed. Its vertices
are interactions and its edges are internal lines. Write $V$ for
its vertex set and $N=|V|$.  The graph's contribution is organized by its \emph{maximal
tubings}. A proper tube is a connected set $S\subsetneq V$ with
$|S|\geq2$. Tubes are compatible when they are nested or disjoint. A tubing
is a collection of compatible tubes, and it is maximal when no
further proper tube can be added. After removing the universal
single-vertex and total-energy factors, the scalar tree coefficient
is a sum over maximal tubings, with each tube supplying a factor $1/E_S$.

The kinematic map of ref.~\cite{CosmologicalWavefunctions}
replaces these tube energies by Mandelstam invariants. Assign
a massless momentum to each vertex and add an auxiliary massless
momentum $p_q$, with $q\notin V$, imposing
\begin{equation}
 s_S=\left(\sum_{a\in S}p_a\right)^2,\qquad
 \sum_{a\in V}p_a+p_q=0,\qquad
 s_{\{a\}}=s_V=0,\qquad c_{ab}=-2p_a\cdot p_b.
 \label{eq:warmup-tree-map}
\end{equation}
With $\operatorname{MT}(G)$ denoting the maximal tubings, the
resulting on-shell object is
\begin{equation}
 \psi_G(s)=
 \sum_{\tau\in\operatorname{MT}(G)}
       \prod_{S\in\tau}\frac1{s_S},
 \qquad \psi_G=1\quad\text{for }|V|\leq2.
 \label{eq:warmup-tree-sum}
\end{equation}
Singletons and the full vertex set supply no propagators.
Figure~\ref{fig:warmup-tubings} shows individual terms of this sum.

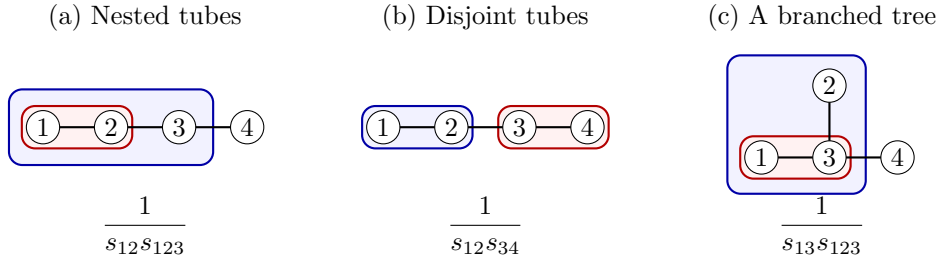
\begin{figure}[htbp]
\centering
\begin{tikzpicture}[x=1cm,y=1cm,
 vertex/.style={circle,draw=black,fill=white,minimum size=4.4mm,
                inner sep=0pt,font=\small},
 outer tube/.style={draw=blue!65!black,fill=blue!5,
                    rounded corners=5pt,thick},
 inner tube/.style={draw=red!70!black,fill=red!5,
                    rounded corners=5pt,thick},
 every node/.style={font=\small}]
 \begin{scope}
  \node at (1.35,1.85) {(a) Nested tubes};
  \path[outer tube] (-0.45,-0.1) rectangle (2.25,0.9);
  \path[inner tube] (-0.28,0.12) rectangle (1.18,0.68);
  \draw[thick] (0,0.4)--(2.7,0.4);
  \foreach \v/\x in {1/0,2/0.9,3/1.8,4/2.7}
   \node[vertex] at (\x,0.4) {$\v$};
  \node at (1.35,-0.9) {$\displaystyle\frac{1}{s_{12}s_{123}}$};
 \end{scope}
 \begin{scope}[xshift=4.5cm]
  \node at (1.35,1.85) {(b) Disjoint tubes};
  \path[outer tube] (-0.28,0.12) rectangle (1.18,0.68);
  \path[inner tube] (1.52,0.12) rectangle (2.98,0.68);
  \draw[thick] (0,0.4)--(2.7,0.4);
  \foreach \v/\x in {1/0,2/0.9,3/1.8,4/2.7}
   \node[vertex] at (\x,0.4) {$\v$};
  \node at (1.35,-0.9) {$\displaystyle\frac{1}{s_{12}s_{34}}$};
 \end{scope}
 \begin{scope}[xshift=10.4cm]
  \node at (-0.1,1.85) {(c) A branched tree};
  \path[outer tube] (-1.35,-0.48) rectangle (0.48,1.35);
  \path[inner tube] (-1.18,-0.28) rectangle (0.28,0.28);
  \draw[thick] (-0.9,0)--(0.9,0) (0,0)--(0,0.95);
  \node[vertex] at (-0.9,0) {$1$};
  \node[vertex] at (0,0.95) {$2$};
  \node[vertex] at (0,0) {$3$};
  \node[vertex] at (0.9,0) {$4$};
  \node at (-0.1,-0.9) {$\displaystyle\frac{1}{s_{13}s_{123}}$};
 \end{scope}
\end{tikzpicture}
\caption{Examples of maximal tubings and their contributions to
 $\psi_G(s)$. Each outline encloses one tube. Panels (a) and (b)
 show two of the five maximal tubings of the chain $1-2-3-4$;
 panel (c) shows one of the six for the star with edges $13,23,34$.
 Each tubing has two proper tubes. The singleton and full-graph
 factors have been stripped.}
\label{fig:warmup-tubings}
\end{figure}

As shown in ref.~\cite{CosmologicalWavefunctions}, this is precisely
the Cachazo--He--Yuan (CHY) integral whose integrand is the square of the Cayley
function associated with $G$ \cite{Gao:2017dek}. Thus the map
places the tree wavefunction in an on-shell construction that
generalizes ordered cubic amplitudes. The simplest graph, a chain
$P_m:1-2-\cdots-m$, gives the ordinary amplitude:
\begin{equation}\label{eq:pathisamplitude}
 \psi_{P_m}=A_{m+1},\qquad
 s_{\{i,\ldots,j\}}=X_{i,j+1},\qquad q=m+1.
\end{equation}
Throughout, $\psi_G$ denotes this stripped, kinematically mapped
function. The map is distinct from
taking a total-energy residue.

The corresponding \emph{blob zeros} were introduced in
ref.~\cite{CosmologicalWavefunctions}, extending the cosmological
zeros of ref.~\cite{HiddenCosmologicalZeros}. For the rooted
version used here, choose a vertex $r$ and partition the connected
components of $G-r$ into two nonempty collections with vertex
unions $A$ and $B$. Then
\begin{equation}
 \left.\psi_G(s)\right|_{c_{ab}=0}=0, \quad \ a\in A,\ b\in B
 \label{eq:warmup-blob}
\end{equation}
 The conditions set every
momentum contraction between the two sides to zero; the two
distinguished labels are $r$ and the auxiliary label $q$.
For the chain $P_{n-1}$, choosing $r=j-1$ recovers the ordinary
rectangular zero~\eqref{eq:warmup-cubic-zero}. In particular,
\begin{equation}
 \psi_{1-2-3}=\frac1{s_{12}}+\frac1{s_{23}}
 =\frac{c_{13}}{s_{12}s_{23}},
 \label{eq:warmup-path-seed}
\end{equation}
which is the four-point cancellation~\eqref{eq:warmup-four-cubic}.
Branching replaces the two intervals on either side of the root
by collections of branches, giving the zeros whose generic
splittings we study below.

For later use, the rectangle relation in these variables is
\begin{equation}
 s_I+s_J-s_{I\cup J}-s_{I\cap J}
 =\sum_{a\in I\setminus J}\sum_{b\in J\setminus I}c_{ab},
 \label{eq:warmup-tree-rectangle}
\end{equation}
with $s_\varnothing=0$. For connected rooted pieces
$\{r\}\cup I$ and $\{r\}\cup J$, with $I\subseteq A$ and
$J\subseteq B$, this reduces to
\begin{equation}
 s_{\{r\}\cup I}+s_{\{r\}\cup J}
 -s_{\{r\}\cup I\cup J}
 =\sum_{a\in I}\sum_{b\in J}c_{ab}.
 \label{eq:warmup-blob-rectangle}
\end{equation}
For $I=A$, $J=B$, the union is $V$ and the last $s$ vanishes.
This is the tree counterpart of eq.~\eqref{eq:warmup-full-rectangle}.

In graph subscripts a dash denotes an edge, parentheses indicate
branches, and $[U]$ denotes a contracted vertex; for example,
$\psi_{2-3-(4,6)}$ has edges $23,34,36$. Concatenated labels denote
sets, so $s_{134}=s_{\{1,3,4\}}$. The digits in $222$, $2222$ and
$223$ count the vertices on each arm, including the common centre.
Thus $222$ has four vertices, while $2222$ and $223$ each have five;
the auxiliary momentum adds one kinematic point.

\subsection{Linear, regular, and maximal zeros}
\label{subsec:zero-splitting-conventions}

We study zeros of these amplitudes on loci defined by homogeneous
linear relations among the planar variables $X_{ij}$. Such a locus
$Z$ is a \emph{regular linear zero} if $A_n|_Z=0$ as a rational
identity and no physical channel vanishes identically on $Z$.
The latter condition ensures that the restriction is well defined
at generic points of $Z$, so the vanishing describes a cancellation
away from factorization poles. We characterize $Z$ by its full space of defining linear
equations, $W_Z$, and its rank $r(Z)=\dim W_Z$, which counts the number of independent
kinematic constraints.

We count a regular linear zero as \emph{maximal} when it is not properly contained in any other regular linear zero at the same multiplicity. In other words, an additional condition imposed on a maximal zero defines a restricted locus, which we generally ignore. 

\subsection{On-shell construction and the scope of row relaxation}
\label{subsec:warmup-recursion}

\paragraph{A zero-preserving on-shell deformation.}
The splitting relations in ref.~\cite{JonesParanjape}
are derived by a BCFW-style contour argument \cite{BCFW}. For the corner choice $(X_{1j},c_{1,n-1})$, introduce a complex
parameter $z$ and set
\begin{equation}
 \widehat X_{ab}(z)=
 \begin{cases}
  X_{1b}-z,&a=1,\quad j\leq b\leq n-1,\\
  X_{an}+z,&b=n,\quad 2\leq a\leq j-1,\\
  X_{ab},&\text{otherwise}.
 \end{cases}
 \label{eq:warmup-bcfw-shift}
\end{equation}
Substitution in eq.~\eqref{eq:warmup-rectangle} verifies that every
$c_{uv}$ in the chosen rectangle is unchanged.
Define $\widehat A_n^{\phi^3}(z)
=A_n^{\phi^3}(\widehat X(z))$.
For this cubic deformation the large-$z$ behavior is
$O(z^{-2})$. Consequently Cauchy's theorem and the additional residue relation give
\begin{equation}
 A_n^{\phi^3}
 =-\sum_{z_\alpha}
   \frac{\operatorname{Res}_{z=z_\alpha}\widehat A_n^{\phi^3}(z)}
        {z_\alpha},
 \qquad
 \sum_{z_\alpha}
   \operatorname{Res}_{z=z_\alpha}\widehat A_n^{\phi^3}(z)=0.
 \label{eq:warmup-residue-theorems}
\end{equation}
The sum runs over the nonzero finite pole positions of the shifted
amplitude. Each residue is fixed by ordinary tree factorization.
The second equality, often called a bonus relation, permits one
residue to be eliminated from the first equality.

First keep only $c_{1,n-1}$ relaxed inside the rectangle.
All residues except those at $z=X_{1j}$ and $z=-X_{j-1,n}$
then contain a lower amplitude on a hidden zero.
Using the bonus relation and the surviving factorization gives
\begin{align}
 \left.A_n^{\phi^3}\right|_{\substack{
 c_{uv}=0\ {\rm in\ the\ rectangle}\\(u,v)\ne(1,n-1)}}
 &=
 \left(\frac1{X_{1j}}+\frac1{X_{j-1,n}}\right)
 A^{\phi^3}[1,\ldots,j]A^{\phi^3}[j-1,\ldots,n]
 \notag\\
 &=\frac{c_{1,n-1}}{X_{1j}X_{j-1,n}}\,
 A^{\phi^3}[1,\ldots,j]A^{\phi^3}[j-1,\ldots,n].
 \label{eq:warmup-corner-splitting}
\end{align}
The second line uses eq.~\eqref{eq:warmup-full-rectangle} on this
slice. Setting the last relaxed variable to zero recovers the
hidden zero itself. 

\paragraph{Relaxing a boundary row.}
Restoring further $c_{1v}$ introduces additional residues.
Iterating lower-point splittings and combining these residues gives
the boundary-row exact splitting formula of ref.~\cite{JonesParanjape}.
In our notation its domain is
\begin{equation}
 c_{uv}=0,\qquad 2\leq u\leq j-2,\quad j\leq v\leq n-1;
 \qquad c_{1v}\ \text{arbitrary}.
 \label{eq:warmup-row-slice}
\end{equation}
On precisely this slice the formula reads
\begin{align}
 \left.A_n^{\phi^3}\right|_{\eqref{eq:warmup-row-slice}}
 ={}&
 \frac{\displaystyle\sum_{v=j}^{n-1}c_{1v}}
      {X_{1j}X_{j-1,n}}\,
 A^{\phi^3}[1,\ldots,j]A^{\phi^3}[j-1,\ldots,n]
 \notag\\
 &+\sum_{k=j+1}^{n-1}
 \frac{\displaystyle\sum_{v=j}^{k-1}c_{1v}}
      {X_{1j}X_{1,k}X_{j-1,k}}
 \notag\\[-1mm]
 &\qquad\times
 A^{\phi^3}[1,\ldots,j]\,
 A^{\phi^3}[j-1,\ldots,k]\,
 A^{\phi^3}[k,\ldots,n,1].
 \label{eq:warmup-jp-row}
\end{align}

Related cubic recursions use one-parameter kinematic deformations
and admit associahedral triangulation interpretations
\cite{HeYang,YangTriangulations}.

\paragraph{NLSM.}
For the NLSM, ref.~\cite{JonesParanjape} derives an
all-multiplicity formula expressing the amplitude as a sum
of products of lower-point pion amplitudes. It allows some
conditions defining a hidden zero to be relaxed, while
retaining others and imposing additional kinematic relations.
The formula therefore holds on a restricted locus, rather
than at generic kinematics.

\section{A fully relaxed splitting formula for \texorpdfstring{$\operatorname{Tr}(\phi^3)$}{Tr(phi3)}}\label{sec:seed}
We now extend the boundary-row identity reviewed in section~\ref{sec:warmup}
to fully relaxed kinematics. If a zero is defined by $r$ independent
linear equations, a splitting whose numerators span its equation
space requires at least $r$ factorized terms. We call a formula with
exactly $r$ terms \emph{strict}. The ${\rm Tr}(\phi^3)$ splitting formula below attains
this bound at every multiplicity. After deriving the corresponding
wavefunction formulas (section \ref{sec:wavefunctions}) and NLSM formulas (section \ref{sec:nlsm}), we use their lower factors to
construct further zeros in section~\ref{sec:newzeros}.

\subsection{The general exact splitting formula for ${\rm Tr}(\phi^3)$}

For the rectangle in eq.~\eqref{eq:warmup-cubic-zero},
introduce the cumulative relations
\begin{equation}\label{eq:cumulative}
 R_{ak}=X_{aj}+X_{j-1,k}-X_{a,k}
 =\sum_{u=a}^{j-2}\sum_{v=j}^{k-1}c_{uv}.
\end{equation}
These relations span the equation space of the rectangular zero.
Combine each complete lower amplitude with its attaching
propagator into a cubic current,
\begin{equation}\label{eq:cubic-dressed-current}
 \mathcal J^{\phi^3}[p_1,\ldots,p_m]
 :=\frac{A[p_1,\ldots,p_m]}{X_{p_1p_m}},
 \qquad m\geq3,
 \qquad \mathcal J^{\phi^3}[n,1]:=1.
\end{equation}
The unit convention accounts for the full cell
$(a,k)=(1,n)$, where there is no outside amplitude or attaching
propagator. In particular, the vanishing boundary channel
$X_{1n}$ never appears in a denominator.

\begin{theorem}[Generic ${\rm Tr}(\phi^3)$ exact splitting formula]\label{thm:rooted}
For $n\geq4$ and $3\leq j\leq n-1$, the ordered cubic amplitude
at generic kinematics obeys
\begin{equation}\label{eq:rooted}\boxed{
 A_n=\sum_{a=1}^{j-2}\sum_{k=j+1}^{n}R_{ak}\,
 \mathcal J^{\phi^3}[a,\ldots,j]\,
 \mathcal J^{\phi^3}[j-1,\ldots,k]\,
 \mathcal J^{\phi^3}[k,\ldots,n,1,\ldots,a]}.
\end{equation}
The identity has $(j-2)(n-j)$ complete-product terms. The
numerators $R_{ak}$ form a basis of the rectangular zero
conditions, and all denominators are physical channels.
\end{theorem}

Equation~\eqref{eq:rooted} makes the chosen zero an organizing
principle for a splitting of the full amplitude. Each independent
zero condition is paired with one product of complete lower
amplitudes, so the splitting is strict. 
Setting all $R_{ak}=0$ exposes the zero term by term.

\paragraph{Matrix form.}\label{sec:hosf-zero-organization}
The above formula also has a simple matrix representation. We indexed rows by $a=1,\ldots,j-2$ and columns
by $k=j+1,\ldots,n$. Collect the shared currents into vectors
\begin{equation}\label{eq:zero-factor-organization}
 \ell_a=\mathcal J^{\phi^3}[a,\ldots,j],\qquad
 r_k=\mathcal J^{\phi^3}[j-1,\ldots,k].
\end{equation}
The matrix coupling them contains the zero numerator and the
outside current:
\begin{equation}\label{eq:hosf-corner-entry}
 \mathsf K^{(j)}_{ak}
 =R_{ak}\,\mathcal J^{\phi^3}[k,\ldots,n,1,\ldots,a].
\end{equation}
The exact splitting formula is then simply
\begin{equation}
A_n=\boldsymbol\ell^{\mathsf T}\mathsf K^{(j)}\boldsymbol r.
\end{equation}
To display the entries explicitly, abbreviate
$\mathcal J_{uv}:=\mathcal J^{\phi^3}[u,\ldots,v]$, with the
labels taken in cyclic order and $\mathcal J_{n1}=1$. Then
\begin{equation}\label{eq:hosf-matrix}
 A_n=
 \begin{pmatrix}
  \mathcal J_{1j}&\cdots&\mathcal J_{j-2,j}
 \end{pmatrix}
 \begin{pmatrix}
  R_{1,j+1}\mathcal J_{j+1,1}
   &\cdots&R_{1n}\\
  \vdots&\ddots&\vdots\\
  R_{j-2,j+1}\mathcal J_{j+1,j-2}
   &\cdots&R_{j-2,n}\mathcal J_{n,j-2}
 \end{pmatrix}
 \begin{pmatrix}
  \mathcal J_{j-1,j+1}\\
  \vdots\\
  \mathcal J_{j-1,n}
 \end{pmatrix}.
\end{equation}
For a skinny rectangle, coincident first and last rows or
columns are included only once.
Equivalently, its individual terms form the rectangular matrix
\begin{equation}\label{eq:hosf-term-array}
 \mathsf H^{(j)}_{ak}=\ell_a\mathsf K^{(j)}_{ak}r_k,
 \qquad
 A_n=\sum_{a=1}^{j-2}\sum_{k=j+1}^{n}\mathsf H^{(j)}_{ak}.
\end{equation}
Every entry is one complete exact splitting summand. The current $\ell_a$
is common to row $a$, and $r_k$ to column $k$; the outside
current depends on both indices. This organization also survives
the NLSM projection and will be used to construct new zeros.

\paragraph{Proof and cancellations.}
Later, we will see that, Theorem~\ref{thm:rooted} is a path specialization of Theorem~\ref{thm:rootedtree} (see examples in
section~\ref{sec:treeexamples}). A combinatorial proof of Theorem~\ref{thm:rootedtree} is given in appendix~\ref{sec:treeproof}.

On the boundary-row slice~\eqref{eq:warmup-row-slice}, all terms
with $a\geq2$ vanish and $R_{1k}=\sum_{v=j}^{k-1}c_{1v}$, reducing
eq.~\eqref{eq:rooted} to eq.~\eqref{eq:warmup-jp-row}.

The exact splitting formula reproduces the amplitude by cancelling contributions
with incompatible channels. For a proper cell, the elementary identity is
\begin{equation}\label{eq:local}
 \frac{X_{aj}+X_{j-1,k}-X_{a,k}}
      {X_{aj}X_{j-1,k}X_{a,k}}
 =
 \frac1{X_{j-1,k}X_{a,k}}
 +\frac1{X_{aj}X_{a,k}}
 -\frac1{X_{aj}X_{j-1,k}}.
\end{equation}
For example, at five points,
\begin{align}\label{eq:crossingexample}
 A_5={}&
 \frac{X_{13}+X_{25}}{X_{13}X_{25}}
 A[1,2,3]A[2,3,4,5]\notag\\
 &+\frac{X_{13}+X_{24}-X_{14}}{X_{13}X_{24}X_{14}}
 A[1,2,3]A[2,3,4]A[4,5,1].
\end{align}
Using $A_3=1$ and $A[2,3,4,5]=1/X_{24}+1/X_{35}$,
the first term contains $+1/(X_{13}X_{24})$, which is not
a physical diagram because the two diagonals cross.
The second term contains its negative. These contributions
cancel, leaving exactly the five physical diagrams.
Thus individual exact splitting terms need not be sums of physical
diagrams, although their sum is the full amplitude.
Appendix~\ref{sec:treeproof} proves this cancellation at
arbitrary multiplicity.

\subsection{Six-point examples: skinny and non-skinny zeros}\label{sec:sixseeds}
At six points the nine independent planar channels are
\begin{equation}\label{eq:sixcoords}
 X_{13},\ X_{14},\ X_{15},\ X_{24},\ X_{25},\ X_{26},\ X_{35},\ X_{36},\ X_{46}.
\end{equation}
The full amplitude has fourteen cubic diagrams. The rooted identity collects them into three or four terms, according to the chosen rectangle.

The left skinny rectangle, denoted $H_L$, has equations
\begin{equation}
 H_L:\quad X_{13}+X_{26}=X_{13}+X_{24}-X_{14}=X_{13}+X_{25}-X_{15}=0.
\end{equation}
Its rooted identity is
\begin{equation}\label{eq:HL}
\begin{aligned}
A_6={}&\frac{X_{13}+X_{26}}{X_{13}X_{26}}A[2,3,4,5,6]
 \\
 &+\frac{X_{13}+X_{24}-X_{14}}{X_{13}X_{24}X_{14}}A[1,4,5,6]
 \\
 &+\frac{X_{13}+X_{25}-X_{15}}{X_{13}X_{25}X_{15}}A[2,3,4,5].
\end{aligned}
\end{equation}
Equation~\eqref{eq:HL} is the six-point $j=3$ case of the
boundary-row formula~\eqref{eq:warmup-jp-row}
\cite{JonesParanjape}. The three numerators are
$c_{13}+c_{14}+c_{15}$, $c_{13}$ and $c_{13}+c_{14}$, respectively.
The first term contains a complete five-point amplitude and the
other two contain complete four-point amplitudes.

The non-skinny rectangle has two rows and two columns. Its four independent equations can be written as
\begin{equation}\label{eq:HFzero}
 H_F:\quad X_{14}+X_{36}=X_{14}+X_{35}-X_{15}=X_{24}+X_{35}-X_{25}=X_{24}+X_{36}-X_{26}=0.
\end{equation}
Taking $j=4$ in eq.~\eqref{eq:rooted} gives
\begin{equation}\label{eq:HF}
\begin{aligned}
A_6={}&\frac{X_{14}+X_{36}}{X_{14}X_{36}}A[1,2,3,4]A[3,4,5,6]
 \\
 &+\frac{X_{14}+X_{35}-X_{15}}{X_{14}X_{35}X_{15}}A[1,2,3,4]
 \\
 &+\frac{X_{24}+X_{35}-X_{25}}{X_{24}X_{35}X_{25}}A[1,2,5,6]
 \\
 &+\frac{X_{24}+X_{36}-X_{26}}{X_{24}X_{36}X_{26}}A[3,4,5,6].
\end{aligned}
\end{equation}
The four displayed numerators are, in order,
$c_{14}+c_{15}+c_{24}+c_{25}$, $c_{14}+c_{24}$, $c_{24}$ and
$c_{24}+c_{25}$. They span the non-skinny rectangle
$c_{14}=c_{15}=c_{24}=c_{25}=0$.

On the boundary-row slice $c_{24}=c_{25}=0$, the last two terms
vanish and eq.~\eqref{eq:HF} reduces to the $n=6$, $j=4$ specialization of
eq.~\eqref{eq:warmup-jp-row}.
The last two terms, proportional to $c_{24}$ and $c_{24}+c_{25}$,
are absent on that slice and required at generic kinematics;
no lower amplitude is common to all four summands.

\section{The \texorpdfstring{$\delta$}{delta}-shift and exact splitting in the NLSM}
\label{sec:nlsm}

The $\delta$-shift carries the generic cubic splitting into the
NLSM, preserving its rectangular arrangement and explicit zero
numerators.

\subsection{One formula in terms of projected currents}
\label{subsec:nlsm-hosf}

In this section write $\mathcal J=\mathcal J^{\mathrm{NLSM}}$.
For an ordered polygon $S=(s_1,\ldots,s_m)$ with $m\geq3$
occurring in the exact splitting formula, keep its displayed ordering and
distinguish its closing parent channel $X_{s_1s_m}$. Define
\begin{equation}\label{eq:nlsm-projected-current}
 \boxed{\mathcal J[S]=
 \begin{cases}
 \displaystyle\frac{A^{\mathrm{NLSM}}[S]}{X_{s_1s_m}},
       &m\ \text{even},\\[7pt]
 (-1)^{s_1}A^{\mathrm{ext}}[S],&m\ \text{odd}.
 \end{cases}}
 \qquad \mathcal J[n,1]:=1.
\end{equation}
As in the cubic formula, the unit convention represents the
absence of an outside amplitude and attaching propagator at
the full cell $(a,k)=(1,n)$.

The endpoint labels are the original labels of the even-point
parent polygon. In the displayed exact splitting lists their parities
alternate, so the closing endpoints have opposite parity for
even $m$ and the same parity for odd $m$. In the odd case the
three scalar legs of $A^{\mathrm{ext}}[S]$ are
$s_1,s_{m-1},s_m$, in this scalar-trace order.

\begin{theorem}[Fully relaxed NLSM exact splitting formula]
\label{thm:nlsmroot}
For every even $n\geq4$ and $3\leq j\leq n-1$, the complete
NLSM amplitude obeys
\begin{equation}\label{eq:nlsm-current-hosf}
 A_n^{\mathrm{NLSM}}
 =\sum_{a=1}^{j-2}\sum_{k=j+1}^{n}R_{ak}\,
 \mathcal J[a,\ldots,j]\,
 \mathcal J[j-1,\ldots,k]\,
 \mathcal J[k,\ldots,n,1,\ldots,a].
\end{equation}
The numerators $R_{ak}$ are those of eq.~\eqref{eq:cumulative},
and the identity holds at generic planar kinematics.
\end{theorem}

In the matrix form~\eqref{eq:hosf-matrix}, replace each cubic
current $\mathcal J^{\phi^3}[S]$ in $\boldsymbol\ell$,
$\boldsymbol r$ and $\mathsf K$ by its projected current
$\mathcal J[S]$. The unit outside current remains one, so
$\mathsf K^{(j)}_{1n}=R_{1n}$ automatically. Thus the NLSM
formula preserves the rows, columns and zero numerators of
the cubic exact splitting formula, while its lower factors become pion or mixed
amplitudes. Setting every $R_{ak}$ to zero exposes the same
rectangular zero term by term.
\subsection{Derivation from the \texorpdfstring{$\delta$}{delta}-shift}
\label{subsec:nlsm-derivation}

For the inherited polygons in the exact splitting formula, the lower-point limits
under the parent shift~\eqref{eq:warmup-delta} are
\begin{equation}
\begin{aligned}
 A_m^{\phi^3}[S](X^{(\delta)})
 &=\delta^{-(m-2)}A_m^{\mathrm{NLSM}}[S]
   +O(\delta^{-(m-1)}), &&m\ \text{even},\\
 A_m^{\phi^3}[S](X^{(\delta)})
 &=\delta^{-(m-3)}A_m^{\mathrm{ext}}[S]
   +O(\delta^{-(m-2)}), &&m\ \text{odd}.
\end{aligned}
\label{eq:mixedblock}
\end{equation}
These limits use the inherited endpoint parities and the scalar
assignment above \cite{NLSMShift}. They concern the displayed
exact splitting polygons, rather than arbitrary parity assignments on
lower polygons.

For even $m$, the closing channel is unshifted. For odd $m$,
$X_{s_1s_m}^{(\delta)}=X_{s_1s_m}+(-1)^{s_1}\delta$;
its inverse supplies the extra power $\delta^{-1}$ and the
sign in eq.~\eqref{eq:nlsm-projected-current}. Consequently
both cases have the same current limit:
\begin{equation}\label{eq:nlsm-current-limit}
 \boxed{\mathcal J[S]=\lim_{\delta\to\infty}
 \delta^{m-2}\,
 \frac{A_m^{\phi^3}[S](X^{(\delta)})}
      {X_{s_1s_m}^{(\delta)}}.}
\end{equation}
This explains the term \emph{projected cubic current}: the
complete amplitude and its attaching propagator are projected
together. The attaching channel is a parent channel, although
it is a boundary edge of the lower polygon.

\begin{proof}[Proof of theorem~\ref{thm:nlsmroot}]
Group each lower amplitude in the cubic exact splitting formula
with its attaching propagator. For a proper cell $(a,k)$, the
three resulting currents belong to the inherited polygons
$[a,\ldots,j]$, $[j-1,\ldots,k]$, and
$[k,\ldots,n,1,\ldots,a]$. Their sizes sum to $n+4$, so the
powers $m-2$ in eq.~\eqref{eq:nlsm-current-limit} sum to
$n-2$. At the corner $(1,n)$, the outside current is one;
the sizes of the two remaining polygons sum to $n+2$ and
give the same power. Each $R_{ak}$ is invariant under the
shift by eq.~\eqref{eq:warmup-rectangle}. Multiplying the
cubic identity by $\delta^{n-2}$ and taking the limit
cell by cell therefore gives eq.~\eqref{eq:nlsm-current-hosf}.
\end{proof}

\subsection{Six-point skinny and non-skinny examples}
\label{subsec:nlsm-six-examples}

\paragraph{Skinny zero.}
For $j=3$, the three relaxed variables are $c_{13},c_{14},c_{15}$.
Equation~\eqref{eq:nlsm-current-hosf} gives
\begin{align}
 A_6^{\mathrm{NLSM}}[1,\ldots,6]
 ={}&-(c_{13}+c_{14}+c_{15})A_5^{\mathrm{ext}}[2,3,4,5,6]
 \notag\\
 &-\frac{c_{13}}{X_{14}}A_4^{\mathrm{NLSM}}[4,5,6,1]
 +\frac{c_{13}+c_{14}}{X_{25}}A_4^{\mathrm{NLSM}}[2,3,4,5].
 \label{eq:nlsmskinny}
\end{align}
The three terms correspond, respectively, to $(a,k)=(1,6),(1,4),(1,5)$.
The complete five-point mixed amplitude is
\begin{align}
 A_5^{\mathrm{ext}}[2,3,4,5,6]
 &=1+\frac{A_4^{\mathrm{NLSM}}[2,3,4,5]}{X_{25}}
      +\frac{A_4^{\mathrm{NLSM}}[3,4,5,6]}{X_{36}}
 \notag\\
 &=1-\frac{X_{24}+X_{35}}{X_{25}}
      -\frac{X_{35}+X_{46}}{X_{36}}.
 \label{eq:mixedfive}
\end{align}
Its scalar legs are $2,5,6$; legs $3,4$ are pions.
Substitution and collection of identical four-point factors gives
\begin{align}
 A_6^{\mathrm{NLSM}}
 ={}&-(c_{13}+c_{14}+c_{15})
 -\frac{c_{13}}{X_{14}}A_4^{\mathrm{NLSM}}[4,5,6,1]
 \notag\\
 &-\frac{c_{15}}{X_{25}}A_4^{\mathrm{NLSM}}[2,3,4,5]
 -\frac{c_{13}+c_{14}+c_{15}}{X_{36}}
 A_4^{\mathrm{NLSM}}[3,4,5,6].
 \label{eq:nlsm-skinny-resolved}
\end{align}
The first term is a contact contribution; the unexpanded
eq.~\eqref{eq:nlsmskinny} retains complete lower amplitudes in all terms.
The relevant rectangle identities are
\begin{equation}
\begin{aligned}
 c_{13}&=X_{13}+X_{24}-X_{14},&
 c_{14}&=X_{14}+X_{25}-X_{15}-X_{24},\\
 c_{15}&=X_{15}+X_{26}-X_{25},&
 c_{13}+c_{14}+c_{15}&=X_{13}+X_{26}.
\end{aligned}
\label{eq:nlsm-skinny-rectangles}
\end{equation}
Both forms vanish term by term on the skinny zero, but remain exact
when all three $c$'s are restored.

\paragraph{Non-skinny zero.}
At six points the inequivalent non-skinny rectangle is $j=4$, with
$c_{14},c_{15},c_{24},c_{25}$ relaxed. Its four cells give
\begin{align}
 A_6^{\mathrm{NLSM}}
 ={}&
 \frac{c_{14}+c_{15}+c_{24}+c_{25}}{X_{14}X_{36}}\,
 A_4^{\mathrm{NLSM}}[1,2,3,4]A_4^{\mathrm{NLSM}}[3,4,5,6]
 \notag\\
 &+\frac{c_{14}+c_{24}}{X_{14}}A_4^{\mathrm{NLSM}}[1,2,3,4]
 -\frac{c_{24}}{X_{25}}A_4^{\mathrm{NLSM}}[5,6,1,2]
 \notag\\
 &+\frac{c_{24}+c_{25}}{X_{36}}A_4^{\mathrm{NLSM}}[3,4,5,6].
 \label{eq:nlsm-nonskinny}
\end{align}
In order, these are $(a,k)=(1,6),(1,5),(2,5),(2,6)$.
All odd lower factors in these proper cells are three-point mixed
amplitudes, equal to one. The numerators are
\begin{equation}
\begin{aligned}
 c_{14}+c_{15}+c_{24}+c_{25}&=X_{14}+X_{36},\\
 c_{14}+c_{24}&=X_{14}+X_{35}-X_{15},\\
 c_{24}&=X_{24}+X_{35}-X_{25},\\
 c_{24}+c_{25}&=X_{24}+X_{36}-X_{26}.
\end{aligned}
\label{eq:nlsm-fat-rectangles}
\end{equation}
There is no residual condition such as $c_{14}+c_{15}=0$.
The four relaxed variables are independent.

As a direct check, substituting the four-point amplitudes and the
rectangle relations into either six-point exact splitting formula gives
\begin{align}
 A_6^{\mathrm{NLSM}}
 ={}&
 \frac{(X_{13}+X_{24})(X_{15}+X_{46})}{X_{14}}
 +\frac{(X_{15}+X_{26})(X_{24}+X_{35})}{X_{25}}
 \notag\\
 &+\frac{(X_{13}+X_{26})(X_{35}+X_{46})}{X_{36}}
 -(X_{13}+X_{24}+X_{35}+X_{46}+X_{15}+X_{26}).
 \label{eq:nlsm-full-six}
\end{align}
This also follows by extracting the $\delta^{-4}$ coefficient of
the sum of the fourteen cubic six-point diagrams.
















\section{Exact splitting for cosmological tree wavefunctions}
\label{sec:wavefunctions}

The same organization by a chosen zero extends to the
kinematically mapped tree wavefunctions reviewed in
subsection~\ref{subsec:warmup-zeros}: a rooted blob zero labels
a splitting valid at generic kinematics.
Each term is associated with two connected subsets meeting at a
chosen root and contains a blob-zero numerator multiplying complete
lower tree functions. We give the combinatorial prescription here. Appendix~\ref{sec:treeproof}
proves it by cancelling incompatible channel products and identifying
the survivors with maximal tubings. Its path-graph specialization
supplies the proof of theorem~\ref{thm:rooted}. Applications to
zeros beyond the rooted blob family are deferred to
subsection~\ref{subsec:tree-new-zeros}.

\subsection{Combinatorial construction for an arbitrary tree}
\label{sec:treeconstruction}

Let $G$ be a tree with $N$ vertices and choose a rooted branch
partition $V=\{r\}\sqcup A\sqcup B$ as in
eq.~\eqref{eq:warmup-blob}. The construction applies to every
such tree and partition. For $N\geq2$, each maximal tubing
has $N-2$ propagators, as shown in appendix~\ref{sec:treeproof}.

The admissible subsets on each side are
\begin{align}\label{eq:rootedtreepieces}
 \mathcal R_r(A)
 &=\{\emptyset\ne I\subseteq A:
                  G[\{r\}\cup I]\text{ is connected}\},\notag\\
 \mathcal R_r(B)
 &=\{\emptyset\ne J\subseteq B:
                  G[\{r\}\cup J]\text{ is connected}\},\\
 L_I&=\{r\}\cup I,\qquad
 R_J=\{r\}\cup J,\qquad U_{IJ}=L_I\cup R_J.\notag
\end{align}
Each pair $(I,J)$ is a cell, with $L_I\cap R_J=\{r\}$.
The rectangle relation gives its numerator,
\begin{equation}\label{eq:blobnumerators}
 C_{I\mid J}:=\sum_{a\in I}\sum_{b\in J}c_{ab}
       =s_{L_I}+s_{R_J}-s_{U_{IJ}},
\end{equation}
using $s_r=0$. The unique full cell is $(I,J)=(A,B)$, where
$U_{AB}=V$, $s_V=0$ and $C_{A\mid B}=s_{L_A}+s_{R_B}$.

The lower functions are evaluated on inherited channels. In an
induced subtree $G[S]$, each proper channel keeps its parent
variable $s_Q$; the full set $S$ and singletons supply no
propagators. For the outside function, contract a connected set
$U$ to a vertex $[U]$, with contraction map $\pi_U:G\to G/U$,
and define
\begin{equation}\label{eq:treepullback}
 \psi_{G/U}^{\langle U\rangle}(s)
 =\sum_{\tau\in\operatorname{MT}(G/U)}
   \prod_{Q\in\tau}\frac1{s_{\pi_U^{-1}(Q)}},
 \qquad
 \pi_U^{-1}(Q)=
 \begin{cases}
 Q,&[U]\notin Q,\\
 U\cup(Q\setminus\{[U]\}),&[U]\in Q.
 \end{cases}
\end{equation}
Thus the quotient is a complete lower tree function with
denominators $s_{\pi_U^{-1}(Q)}$, rather than differences
$s_{\pi_U^{-1}(Q)}-s_U$. No masslessness condition is imposed
on the total momentum in $U$. A lower zero must hold at these
inherited channel values, including any kinematic relations
required by its derivation.

As for cubic amplitudes, combine each lower function with its
attaching propagator. For a connected proper set $S$ with
$|S|\geq2$, define the subtree and outside currents
\begin{equation}\label{eq:tree-currents}
 \mathcal J_G[S]:=\frac{\psi_{G[S]}(s)}{s_S},\qquad
 \mathcal J_G^{\mathrm{out}}[S]
 :=\frac{\psi_{G/S}^{\langle S\rangle}(s)}{s_S},\qquad
 \mathcal J_G^{\mathrm{out}}[V]:=1.
\end{equation}
The unit convention represents the absence of an outside
amplitude and attaching propagator at the full cell. The
subtree currents only occur for proper sets, so no current
introduces a propagator $1/s_V$.

\begin{theorem}[Rooted tree exact splitting]
\label{thm:rootedtree}
For any tree $G$ and rooted branch partition $(r;A,B)$,
the mapped tree function obeys
\begin{align}
 \psi_G(s)
 &=\sum_{\substack{I\in\mathcal R_r(A)\\J\in\mathcal R_r(B)}}
 (s_{L_I}+s_{R_J}-s_{U_{IJ}})\,
 \mathcal J_G[L_I]\,\mathcal J_G[R_J]\,
 \mathcal J_G^{\mathrm{out}}[U_{IJ}]
 \label{eq:rootedtree}\\
 &=\sum_{\substack{I\in\mathcal R_r(A)\\J\in\mathcal R_r(B)}}
 C_{I\mid J}\,
 \mathcal J_G[L_I]\,\mathcal J_G[R_J]\,
 \mathcal J_G^{\mathrm{out}}[U_{IJ}].
 \label{eq:treehosfs}
\end{align}
The identity holds at generic kinematics, with all blob-zero
conditions relaxed. The sum includes the full cell $(A,B)$.
\end{theorem}

The prescription is therefore to choose one connected rooted
subset on each side, form their union, and multiply the two
subtree currents and the outside current by $C_{I\mid J}$.
At the full cell the outside current is one, giving
$C_{A\mid B}\mathcal J_G[L_A]\mathcal J_G[R_B]$ automatically.
All $c_{ab}$ remain arbitrary in the identity; imposing the
blob-zero conditions~\eqref{eq:warmup-blob} makes every term vanish.
For a path, the three currents become precisely the cubic
currents in eq.~\eqref{eq:rooted}, including the unit outside factor.

The number of cells is
\begin{equation}\label{eq:blobtermcount}
 K_{G,r;A,B}=|\mathcal R_r(A)|\,|\mathcal R_r(B)|.
\end{equation}
This need not equal the rank $|A||B|$ of the blob zero:
different cells can have linearly dependent numerators.
Thus the tree formula need not be strict.

\subsection{Examples}
\label{sec:treeexamples}

We now suppress the argument $s$.

\paragraph{Chains and the ordered cubic formula.}
Choose $r=j-1$ on $P_{n-1}$. Its cells have
\begin{equation}\label{eq:treechaincells}
 I=\{a,\ldots,j-2\},\quad J=\{j,\ldots,k-1\},\qquad
 1\leq a\leq j-2,\quad j+1\leq k\leq n.
\end{equation}
Their three channels map to $X_{aj}$, $X_{j-1,k}$ and
$X_{a,k}$, and $C_{I\mid J}=R_{ak}$ of eq.~\eqref{eq:cumulative}.
For a proper cell, the lower factors become
\begin{equation}
 A[a,\ldots,j],\qquad A[j-1,\ldots,k],\qquad
 A[k,\ldots,n,1,\ldots,a].
\end{equation}
At the full cell the outside factor is absent, as encoded by
$\mathcal J^{\phi^3}[n,1]=1$. Thus theorem~\ref{thm:rootedtree}
proves the general cubic formula
\eqref{eq:rooted}, including the skinny and non-skinny six-point examples.
For a four-vertex chain, for instance,
\begin{equation}\label{eq:treeP4example}
 \psi_{1-2-3-4}
 =\frac{c_{13}+c_{14}}{s_{12}s_{234}}\psi_{2-3-4}
  +\frac{c_{13}}{s_{12}s_{23}s_{123}}.
\end{equation}
There is no non-skinny rooted zero on $P_4$, since deleting an internal
vertex leaves sides of sizes one and two. The first non-skinny example is
$P_5$, which is precisely the $A_6$ formula~\eqref{eq:HF} under
eq.~\eqref{eq:pathisamplitude}. 
\paragraph{The $222$ star.}
Let $G_{222}$ have centre $3$ and edges $13,23,34$, with auxiliary
label $5$. Take $r=3$, $A=\{1\}$ and $B=\{2,4\}$.
The three choices of $J$ are $\{2\}$, $\{4\}$ and $\{2,4\}$.
Equation~\eqref{eq:treehosfs} gives
\begin{equation}\label{eq:tree222hosf}
 \psi_{222}
 =\frac{c_{12}+c_{14}}{s_{13}s_{234}}\psi_{2-3-4}
  +\frac{c_{12}}{s_{13}s_{23}s_{123}}
  +\frac{c_{14}}{s_{13}s_{34}s_{134}},
 \qquad \psi_{2-3-4}=\frac1{s_{23}}+\frac1{s_{34}}.
\end{equation}
The omitted two-vertex factors equal one. The rectangle relations are
\begin{equation}\label{eq:222blob}
 c_{12}=s_{13}+s_{23}-s_{123},\qquad
 c_{14}=s_{13}+s_{34}-s_{134},\qquad
 c_{12}+c_{14}=s_{13}+s_{234}.
\end{equation}
The two proper cells generate the negative monomials
$-1/(s_{13}s_{23})$ and $-1/(s_{13}s_{34})$.
They cancel the identical positive monomials from the full cell.
The survivors are the six maximal tubings:
\begin{align}\label{eq:222function}
 \psi_{222}={}&\frac1{s_{13}s_{123}}+\frac1{s_{23}s_{123}}
 +\frac1{s_{13}s_{134}}+\frac1{s_{34}s_{134}}\notag\\
 &+\frac1{s_{23}s_{234}}+\frac1{s_{34}s_{234}}.
\end{align}
The zero $c_{12}=c_{14}=0$ is regular: taking
$s_{13}=a$, $s_{23}=b$ and $s_{34}=c$ gives
$s_{123}=a+b$, $s_{134}=a+c$ and $s_{234}=-a$,
so every channel is nonzero at generic $a,b,c$. Permuting its leaves gives the other two
branch-isolating blob zeros. With the auxiliary label, this is a
five-point kinematic object, although its channel structure differs
from that of a single ordered $A_5$.

\paragraph{The $223$ star and a cancellation between neighbouring cells.}
Now take edges $13,23,34,45$ and auxiliary label $6$.
Keep $r=3$, $A=\{1\}$ and $B=\{2,4,5\}$.
The allowed right channels are $23,34,234,345,2345$, giving
\begin{align}\label{eq:tree223hosf}
 \psi_{223}={}&
 \frac{c_{12}+c_{14}+c_{15}}{s_{13}s_{2345}}\psi_{2-3-4-5}
 +\frac{c_{12}}{s_{13}s_{23}s_{123}}
       \psi_{[123]-4-5}^{\langle123\rangle}\notag\\
 &+\frac{c_{12}+c_{14}}{s_{13}s_{234}s_{1234}}\psi_{2-3-4}
 +\frac{c_{14}}{s_{13}s_{34}s_{134}}
       \psi_{2-[134]-5}^{\langle134\rangle}\notag\\
 &+\frac{c_{14}+c_{15}}{s_{13}s_{345}s_{1345}}\psi_{3-4-5}.
\end{align}
For example,
\begin{equation}\label{eq:tree223quotients}
 \psi_{[123]-4-5}^{\langle123\rangle}
 =\frac1{s_{1234}}+\frac1{s_{45}},\qquad
 \psi_{2-[134]-5}^{\langle134\rangle}
 =\frac1{s_{1234}}+\frac1{s_{1345}}.
\end{equation}
All numerators are converted to $s$ variables by
eq.~\eqref{eq:blobnumerators}. The four-vertex chain factor, corresponding to five-point kinematics, is
\begin{equation}
 \psi_{2-3-4-5}
 =\frac1{s_{23}s_{45}}+\frac1{s_{23}s_{234}}
  +\frac1{s_{34}s_{234}}+\frac1{s_{34}s_{345}}
  +\frac1{s_{45}s_{345}}.
\end{equation}
To illustrate Part I of the proof in appendix~\ref{sec:treeproof}, choose $1/s_{1234}$ from
$\psi_{[123]-4-5}^{\langle123\rangle}$. The negative kernel term at
$(L,R)=(13,23)$ is
$-1/(s_{13}s_{23}s_{1234})$. The sibling of $[123]$ is $\{4\}$.
Enlarging the right channel to $234$ produces the positive kernel
$1/(s_{13}s_{1234})$ multiplied by $1/s_{23}$ from $\psi_{2-3-4}$.
The two contributions cancel. Choosing instead $1/s_{45}$ makes the
sibling $\{4,5\}$, so the partner lies in the full cell.
All eight negative occurrences cancel in this manner, leaving the
eighteen maximal tubings, once each.

Appendix~\ref{app:treecatalogue} gives an example with two
adjacent branching vertices.




\section{New zeros from exact splitting}
\label{sec:newzeros}

The exact splitting formulas hold with all conditions of the generating zero
relaxed. Their complete lower factors can therefore expose
different zeros of the full amplitude. If a specified lower
amplitude vanishes, every term containing it vanishes without
imposing that term's numerator
condition. We impose numerator conditions only on the remaining
terms.

We give explicit conditions for five families of new zeros.
Each family is accompanied by a diagram of the exact splitting terms it
removes and an example at its lowest multiplicity.  We retain only regular maximal zeros, as defined
in subsection~\ref{subsec:zero-splitting-conventions}.

\subsection{New zeros at arbitrary multiplicity}
\label{sec:branches}

Here $v$ denotes the column endpoint called $k$ in the exact splitting formula;
Family 1 uses $k$ independently to label its lower rectangle.
We use the term matrix $\mathsf H^{(j)}$ introduced in
eq.~\eqref{eq:hosf-term-array}: a zero of a shared lower
amplitude removes its row or column, while a zero of an outside
factor can remove an individual entry. We impose numerator
conditions on the remaining entries. In the diagrams, red
circles or ellipses enclose the terms removed by lower-amplitude
zeros. Dots continue the labelled row and column ranges;
coincident endpoint labels denote a single row or column.

\subsubsection{Family 1: Two-rectangle zeros}
\label{sec:loweramplitudezeros}
\label{sec:tworectanglezeros}
Denote this family by $\mathcal R_{n;m,k}$. For $5\leq m<n$,
impose the ordinary rectangular zero of
$A[1,\ldots,m]$ selected by $X_{1k}$, with
$3\leq k\leq m-1$. Add
\begin{equation}\label{eq:rowconditions}
 \boxed{X_{am}+X_{m-1,v}-X_{av}=0,\qquad
 2\leq a\leq m-2,\quad m+1\leq v\leq n.}
\end{equation}
Evaluate the ordinary zero inside $A[1,\ldots,m]$, where
$(1,m)$ is a boundary edge. The full-amplitude channel $X_{1m}$
remains unconstrained.

Use the exact splitting formula with $j=m$. The zero of $A[1,\ldots,m]$ removes
the red first row in figure~\ref{fig:hosf-row-zero};
eq.~\eqref{eq:rowconditions} removes all other rows.
The rank is
\begin{equation}\label{eq:rectcorerank}
 (k-2)(m-k)+(m-3)(n-m).
\end{equation}
The first example is at six points. The name refers to the two
rectangles occupied by these conditions in the kinematic mesh,
shown in panel~(b) of figure~\ref{fig:mesh-zero-families}.

\begin{figure}[H]
\centering
\begin{tikzpicture}[x=0.9cm,y=0.65cm,
 every node/.style={font=\small}]
 \node at (1.55,1.25) {$v=m+1,\ldots,n$};
 \node[anchor=east] at (-0.65,0.4) {$a=1$};
 \node[anchor=east] at (-0.65,-0.7) {$a=2$};
 \node[anchor=east] at (-0.65,-1.8) {$\vdots$};
 \node[anchor=east] at (-0.65,-2.9) {$a=m-2$};
 \draw (-0.45,0.95)--(-0.6,0.95)--(-0.6,-3.35)--(-0.45,-3.35);
 \draw (3.55,0.95)--(3.7,0.95)--(3.7,-3.35)--(3.55,-3.35);
 \foreach \y in {0.4,-0.7,-2.9} {
  \node at (0,\y) {$\bullet$};
  \node at (1.55,\y) {$\cdots$};
  \node at (3.1,\y) {$\bullet$};
 }
 \foreach \x in {0,1.55,3.1} \node at (\x,-1.8) {$\vdots$};
 \draw[red!75!black,thick] (1.55,0.4) ellipse (1.97 and 0.43);
 \draw[red!75!black,-{Stealth},thick] (4.6,0.4)--(3.55,0.4);
 \node[anchor=west,text=red!75!black,align=left] at (4.75,0.4)
  {Ordinary rectangular zero:\\[2pt]
   $A[1,\ldots,m]=0$};
 \node[anchor=west,align=left,text width=5.1cm] at (4.25,-1.95)
  {Every other entry vanishes when\\[2pt]
   $X_{am}+X_{m-1,v}-X_{av}=0$.};
\end{tikzpicture}
\caption{Family 1: two-rectangle zeros, using the exact splitting formula with
 $j=m$. The red ellipse encloses the first row, whose terms
 all contain $A[1,\ldots,m]$. Impose its ordinary rectangular
 zero selected by $X_{1k}$. In all other rows, set the displayed
 numerator to zero for
 $2\leq a\leq m-2$ and $m+1\leq v\leq n$.
 The dots represent any number of intervening entries; the
 array has only one column when $n=m+1$.}
\label{fig:hosf-row-zero}
\end{figure}

\paragraph{Six-point example.}
Take $n=6$, $m=5$ and $k=4$, and call this zero $Z_I$.
Its conditions are
\begin{equation}\label{eq:ZI}
 \begin{gathered}
 \underbrace{X_{14}+X_{35}=X_{24}+X_{35}-X_{25}=0}
 _{A[1,2,3,4,5]=0},\\[4pt]
 X_{25}+X_{46}-X_{26}=X_{35}+X_{46}-X_{36}=0.
 \end{gathered}
\end{equation}

\subsubsection{Family 2: Three-pair zeros}
\label{sec:beyond}\label{sec:atomic}
For $n\geq6$, impose three pair sums on the channels of
$A[1,\ldots,6]$, followed by the displayed conditions on the
remaining channels:
\begin{equation}\label{eq:pairfamilyzero}
 \boxed{\begin{gathered}
 X_{13}+X_{26}=X_{24}+X_{35}=X_{46}+X_{15}=0,\\
 X_{a6}+X_{5v}-X_{av}=0,\qquad
 a=2,3,4,\quad7\leq v\leq n.
 \end{gathered}}
\end{equation}
We denote this family by $\mathcal P_n$. Its rank is
$3n-15$ for $n\geq6$, with no added equations at $n=6$.
The added conditions form the mesh strip
\[
 c_{uw}=0,\qquad u=2,3,4,\quad6\leq w\leq n-1.
\]
For $n>6$, use the $j=6$ exact splitting formula shown in
figure~\ref{fig:hosf-three-pair-zero}. The three pair sums set
$A[1,\ldots,6]=0$, removing the first row. The second line
of eq.~\eqref{eq:pairfamilyzero} removes the other three rows.

\begin{figure}[H]
\centering
\begin{tikzpicture}[x=0.9cm,y=0.65cm,
 every node/.style={font=\small}]
 \node at (1.55,1.35) {$v=7,\ldots,n$};
 \node[anchor=east] at (-0.65,0.4) {$a=1$};
 \node[anchor=east] at (-0.65,-0.8) {$a=2$};
 \node[anchor=east] at (-0.65,-2.0) {$a=3$};
 \node[anchor=east] at (-0.65,-3.2) {$a=4$};
 \draw (-0.45,0.95)--(-0.6,0.95)--(-0.6,-3.65)--(-0.45,-3.65);
 \draw (3.55,0.95)--(3.7,0.95)--(3.7,-3.65)--(3.55,-3.65);
 \foreach \y in {0.4,-0.8,-2.0,-3.2} {
  \node at (0,\y) {$\bullet$};
  \node at (1.55,\y) {$\cdots$};
  \node at (3.1,\y) {$\bullet$};
 }
 \draw[red!75!black,thick] (1.55,0.4) ellipse (1.97 and 0.43);
 \draw[red!75!black,-{Stealth},thick] (4.6,0.4)--(3.55,0.4);
 \node[anchor=west,text=red!75!black,align=left] at (4.75,0.4)
  {$X_{13}+X_{26}=0$\\[2pt]
   $X_{24}+X_{35}=0$\\[2pt]
   $X_{46}+X_{15}=0$\\[3pt]
   $\Longrightarrow\ A[1,\ldots,6]=0$};
 \node[anchor=west,align=left,text width=5.1cm] at (4.25,-2.9)
  {In the other three rows:\\[2pt]
   $X_{a6}+X_{5v}-X_{av}=0$,\\[2pt]
   $a=2,3,4$,\quad $7\leq v\leq n$.};
\end{tikzpicture}
\caption{Family 2: three-pair zeros at $n>6$. The $j=6$
 exact splitting formula has four rows and $n-6$ columns. The three pair sums
 remove the circled first row through its common factor
 $A[1,\ldots,6]$. The displayed numerator conditions remove
 every other entry. At $n=7$ there is only one column. The
 six-point starting zero is shown separately below.}
\label{fig:hosf-three-pair-zero}
\end{figure}

\paragraph{Six-point example.}
At $n=6$, write $Z_P=\mathcal P_6$. The conditions reduce to
\begin{equation}\label{eq:ZP}
 X_{15}+X_{46}=X_{24}+X_{35}=X_{13}+X_{26}=0.
\end{equation}
In the skinny exact splitting formula~\eqref{eq:HL}, the last pair kills the
boundary numerator, while the first two pairs set the four-point
factors in the other two terms to zero.

\subsubsection{Family 3: Balanced-column zeros}
\label{sec:coverfamilies}

Denote this family by $\mathcal C_{n,j}$. Choose the column
$v=j+2$ of the exact splitting array. Its terms all contain
$A[j-1,j,j+1,j+2]$. Set this four-point amplitude to zero and
impose the numerators outside the column. For $3\leq j\leq n-3$,
the conditions are
\begin{equation}\label{eq:columnfamily}
 \boxed{\begin{gathered}
 X_{j-1,j+1}+X_{j,j+2}=0,\\
 X_{aj}+X_{j-1,v}-X_{av}=0,\quad
 1\leq a\leq j-2,\quad j+1\leq v\leq n,
 \quad v\neq j+2.
 \end{gathered}}
\end{equation}
The rank is $(j-2)(n-j-1)+1$, and the first member has $n=6$.
The selected column is outlined in figure~\ref{fig:hosf-column-zero}.
The name refers to the two adjacent mesh columns whose entries
cancel row by row; see panel~(d) of
figure~\ref{fig:mesh-zero-families}.

\begin{figure}[H]
\centering
\begin{tikzpicture}[x=0.85cm,y=0.65cm,
 every node/.style={font=\small}]
 \node at (0,1.05) {$j+1$};
 \node at (1.3,1.05) {$j+2$};
 \node at (2.6,1.05) {$\cdots$};
 \node at (3.9,1.05) {$n$};
 \node at (1.95,1.8) {column $v$};
 \node[anchor=east] at (-0.7,0) {$a=1$};
 \node[anchor=east] at (-0.7,-1.1) {$\vdots$};
 \node[anchor=east] at (-0.7,-2.2) {$a=j-2$};
 \draw (-0.45,0.5)--(-0.6,0.5)--(-0.6,-2.7)--(-0.45,-2.7);
 \draw (4.35,0.5)--(4.5,0.5)--(4.5,-2.7)--(4.35,-2.7);
 \foreach \y in {0,-2.2} {
  \foreach \x in {0,1.3,3.9} \node at (\x,\y) {$\bullet$};
  \node at (2.6,\y) {$\cdots$};
 }
 \foreach \x in {0,1.3,3.9} \node at (\x,-1.1) {$\vdots$};
 \draw[red!75!black,thick] (1.3,-1.1) ellipse (0.48 and 1.57);
 \draw[red!75!black,-{Stealth},thick]
  (4.9,0.65)--(1.3,0.65)--(1.3,0.43);
 \node[anchor=west,text=red!75!black,align=left] at (5.05,0.65)
  {$X_{j-1,j+1}+X_{j,j+2}=0$};
 \node[anchor=west,align=left,text width=5.3cm] at (5.05,-1.3)
  {Every other entry vanishes when\\[2pt]
   $X_{aj}+X_{j-1,v}-X_{av}=0$,\\[2pt]
   with $v\neq j+2$.};
\end{tikzpicture}
\caption{Family 3: balanced-column zeros. Every term in the
 red column contains $A[j-1,j,j+1,j+2]$, which vanishes on the
 displayed pair sum. The numerators are set to zero in every
 other column. At $j=3$, the array has only one row.}
\label{fig:hosf-column-zero}
\end{figure}

\paragraph{Six-point example.}
For $n=6$ and $j=3$, write $Z_H=\mathcal C_{6,3}$. The conditions are
\begin{equation}\label{eq:nlsm-ZH}
 X_{24}+X_{35}=0,\qquad
 X_{13}+X_{24}-X_{14}=0,\qquad X_{13}+X_{26}=0.
\end{equation}

\subsubsection{Family 4: Row--column zeros}
\label{sec:rowcolumnzeros}

Denote this family by $\mathcal X_{n,j}$. Choose both the row
$a=j-3$ and the column $v=j+2$.
They contain the four-point amplitudes $A[j-3,j-2,j-1,j]$
and $A[j-1,j,j+1,j+2]$, respectively. Setting both amplitudes
to zero leaves numerator conditions only outside their union:
\begin{equation}\label{eq:crossfamily}
 \boxed{\begin{gathered}
 X_{j-3,j-1}+X_{j-2,j}=0,\qquad
 X_{j-1,j+1}+X_{j,j+2}=0,\\
 X_{aj}+X_{j-1,v}-X_{av}=0,\quad
 1\leq a\leq j-2,\quad j+1\leq v\leq n,\\
 a\neq j-3,\qquad v\neq j+2.
 \end{gathered}}
\end{equation}
For $5\leq j\leq n-3$, their rank is
$(j-3)(n-j-1)+2$. The intersection of the selected row and
column already vanishes through either lower amplitude.
The first member has $n=8$. Figure~\ref{fig:hosf-row-column-zero}
shows the row and column removed by the two pair sums.

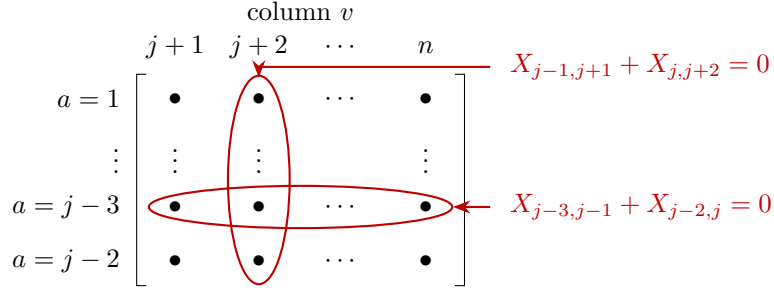
\begin{figure}[H]
\centering
\begin{tikzpicture}[x=0.85cm,y=0.65cm,
 every node/.style={font=\small}]
 \node at (0,1.05) {$j+1$};
 \node at (1.3,1.05) {$j+2$};
 \node at (2.6,1.05) {$\cdots$};
 \node at (3.9,1.05) {$n$};
 \node at (1.95,1.8) {column $v$};
 \node[anchor=east] at (-0.7,0) {$a=1$};
 \node[anchor=east] at (-0.7,-1.1) {$\vdots$};
 \node[anchor=east] at (-0.7,-2.2) {$a=j-3$};
 \node[anchor=east] at (-0.7,-3.3) {$a=j-2$};
 \draw (-0.45,0.5)--(-0.6,0.5)--(-0.6,-3.8)--(-0.45,-3.8);
 \draw (4.35,0.5)--(4.5,0.5)--(4.5,-3.8)--(4.35,-3.8);
 \foreach \y in {0,-2.2,-3.3} {
  \foreach \x in {0,1.3,3.9} \node at (\x,\y) {$\bullet$};
  \node at (2.6,\y) {$\cdots$};
 }
 \foreach \x in {0,1.3,3.9} \node at (\x,-1.1) {$\vdots$};
 \draw[red!75!black,thick] (1.3,-1.65) ellipse (0.48 and 2.12);
 \draw[red!75!black,thick] (1.95,-2.2) ellipse (2.36 and 0.43);
 \draw[red!75!black,-{Stealth},thick]
  (4.9,0.65)--(1.3,0.65)--(1.3,0.43);
 \node[anchor=west,text=red!75!black] at (5.05,0.65)
  {$X_{j-1,j+1}+X_{j,j+2}=0$};
 \draw[red!75!black,-{Stealth},thick] (4.9,-2.2)--(4.32,-2.2);
 \node[anchor=west,text=red!75!black] at (5.05,-2.2)
  {$X_{j-3,j-1}+X_{j-2,j}=0$};
 \node[anchor=north west,align=left] at (-0.6,-4.35)
  {Outside the red row and column:\\[2pt]
   $X_{aj}+X_{j-1,v}-X_{av}=0$,\quad
   $a\neq j-3$,\quad $v\neq j+2$.};
\end{tikzpicture}
\caption{Family 4: row--column zeros. The upper pair sum
 sets $A[j-1,j,j+1,j+2]=0$ and removes the red column. The lower
 pair sum sets $A[j-3,j-2,j-1,j]=0$ and removes the red row.
 Outside these outlines, impose
 $X_{aj}+X_{j-1,v}-X_{av}=0$.
 No additional condition is required at their intersection.}
\label{fig:hosf-row-column-zero}
\end{figure}

\paragraph{Eight-point example.}

At $n=8$, choose $j=5$. The row--column construction gives
six independent conditions:
\begin{equation}\label{eq:crossexample}
 \begin{gathered}
 X_{24}+X_{35}=X_{46}+X_{57}=0,\\
 X_{15}+X_{46}-X_{16}=X_{15}+X_{48}=0,\\
 X_{35}+X_{46}-X_{36}=X_{35}+X_{48}-X_{38}=0.
 \end{gathered}
\end{equation}

\subsubsection{Family 5: Opposite-end zeros}
\label{sec:oppositeendzeros}

Denote this family by $\mathcal D_n$. In the skinny exact splitting formula with
$j=3$, the term with $v=5$ contains
$A[2,3,4,5]$ as its middle factor, while the term with $v=n-2$
contains $A[n-2,n-1,n,1]$ as its last factor. Set these
two amplitudes to zero and impose the numerators of the other
terms. For $n\geq8$, this gives
\begin{equation}\label{eq:twoendfamily}
 \boxed{\begin{gathered}
 X_{24}+X_{35}=0,\qquad X_{n-2,n}+X_{1,n-1}=0,\\
 X_{13}+X_{2v}-X_{1v}=0,\quad
 4\leq v\leq n,\quad v\notin\{5,n-2\}.
 \end{gathered}}
\end{equation}
The rank is $n-3$. The two vanishing four-point amplitudes sit
at opposite ends of the polygon, and their separation grows
with $n$. Figure~\ref{fig:hosf-opposite-end-zero} marks the two selected terms.

\begin{figure}[H]
\centering
\begin{tikzpicture}[x=1.15cm,y=0.65cm,
 every node/.style={font=\small}]
 \node[anchor=east] at (-0.5,0.9) {$v=$};
 \node at (0,0.9) {$4$};
 \node at (1.15,0.9) {$5$};
 \node at (2.3,0.9) {$\cdots$};
 \node at (3.45,0.9) {$n-2$};
 \node at (4.6,0.9) {$n-1$};
 \node at (5.75,0.9) {$n$};
 \node[anchor=east] at (-0.5,0) {$a=1$};
 \draw (-0.23,0.45)--(-0.35,0.45)--(-0.35,-0.45)--(-0.23,-0.45);
 \draw (5.98,0.45)--(6.1,0.45)--(6.1,-0.45)--(5.98,-0.45);
 \foreach \x in {0,1.15,3.45,4.6,5.75} \node at (\x,0) {$\bullet$};
 \node at (2.3,0) {$\cdots$};
 \draw[red!75!black,thick] (1.15,0) ellipse (0.33 and 0.4);
 \draw[red!75!black,thick] (3.45,0) ellipse (0.33 and 0.4);
 \draw[red!75!black,-{Stealth},thick] (1.15,-1.0)--(1.15,-0.43);
 \draw[red!75!black,-{Stealth},thick] (4.2,-2.35)--(3.45,-2.35)--(3.45,-0.43);
 \node[anchor=north,text=red!75!black] at (1.15,-1.15)
  {$X_{24}+X_{35}=0$};
 \node[anchor=west,text=red!75!black] at (4.25,-2.35)
  {$X_{n-2,n}+X_{1,n-1}=0$};
 \node[anchor=north west,align=left] at (-0.35,-3.2)
  {At every uncircled entry:\\[2pt]
   $X_{13}+X_{2v}-X_{1v}=0$,\quad
   $v\notin\{5,n-2\}$.};
\end{tikzpicture}
\caption{Family 5: opposite-end zeros, using the one-row exact splitting formula
 with $j=3$. The term at $v=5$ vanishes through $A[2,3,4,5]$;
 the term at $v=n-2$ vanishes through $A[n-2,n-1,n,1]$.
 For all remaining values of $v$, impose
 $X_{13}+X_{2v}-X_{1v}=0$. These two entries are distinct for
 $n\geq8$.}
\label{fig:hosf-opposite-end-zero}
\end{figure}

At six points the analogous choice gives the three-pair zero.
At seven points the two selected terms coincide, so either
four-point zero removes the term and imposing both adds a
further restriction of either existing zero. The range
$n\geq8$ keeps the two selected terms distinct.

\paragraph{Eight-point example.}
At $n=8$, the conditions are
\begin{equation}\label{eq:opposite-end-eight}
 \begin{gathered}
 X_{24}+X_{35}=0,\qquad X_{68}+X_{17}=0,\\
 X_{13}+X_{24}-X_{14}=0,\qquad
 X_{13}+X_{27}-X_{17}=0,\qquad X_{13}+X_{28}=0.
 \end{gathered}
\end{equation}

In the stated ranges, all five constructions give regular maximal
zeros. The exact splitting formula makes their vanishing manifest, while the free-channel
parametrizations in appendix~\ref{cubic:sec:regularity} show that they
avoid propagator poles. For maximality, consider a small relaxation of
the defining equations. At the lowest multiplicities, the coefficients
of these relaxations are independent functions of the free channels.
Taking residues on a free four-point subpolygon reduces a larger member
to a smaller one; complementary choices of subpolygon retain all the
defining equations. Thus no nontrivial linear combination of the
conditions can be relaxed while keeping a regular zero. The families
are distinct from ordinary rectangular zeros,
as explained in appendix~\ref{cubic:sec:newness}.

\paragraph{New zeros on the kinematic mesh}
Figure~\ref{fig:mesh-zero-families} compares the five families
with an ordinary rectangular zero. The new patterns retain
regions where individual mesh variables vanish, but also allow
neighbouring diamonds to have nonzero values whose sums are
fixed. The channel equations above specify the conditions for each
family; the mesh displays their geometric pattern.
\begin{figure}[H]
\centering
\begingroup
\definecolor{meshshade}{HTML}{DCEBF7}
\definecolor{meshedge}{HTML}{65849F}
\definecolor{meshink}{HTML}{20303D}
\definecolor{meshlight}{HTML}{B5A28C}
\definecolor{meshextension}{HTML}{D7EFDF}
\definecolor{meshextensionedge}{HTML}{548367}
\definecolor{meshblock}{HTML}{FFE0AE}
\definecolor{meshblockedge}{HTML}{B47C34}
\definecolor{meshintersection}{HTML}{F1C785}
\resizebox{\textwidth}{!}{%
\begin{tikzpicture}[x=1pt,y=-1pt,every node/.style={inner sep=0pt,outer sep=0pt}]
\path[use as bounding box] (0,0) rectangle (450.000,412.500);
\path[fill=meshshade,draw=meshedge,line width=0.443pt] (5.795,10.909)--(10.568,6.136)--(15.341,10.909)--(10.568,15.682)--cycle;
\path[fill=meshextension,draw=meshextensionedge,line width=0.443pt] (15.341,10.909)--(20.114,6.136)--(24.886,10.909)--(20.114,15.682)--cycle;
\node[anchor=west,text=meshink,font=\fontsize{8.18}{9.82}\selectfont] at (31.023,10.909) {Blue / green: each $c=0$.};
\path[fill=meshblock,draw=meshblockedge,line width=0.750pt] (217.841,10.909)--(222.614,6.136)--(232.159,15.682)--(227.386,20.455)--cycle;
\draw[draw=meshlight,line width=0.273pt,dash pattern=on 1pt off 1pt] (222.614,15.682)--(227.386,10.909);
\node[anchor=west,text=meshink,font=\fontsize{8.18}{9.82}\selectfont] at (238.295,11.932) {Outline: sum equals the value shown.};
\node[anchor=center,text=meshink,font=\fontsize{8.86}{10.64}\selectfont\bfseries] at (75.000,35.455) {(a) Old rectangular zeros};
\node[anchor=center,text=meshink,font=\fontsize{7.16}{8.59}\selectfont] at (75.000,47.045) {$3\leq j\leq n-1$};
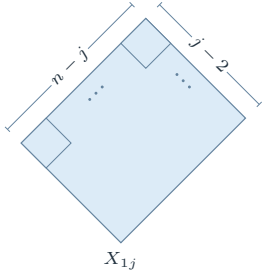
\path[fill=meshshade,draw=meshedge,line width=0.477pt] (70.057,153.580)--(30.511,114.034)--(79.943,64.602)--(119.489,104.148)--cycle;
\path[fill=meshshade,draw=meshedge,line width=0.409pt] (40.398,123.920)--(30.511,114.034)--(40.398,104.148)--(50.284,114.034)--cycle;
\path[fill=meshshade,draw=meshedge,line width=0.409pt] (79.943,84.375)--(70.057,74.489)--(79.943,64.602)--(89.830,74.489)--cycle;
\fill[meshedge] (57.699,96.733) circle (0.580pt);
\fill[meshedge] (60.170,94.261) circle (0.580pt);
\fill[meshedge] (62.642,91.790) circle (0.580pt);
\fill[meshedge] (97.244,91.790) circle (0.580pt);
\fill[meshedge] (94.773,89.318) circle (0.580pt);
\fill[meshedge] (92.301,86.847) circle (0.580pt);
\draw[draw=meshedge,line width=0.307pt] (25.057,108.580)--(74.489,59.148);
\draw[draw=meshedge,line width=0.307pt] (24.093,107.615)--(26.021,109.544);
\draw[draw=meshedge,line width=0.307pt] (73.524,58.183)--(75.453,60.112);
\node[anchor=center,text=meshink,font=\fontsize{6.82}{8.18}\selectfont,rotate=45.0,fill=white,inner sep=1pt] at (49.773,83.864) {$n-j$};
\draw[draw=meshedge,line width=0.307pt] (85.398,59.148)--(124.943,98.693);
\draw[draw=meshedge,line width=0.307pt] (86.362,58.183)--(84.433,60.112);
\draw[draw=meshedge,line width=0.307pt] (125.907,97.729)--(123.979,99.657);
\node[anchor=center,text=meshink,font=\fontsize{6.82}{8.18}\selectfont,rotate=-45.0,fill=white,inner sep=1pt] at (105.170,78.920) {$j-2$};
\node[anchor=center,text=meshink,font=\fontsize{6.82}{8.18}\selectfont] at (70.057,160.739) {$X_{1j}$};
\node[anchor=center,text=meshink,font=\fontsize{7.50}{9.00}\selectfont] at (75.000,185.795) {$c_{uw}=0$};
\node[anchor=center,text=meshink,font=\fontsize{7.50}{9.00}\selectfont] at (75.000,196.023) {$1\leq u\leq j-2$};
\node[anchor=center,text=meshink,font=\fontsize{7.50}{9.00}\selectfont] at (75.000,206.250) {$j\leq w\leq n-1$};
\node[anchor=center,text=meshink,font=\fontsize{8.86}{10.64}\selectfont\bfseries] at (225.000,35.455) {(b) 1. Two-rectangle};
\node[anchor=center,text=meshink,font=\fontsize{7.16}{8.59}\selectfont] at (225.000,47.045) {$5\leq m<n,\quad 3\leq k\leq m-1$};
\path[fill=meshshade,draw=meshedge,line width=0.477pt] (210.170,166.875)--(190.398,147.102)--(220.057,117.443)--(239.830,137.216)--cycle;
\path[fill=meshextension,draw=meshextensionedge,line width=0.477pt] (229.943,127.330)--(190.398,87.784)--(220.057,58.125)--(259.602,97.670)--cycle;
\path[fill=meshextension,draw=meshextensionedge,line width=0.409pt] (200.284,97.670)--(190.398,87.784)--(200.284,77.898)--(210.170,87.784)--cycle;
\path[fill=meshextension,draw=meshextensionedge,line width=0.409pt] (220.057,77.898)--(210.170,68.011)--(220.057,58.125)--(229.943,68.011)--cycle;
\fill[meshedge] (207.699,80.369) circle (0.580pt);
\fill[meshedge] (210.170,77.898) circle (0.580pt);
\fill[meshedge] (212.642,75.426) circle (0.580pt);
\fill[meshedge] (237.358,85.312) circle (0.580pt);
\fill[meshedge] (234.886,82.841) circle (0.580pt);
\fill[meshedge] (232.415,80.369) circle (0.580pt);
\path[fill=meshshade,draw=meshedge,line width=0.409pt] (200.284,156.989)--(190.398,147.102)--(200.284,137.216)--(210.170,147.102)--cycle;
\fill[meshedge] (207.699,139.688) circle (0.580pt);
\fill[meshedge] (210.170,137.216) circle (0.580pt);
\fill[meshedge] (212.642,134.744) circle (0.580pt);
\path[fill=meshblock,draw=meshblockedge,line width=0.852pt] (229.943,147.102)--(220.057,137.216)--(229.943,127.330)--(239.830,137.216)--cycle;
\node[anchor=center,text=meshink,font=\fontsize{6.14}{7.36}\selectfont] at (229.943,137.216) {$-X_{1m}$};
\draw[draw=meshedge,line width=0.307pt] (205.057,171.989)--(185.284,152.216);
\draw[draw=meshedge,line width=0.307pt] (204.093,172.953)--(206.021,171.024);
\draw[draw=meshedge,line width=0.307pt] (184.320,153.180)--(186.248,151.252);
\node[anchor=center,text=meshink,font=\fontsize{6.82}{8.18}\selectfont,rotate=-45.0,fill=white,inner sep=1pt] at (195.170,162.102) {$k-2$};
\draw[draw=meshedge,line width=0.307pt] (215.284,171.989)--(244.943,142.330);
\draw[draw=meshedge,line width=0.307pt] (214.320,171.024)--(216.248,172.953);
\draw[draw=meshedge,line width=0.307pt] (243.979,141.365)--(245.907,143.294);
\node[anchor=center,text=meshink,font=\fontsize{6.82}{8.18}\selectfont,rotate=45.0,fill=white,inner sep=1pt] at (230.114,157.159) {$m-k$};
\draw[draw=meshedge,line width=0.307pt] (185.284,82.670)--(214.943,53.011);
\draw[draw=meshedge,line width=0.307pt] (184.320,81.706)--(186.248,83.635);
\draw[draw=meshedge,line width=0.307pt] (213.979,52.047)--(215.907,53.976);
\node[anchor=center,text=meshink,font=\fontsize{6.82}{8.18}\selectfont,rotate=45.0,fill=white,inner sep=1pt] at (200.114,67.841) {$n-m$};
\draw[draw=meshedge,line width=0.307pt] (225.170,53.011)--(264.716,92.557);
\draw[draw=meshedge,line width=0.307pt] (226.135,52.047)--(224.206,53.976);
\draw[draw=meshedge,line width=0.307pt] (265.680,91.593)--(263.752,93.521);
\node[anchor=center,text=meshink,font=\fontsize{6.82}{8.18}\selectfont,rotate=-45.0,fill=white,inner sep=1pt] at (244.943,72.784) {$m-3$};
\node[anchor=center,text=meshink,font=\fontsize{7.50}{9.00}\selectfont] at (225.000,185.795) {$c_{1,m-1}=-X_{1m}$};
\node[anchor=center,text=meshink,font=\fontsize{7.50}{9.00}\selectfont] at (225.000,196.023) {First rectangle: $u=1,\ldots,k-2$};
\node[anchor=center,text=meshink,font=\fontsize{7.50}{9.00}\selectfont] at (225.000,206.250) {$w=k,\ldots,m-1$};
\node[anchor=center,text=meshink,font=\fontsize{8.86}{10.64}\selectfont\bfseries] at (375.000,35.455) {(c) 2. Three-pair};
\node[anchor=center,text=meshink,font=\fontsize{7.16}{8.59}\selectfont] at (375.000,47.045) {$n\geq6$};
\path[fill=meshextension,draw=meshextensionedge,line width=0.477pt] (369.716,132.443)--(338.011,100.739)--(380.284,58.466)--(411.989,90.170)--cycle;
\draw[draw=meshextensionedge,line width=0.341pt] (359.148,121.875)--(401.420,79.602);
\draw[draw=meshextensionedge,line width=0.341pt] (348.580,111.307)--(390.852,69.034);
\draw[draw=meshextensionedge,line width=0.341pt] (380.284,121.875)--(348.580,90.170);
\draw[draw=meshextensionedge,line width=0.341pt] (401.420,100.739)--(369.716,69.034);
\fill[meshedge] (382.926,108.665) circle (0.580pt);
\fill[meshedge] (385.568,106.023) circle (0.580pt);
\fill[meshedge] (388.210,103.381) circle (0.580pt);
\fill[meshedge] (372.358,98.097) circle (0.580pt);
\fill[meshedge] (375.000,95.455) circle (0.580pt);
\fill[meshedge] (377.642,92.812) circle (0.580pt);
\fill[meshedge] (361.790,87.528) circle (0.580pt);
\fill[meshedge] (364.432,84.886) circle (0.580pt);
\fill[meshedge] (367.074,82.244) circle (0.580pt);
\path[fill=meshblock,draw=meshblockedge,line width=0.000pt] (348.580,174.716)--(338.011,164.148)--(369.716,132.443)--(380.284,143.011)--cycle;
\path[fill=meshblock,draw=meshblockedge,line width=0.000pt] (369.716,153.580)--(338.011,121.875)--(348.580,111.307)--(380.284,143.011)--cycle;
\path[fill=meshintersection,draw=meshblockedge,line width=0.000pt] (369.716,153.580)--(359.148,143.011)--(369.716,132.443)--(380.284,143.011)--cycle;
\path[fill=none,draw=meshblockedge,line width=0.750pt] (348.580,174.716)--(338.011,164.148)--(369.716,132.443)--(380.284,143.011)--cycle;
\path[fill=none,draw=meshblockedge,line width=0.750pt] (369.716,153.580)--(338.011,121.875)--(348.580,111.307)--(380.284,143.011)--cycle;
\draw[draw=meshlight,line width=0.273pt,dash pattern=on 1pt off 1pt] (359.148,164.148)--(348.580,153.580);
\draw[draw=meshlight,line width=0.273pt,dash pattern=on 1pt off 1pt] (369.716,153.580)--(359.148,143.011);
\draw[draw=meshlight,line width=0.273pt,dash pattern=on 1pt off 1pt] (359.148,143.011)--(369.716,132.443);
\draw[draw=meshlight,line width=0.273pt,dash pattern=on 1pt off 1pt] (348.580,132.443)--(359.148,121.875);
\path[fill=meshblock,draw=meshblockedge,line width=0.750pt] (348.580,153.580)--(338.011,143.011)--(348.580,132.443)--(359.148,143.011)--cycle;
\node[anchor=center,text=meshink,font=\fontsize{6.48}{7.77}\selectfont,rotate=45] at (353.864,158.864) {$-X_{16}$};
\node[anchor=center,text=meshink,font=\fontsize{6.48}{7.77}\selectfont,rotate=-45] at (353.864,127.159) {$-X_{16}$};
\node[anchor=center,text=meshink,font=\fontsize{5.80}{6.95}\selectfont] at (348.580,143.011) {$-X_{25}$};
\draw[draw=meshedge,line width=0.307pt] (332.898,95.625)--(375.170,53.352);
\draw[draw=meshedge,line width=0.307pt] (331.933,94.661)--(333.862,96.589);
\draw[draw=meshedge,line width=0.307pt] (374.206,52.388)--(376.135,54.317);
\node[anchor=center,text=meshink,font=\fontsize{6.82}{8.18}\selectfont,rotate=45.0,fill=white,inner sep=1pt] at (354.034,74.489) {$n-6$};
\draw[draw=meshedge,line width=0.307pt] (385.398,53.352)--(417.102,85.057);
\draw[draw=meshedge,line width=0.307pt] (386.362,52.388)--(384.433,54.317);
\draw[draw=meshedge,line width=0.307pt] (418.067,84.093)--(416.138,86.021);
\node[anchor=center,text=meshink,font=\fontsize{6.82}{8.18}\selectfont,rotate=-45.0,fill=white,inner sep=1pt] at (401.250,69.205) {$3$};
\node[anchor=center,text=meshink,font=\fontsize{7.50}{9.00}\selectfont] at (375.000,185.795) {$c_{13}+c_{14}+c_{15}=-X_{16}$};
\node[anchor=center,text=meshink,font=\fontsize{7.50}{9.00}\selectfont] at (375.000,196.023) {$c_{15}+c_{25}+c_{35}=-X_{16}$};
\node[anchor=center,text=meshink,font=\fontsize{7.50}{9.00}\selectfont] at (375.000,206.250) {$c_{24}=-X_{25}$};
\node[anchor=center,text=meshink,font=\fontsize{8.86}{10.64}\selectfont\bfseries] at (75.000,221.591) {(d) 3. Balanced-column};
\node[anchor=center,text=meshink,font=\fontsize{7.16}{8.59}\selectfont] at (75.000,233.182) {$3\leq j\leq n-3$};
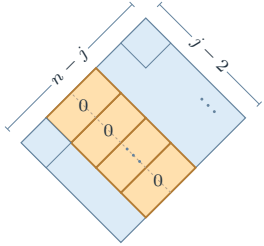
\path[fill=meshshade,draw=meshedge,line width=0.477pt] (70.057,335.966)--(30.511,296.420)--(79.943,246.989)--(119.489,286.534)--cycle;
\path[fill=meshblock,draw=meshblockedge,line width=0.750pt] (79.943,326.080)--(70.057,316.193)--(89.830,296.420)--(99.716,306.307)--cycle;
\draw[draw=meshlight,line width=0.273pt,dash pattern=on 1pt off 1pt] (89.830,316.193)--(79.943,306.307);
\node[anchor=center,text=meshink,font=\fontsize{7.50}{9.00}\selectfont] at (84.886,311.250) {$0$};
\path[fill=meshblock,draw=meshblockedge,line width=0.750pt] (70.057,316.193)--(60.170,306.307)--(79.943,286.534)--(89.830,296.420)--cycle;
\draw[draw=meshlight,line width=0.273pt,dash pattern=on 1pt off 1pt] (79.943,306.307)--(70.057,296.420);
\fill[meshedge] (77.472,303.835) circle (0.580pt);
\fill[meshedge] (75.000,301.364) circle (0.580pt);
\fill[meshedge] (72.528,298.892) circle (0.580pt);
\path[fill=meshblock,draw=meshblockedge,line width=0.750pt] (60.170,306.307)--(50.284,296.420)--(70.057,276.648)--(79.943,286.534)--cycle;
\draw[draw=meshlight,line width=0.273pt,dash pattern=on 1pt off 1pt] (70.057,296.420)--(60.170,286.534);
\node[anchor=center,text=meshink,font=\fontsize{7.50}{9.00}\selectfont] at (65.114,291.477) {$0$};
\path[fill=meshblock,draw=meshblockedge,line width=0.750pt] (50.284,296.420)--(40.398,286.534)--(60.170,266.761)--(70.057,276.648)--cycle;
\draw[draw=meshlight,line width=0.273pt,dash pattern=on 1pt off 1pt] (60.170,286.534)--(50.284,276.648);
\node[anchor=center,text=meshink,font=\fontsize{7.50}{9.00}\selectfont] at (55.227,281.591) {$0$};
\path[fill=meshshade,draw=meshedge,line width=0.409pt] (40.398,306.307)--(30.511,296.420)--(40.398,286.534)--(50.284,296.420)--cycle;
\path[fill=meshshade,draw=meshedge,line width=0.409pt] (79.943,266.761)--(70.057,256.875)--(79.943,246.989)--(89.830,256.875)--cycle;
\fill[meshedge] (107.131,284.062) circle (0.580pt);
\fill[meshedge] (104.659,281.591) circle (0.580pt);
\fill[meshedge] (102.187,279.119) circle (0.580pt);
\draw[draw=meshedge,line width=0.307pt] (25.057,290.966)--(74.489,241.534);
\draw[draw=meshedge,line width=0.307pt] (24.093,290.002)--(26.021,291.930);
\draw[draw=meshedge,line width=0.307pt] (73.524,240.570)--(75.453,242.498);
\node[anchor=center,text=meshink,font=\fontsize{6.82}{8.18}\selectfont,rotate=45.0,fill=white,inner sep=1pt] at (49.773,266.250) {$n-j$};
\draw[draw=meshedge,line width=0.307pt] (85.398,241.534)--(124.943,281.080);
\draw[draw=meshedge,line width=0.307pt] (86.362,240.570)--(84.433,242.498);
\draw[draw=meshedge,line width=0.307pt] (125.907,280.115)--(123.979,282.044);
\node[anchor=center,text=meshink,font=\fontsize{6.82}{8.18}\selectfont,rotate=-45.0,fill=white,inner sep=1pt] at (105.170,261.307) {$j-2$};
\node[anchor=center,text=meshink,font=\fontsize{7.50}{9.00}\selectfont] at (75.000,353.864) {Paired columns: $j+1,j+2$};
\node[anchor=center,text=meshink,font=\fontsize{7.50}{9.00}\selectfont] at (75.000,364.091) {$X_{j-1,j+1}+X_{j,j+2}=0$};
\node[anchor=center,text=meshink,font=\fontsize{8.86}{10.64}\selectfont\bfseries] at (225.000,221.591) {(e) 4. Row-column};
\node[anchor=center,text=meshink,font=\fontsize{7.16}{8.59}\selectfont] at (225.000,233.182) {$5\leq j\leq n-3$};
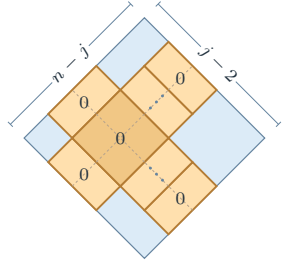
\path[fill=meshshade,draw=meshedge,line width=0.477pt] (225.000,342.273)--(177.273,294.545)--(225.000,246.818)--(272.727,294.545)--cycle;
\path[fill=meshblock,draw=meshblockedge,line width=0.750pt] (234.545,332.727)--(225.000,323.182)--(244.091,304.091)--(253.636,313.636)--cycle;
\draw[draw=meshlight,line width=0.273pt,dash pattern=on 1pt off 1pt] (244.091,323.182)--(234.545,313.636);
\node[anchor=center,text=meshink,font=\fontsize{7.50}{9.00}\selectfont] at (239.318,318.409) {$0$};
\path[fill=meshblock,draw=meshblockedge,line width=0.750pt] (225.000,323.182)--(215.455,313.636)--(234.545,294.545)--(244.091,304.091)--cycle;
\draw[draw=meshlight,line width=0.273pt,dash pattern=on 1pt off 1pt] (234.545,313.636)--(225.000,304.091);
\fill[meshedge] (232.159,311.250) circle (0.580pt);
\fill[meshedge] (229.773,308.864) circle (0.580pt);
\fill[meshedge] (227.386,306.477) circle (0.580pt);
\path[fill=meshblock,draw=meshblockedge,line width=0.750pt] (196.364,294.545)--(186.818,285.000)--(205.909,265.909)--(215.455,275.455)--cycle;
\draw[draw=meshlight,line width=0.273pt,dash pattern=on 1pt off 1pt] (205.909,285.000)--(196.364,275.455);
\node[anchor=center,text=meshink,font=\fontsize{7.50}{9.00}\selectfont] at (201.136,280.227) {$0$};
\path[fill=meshblock,draw=meshblockedge,line width=0.750pt] (205.909,323.182)--(186.818,304.091)--(196.364,294.545)--(215.455,313.636)--cycle;
\draw[draw=meshlight,line width=0.273pt,dash pattern=on 1pt off 1pt] (196.364,313.636)--(205.909,304.091);
\node[anchor=center,text=meshink,font=\fontsize{7.50}{9.00}\selectfont] at (201.136,308.864) {$0$};
\path[fill=meshblock,draw=meshblockedge,line width=0.750pt] (234.545,294.545)--(215.455,275.455)--(225.000,265.909)--(244.091,285.000)--cycle;
\draw[draw=meshlight,line width=0.273pt,dash pattern=on 1pt off 1pt] (225.000,285.000)--(234.545,275.455);
\fill[meshedge] (227.386,282.614) circle (0.580pt);
\fill[meshedge] (229.773,280.227) circle (0.580pt);
\fill[meshedge] (232.159,277.841) circle (0.580pt);
\path[fill=meshblock,draw=meshblockedge,line width=0.750pt] (244.091,285.000)--(225.000,265.909)--(234.545,256.364)--(253.636,275.455)--cycle;
\draw[draw=meshlight,line width=0.273pt,dash pattern=on 1pt off 1pt] (234.545,275.455)--(244.091,265.909);
\node[anchor=center,text=meshink,font=\fontsize{7.50}{9.00}\selectfont] at (239.318,270.682) {$0$};
\path[fill=meshintersection,draw=meshblockedge,line width=0.750pt] (215.455,313.636)--(196.364,294.545)--(215.455,275.455)--(234.545,294.545)--cycle;
\draw[draw=meshlight,line width=0.273pt,dash pattern=on 1pt off 1pt] (205.909,304.091)--(225.000,285.000);
\draw[draw=meshlight,line width=0.273pt,dash pattern=on 1pt off 1pt] (225.000,304.091)--(205.909,285.000);
\node[anchor=center,text=meshink,font=\fontsize{7.50}{9.00}\selectfont] at (215.455,294.545) {$0$};
\draw[draw=meshedge,line width=0.307pt] (171.818,289.091)--(219.545,241.364);
\draw[draw=meshedge,line width=0.307pt] (170.854,288.127)--(172.782,290.055);
\draw[draw=meshedge,line width=0.307pt] (218.581,240.399)--(220.510,242.328);
\node[anchor=center,text=meshink,font=\fontsize{6.82}{8.18}\selectfont,rotate=45.0,fill=white,inner sep=1pt] at (195.682,265.227) {$n-j$};
\draw[draw=meshedge,line width=0.307pt] (230.455,241.364)--(278.182,289.091);
\draw[draw=meshedge,line width=0.307pt] (231.419,240.399)--(229.490,242.328);
\draw[draw=meshedge,line width=0.307pt] (279.146,288.127)--(277.218,290.055);
\node[anchor=center,text=meshink,font=\fontsize{6.82}{8.18}\selectfont,rotate=-45.0,fill=white,inner sep=1pt] at (254.318,265.227) {$j-2$};
\node[anchor=center,text=meshink,font=\fontsize{7.50}{9.00}\selectfont] at (225.000,353.864) {Rows $j-4,j-3$; columns $j+1,j+2$};
\node[anchor=center,text=meshink,font=\fontsize{7.50}{9.00}\selectfont] at (225.000,364.091) {$X_{j-3,j-1}+X_{j-2,j}=0$};
\node[anchor=center,text=meshink,font=\fontsize{7.50}{9.00}\selectfont] at (225.000,374.318) {$X_{j-1,j+1}+X_{j,j+2}=0$};
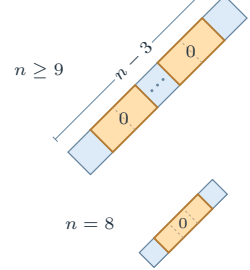
\node[anchor=center,text=meshink,font=\fontsize{8.86}{10.64}\selectfont\bfseries] at (375.000,221.591) {(f) 5. Opposite-end};
\node[anchor=center,text=meshink,font=\fontsize{7.16}{8.59}\selectfont] at (375.000,233.182) {$n\geq8$};
\path[fill=meshshade,draw=meshedge,line width=0.477pt] (348.409,312.955)--(339.545,304.091)--(401.591,242.045)--(410.455,250.909)--cycle;
\path[fill=meshshade,draw=meshedge,line width=0.409pt] (348.409,312.955)--(339.545,304.091)--(348.409,295.227)--(357.273,304.091)--cycle;
\path[fill=meshshade,draw=meshedge,line width=0.409pt] (375.000,286.364)--(366.136,277.500)--(375.000,268.636)--(383.864,277.500)--cycle;
\path[fill=meshshade,draw=meshedge,line width=0.409pt] (401.591,259.773)--(392.727,250.909)--(401.591,242.045)--(410.455,250.909)--cycle;
\path[fill=meshblock,draw=meshblockedge,line width=0.750pt] (357.273,304.091)--(348.409,295.227)--(366.136,277.500)--(375.000,286.364)--cycle;
\draw[draw=meshlight,line width=0.273pt,dash pattern=on 1pt off 1pt] (366.136,295.227)--(357.273,286.364);
\node[anchor=center,text=meshink,font=\fontsize{7.50}{9.00}\selectfont] at (361.705,290.795) {$0$};
\path[fill=meshblock,draw=meshblockedge,line width=0.750pt] (383.864,277.500)--(375.000,268.636)--(392.727,250.909)--(401.591,259.773)--cycle;
\draw[draw=meshlight,line width=0.273pt,dash pattern=on 1pt off 1pt] (392.727,268.636)--(383.864,259.773);
\node[anchor=center,text=meshink,font=\fontsize{7.50}{9.00}\selectfont] at (388.295,264.205) {$0$};
\fill[meshedge] (372.784,279.716) circle (0.580pt);
\fill[meshedge] (375.000,277.500) circle (0.580pt);
\fill[meshedge] (377.216,275.284) circle (0.580pt);
\node[anchor=east,text=meshink,font=\fontsize{6.82}{8.18}\selectfont] at (337.841,271.705) {$n\geq9$};
\draw[draw=meshedge,line width=0.307pt] (334.432,298.977)--(396.477,236.932);
\draw[draw=meshedge,line width=0.307pt] (333.468,298.013)--(335.396,299.942);
\draw[draw=meshedge,line width=0.307pt] (395.513,235.968)--(397.442,237.896);
\node[anchor=center,text=meshink,font=\fontsize{6.82}{8.18}\selectfont,rotate=45.0,fill=white,inner sep=1pt] at (365.455,267.955) {$n-3$};
\path[fill=meshshade,draw=meshedge,line width=0.477pt] (373.636,349.773)--(367.841,343.977)--(396.818,315.000)--(402.614,320.795)--cycle;
\path[fill=meshshade,draw=meshedge,line width=0.409pt] (373.636,349.773)--(367.841,343.977)--(373.636,338.182)--(379.432,343.977)--cycle;
\path[fill=meshshade,draw=meshedge,line width=0.409pt] (396.818,326.591)--(391.023,320.795)--(396.818,315.000)--(402.614,320.795)--cycle;
\path[fill=meshblock,draw=meshblockedge,line width=0.750pt] (379.432,343.977)--(373.636,338.182)--(391.023,320.795)--(396.818,326.591)--cycle;
\draw[draw=meshlight,line width=0.273pt,dash pattern=on 1pt off 1pt] (385.227,338.182)--(379.432,332.386);
\draw[draw=meshlight,line width=0.273pt,dash pattern=on 1pt off 1pt] (391.023,332.386)--(385.227,326.591);
\node[anchor=center,text=meshink,font=\fontsize{6.48}{7.77}\selectfont] at (385.227,332.386) {$0$};
\node[anchor=east,text=meshink,font=\fontsize{6.82}{8.18}\selectfont] at (357.955,332.386) {$n=8$};
\node[anchor=center,text=meshink,font=\fontsize{7.50}{9.00}\selectfont] at (375.000,364.773) {$X_{24}+X_{35}=0$};
\node[anchor=center,text=meshink,font=\fontsize{7.50}{9.00}\selectfont] at (375.000,375.000) {$X_{n-2,n}+X_{1,n-1}=0$};
\node[anchor=center,text=meshink,font=\fontsize{7.84}{9.41}\selectfont] at (225.000,395.455) {Side lengths count diamonds; dotted gaps continue the indicated pattern.};
\node[anchor=center,text=meshink,font=\fontsize{7.84}{9.41}\selectfont] at (225.000,406.364) {The extra channel equations beneath the panels are part of the zero conditions.};
\end{tikzpicture}
}
\endgroup
\caption{Mesh patterns for ordinary rectangular zeros and the
five new families. Blue and green diamonds have $c_{uw}=0$;
green distinguishes the added rectangle in panels (b) and (c).
Within each orange outline, the sum of the enclosed mesh
variables equals the displayed value. The individual variables
inside that outline need not vanish.
Side lengths count diamonds, and dotted gaps indicate that
the pattern continues. The additional channel equations beneath
the panels are part of the zero conditions.
In (c), the two orange strips overlap in one diamond and obey
separate sum conditions; the green rectangle is absent at
$n=6$, where $X_{16}=0$.
In (e), only the total of the four central diamonds is set to
zero. In (f), the two separate pairs at $n\geq9$ become a single
three-diamond block at $n=8$, as shown in the inset.}
\label{fig:mesh-zero-families}
\end{figure}

\paragraph{Further choices of lower zeros.}
The exact splitting formula is not limited to these five patterns. A column can
be removed using a zero of $A[j-1,\ldots,v]$ with any allowed
number of legs, and a row can use a zero of $A[a,\ldots,j]$.
Several rows and columns can be removed together. One then
sets the numerators of all remaining entries to zero, as in
the figures above. The five families displayed here have
all-multiplicity proofs that the resulting conditions avoid
poles and are not restrictions of larger zeros. The same
properties must be checked for additional choices.

\subsection{New zeros in the NLSM}
\label{subsec:nlsm-new-zeros}

The linear conditions of Families 2--5 also make the NLSM
amplitude vanish at every allowed even multiplicity. Family 1
has an all-multiplicity NLSM counterpart as well, but only for
the even values of its parameter $m$ specified below. In each
case we impose the same equations on the channels $X_{ij}$ of
the full amplitude $A_n^{\mathrm{NLSM}}$.

For the parameters $m,k$ of Family 1, the surviving choices are
listed in table~\ref{tab:nlsm-family-survival}; all listed
multiplicities are even.

\begin{table}[H]
\centering
\renewcommand{\arraystretch}{1.25}
\begin{tabular}{@{}lll@{}}
\toprule
Family & Multiplicity & Allowed construction parameters\\
\midrule
1: Two-rectangle & $n=8,10,\ldots$
 & $m=6,8,\ldots,n-2$,\quad $3\leq k\leq m-1$\\
2: Three-pair & $n=6,8,\ldots$ & No further choice\\
3: Balanced-column & $n=6,8,\ldots$ & $3\leq j\leq n-3$\\
4: Row--column & $n=8,10,\ldots$ & $5\leq j\leq n-3$\\
5: Opposite-end & $n=8,10,\ldots$ & No further choice\\
\bottomrule
\end{tabular}
\caption{The members of Families 1--5 whose cubic zero
 conditions also give NLSM zeros. Every choice in the stated
 ranges gives vanishing without forcing a propagator to zero.
 There is no additional parity restriction on $k$ or $j$.}
\label{tab:nlsm-family-survival}
\end{table}

\paragraph{Which members of Family 1 survive.}
At eight points, choose $m=6$ and any of $k=3,4,5$ in the
Family 1 prescription. These conditions give zeros of
$A_8^{\mathrm{NLSM}}$. The cubic choices $m=5$ and $m=7$
generically do not. At ten points, both $m=6$ and $m=8$ work.
More generally, any fixed even $m\geq6$ gives NLSM zeros at
\emph{every} even $n>m$, for every $k$ in the range shown.

For odd $m$, the same cubic conditions instead leave a
product of lower amplitudes that is generically nonzero;
its general form is given in eq.~\eqref{eq:nlsmoddcore}.
For example, the six-point two-rectangle zero
\eqref{eq:ZI}, which uses $(m,k)=(5,4)$, gives
\begin{equation}\label{eq:nlsmliftedstrict}
 A_6^{\mathrm{NLSM}}\big|_{\text{conditions }\eqref{eq:ZI}}
 =-\frac{(X_{13}+X_{24})(X_{15}+X_{46})}{X_{35}}.
\end{equation}
The cubic conditions do not set either numerator factor to
zero. Making this expression vanish would require an
additional condition, so the original cubic zero does not
survive as an NLSM zero.

\paragraph{The surviving families.}
For every multiplicity and parameter choice in
table~\ref{tab:nlsm-family-survival}, the corresponding
conditions define a regular maximal linear zero of the NLSM
amplitude, distinct from every ordinary rectangular zero.
Their ranks are the same as in cubic theory. In particular,
Family 2 obeys
\begin{equation}\label{eq:nlsmpairfamily}
 A_n^{\mathrm{NLSM}}\big|_{\mathcal P_n}=0,
 \qquad \operatorname{rank}\mathcal P_n=3n-15.
\end{equation}
To see maximality, relax the defining equations and examine the
first-order NLSM amplitude. At the lowest multiplicities these
variations have independent coefficients. Four-point factorization
then passes this independence to higher multiplicity: for Family 2 it
separates the three new pair conditions at each step, and for the other
families complementary adjacent-pair deletions retain every condition.
Hence no nontrivial linear combination of the conditions can be
relaxed while preserving the zero.
The survival conditions and the odd-$m$ restriction of
Family 1 are derived in appendix~\ref{app:nlsm-survival}.

\subsection{New non-blob zeros in cosmological wavefunctions}
\label{subsec:tree-new-zeros}

Choose a rooted branch partition $(r;A,B)$ of a branched tree $T$
and use the cells of eq.~\eqref{eq:rootedtree}, writing $L=L_I$
and $R=R_J$. Setting the complete lower function
$\psi_{T[\{r\}\cup A]}$ to zero removes every term with
$L=\{r\}\cup A$. The channel $s_{\{r\}\cup A}$ remains a
physical channel of the parent function.

Evaluate the lower zero on the inherited channels, then impose
the numerators of the remaining cells:
\begin{equation}\label{eq:wflowerreplacement}
 \boxed{\begin{gathered}
 \text{zero conditions for }\psi_{T[\{r\}\cup A]}
 \text{ at its inherited channel values},\\
 s_L+s_R-s_{L\cup R}=0,
 \qquad L\neq\{r\}\cup A.
 \end{gathered}}
\end{equation}
The remaining cells have $L\neq\{r\}\cup A$ and $L\cup R\subsetneq V$.
The conditions must be compatible with the full Mandelstam relations.
Regularity, maximality and containment in a blob zero are checked
separately.

We have not yet established an all-multiplicity family of
regular maximal non-blob zeros on branched trees. The examples
below give new zeros for three small trees. We allow linear
relations among their channels that make exact splitting terms vanish
separately or cancel in groups. Extending these examples to
arbitrary tree size remains an open question.

\subsubsection{A new non-blob zero of the $222$ star}

Take the $222$ star with centre $3$ and leaves $1,2,4$.
It has the zero
\begin{equation}\label{eq:newzeros-222}
 \boxed{\begin{gathered}
 s_{13}+s_{23}=0,\\
 s_{134}+s_{13}+s_{34}=0,\qquad
 s_{234}+s_{23}+s_{34}=0.
 \end{gathered}}
\end{equation}
These three independent conditions define a regular maximal
zero outside the blob family. In particular, they do not imply
any of the blob conditions obtained by separating one leaf
from the other two. This is a new non-blob zero already at
five-point kinematics. It follows by collecting terms of the
exact splitting formula and allowing a cancellation between two groups, rather
than making each original term vanish separately.
Its three-channel sums go beyond the four forms used for the
finite amplitude search in subsection~\ref{sec:zeroalphabet}.
Indeed, put $a=s_{13}$ and $c=s_{34}$. The conditions and mapped
kinematics give $s_{23}=-a$, $s_{123}=3c$,
$s_{134}=-a-c$ and $s_{234}=a-c$. Grouping the six terms
of eq.~\eqref{eq:222function} by their three-vertex channels then gives
$0-1/(ac)+1/(ac)=0$.

\subsubsection{A new non-blob zero of the $2222$ star}

Take the $2222$ star with centre $3$ and leaves $1,2,4,5$.
A zero obtained from its exact splitting formula is
\begin{equation}\label{eq:newzeros-2222}
 \boxed{\begin{gathered}
 s_{13}+s_{34}=0,\qquad s_{23}+s_{35}=0,\\
 s_{1234}+s_{1235}=0,\qquad
 s_{1235}+s_{1345}=0.
 \end{gathered}}
\end{equation}
The four conditions pair the two-vertex and four-vertex
channels. They define a regular maximal non-blob zero at
six-point kinematics. Here the conditions make the exact splitting formula
vanish term by term through its numerators and lower factors.
Use the partition $r=3$, $A=\{1,4\}$ and $B=\{2,5\}$.
Cells using all of either side contain a vanishing three-vertex
tree function. In each remaining cell the contracted tree has
three vertices, and its two channel inverses cancel by the
four-vertex pair sums; mapped kinematics also gives
$s_{1345}+s_{2345}=0$.

\subsubsection{A new non-blob zero of the $223$ tree}

Take the $223$ tree with edges $13$, $23$, $34$ and $45$.
Its exact splitting formula gives the following zero:
\begin{equation}\label{eq:223newzero}
 \boxed{\begin{gathered}
 s_{234}+s_{345}=0,\qquad s_{1234}+s_{45}=0,\\
 s_{23}+s_{1345}=0,\qquad
 s_{134}+s_{45}-s_{1345}=0.
 \end{gathered}}
\end{equation}
These four independent conditions define a regular maximal
non-blob zero at six-point kinematics.
With root $4$ and sides $\{1,2,3\}$ and $\{5\}$, the cells
with left channels $134$ and $1234$ vanish through their numerators.
The other two cells, with left channels $34$ and $234$, cancel:
writing $b=s_{23}$, $c=s_{34}$, $d=s_{45}$ and $g=s_{234}$,
their restrictions are $(c+d+g)(b+c)/(bc^2dg)$ and
$-(c+d+g)(b+c)/(bc^2dg)$, respectively. Here the mapped relations imply
$s_{2345}=-c$. For all three examples, generic points of the stated
loci keep every physical channel nonzero. The maximal-tubing sums also
have independent first-order responses to relaxing their three, four
and four defining conditions, respectively; hence these zeros are
maximal. None is contained in a blob zero: for every possible blob
partition at least one required cross condition remains free. The
partitions are listed in appendix~\ref{app:tree-zero-nonblob}.

\subsection{An alphabet for zero conditions}
\label{sec:zeroalphabet}
\label{subsec:alphabet-conventions}

The zero constructions above motivate a finite search using a
specified set of linear relations, which we call an \emph{alphabet
of zero conditions}. We focus on the ${\rm Tr}(\phi^3)$ amplitude for this search, but it is straightforward to generalize the search to NLSM amplitudes and the tree wavefunctions. We use the following four forms on a fixed path. Identify the channel $X_{u,v+1}$ with the interval $[u,v]$ in the fixed path $V=\{1,\ldots,n-1\}$, and write
\begin{equation}
 S_{[u,v]}=X_{u,v+1},\qquad
 S_\varnothing=S_{\{u\}}=S_V=0.
\end{equation}
For two distinct physical intervals $A$ and $B$, with
$2\leq |A|,|B|\leq n-2$, define
\begin{equation}\label{eq:alphabet}
 \begin{array}{ll}
 P(A,B)=S_A+S_B,&
 U(A,B)=S_A+S_B-S_{A\cup B},\\[2pt]
 I(A,B)=S_A+S_B-S_{A\cap B},&
 Q(A,B)=S_A+S_B-S_{A\cup B}-S_{A\cap B}.
 \end{array}
\end{equation}
The zero equations set the displayed forms to zero. Every $S$ is an
actual planar channel $X_{ij}$, or a boundary channel set to zero.  A union or intersection is admitted only when it is a connected interval or one of these boundary cases.  We remove identities and single-channel equations, and identify conditions which differ only by an overall nonzero constant.  Pair, union, intersection and union-plus-intersection conditions form the $P/U/I/Q$ alphabet $\mathcal L_n$.  Appendix~\ref{app:alphabet} proves closure under inheritance
when the local cut is compatible with the fixed parent path.

A zero is describable in this alphabet only when the letters that vanish on it span its complete equation space,
\begin{equation}\label{eq:alphabetspan}
 \Span(\mathcal L_n\cap W_Z)=W_Z.
\end{equation}
The path $V$ remains fixed. Cyclic rotations identify equivalent
zero loci; a class is retained only if at least one presentation
satisfies eq.~\eqref{eq:alphabetspan} using these four forms.

The search imposes independent subsets of the chosen relations
and checks the full amplitude. Every retained result is verified
by exact rational algebra, has no identically vanishing physical
channel, and passes the maximality test against arbitrary
linear equations with no constant term. For the color-ordered ${\rm Tr}(\phi^3)$ amplitude, cyclic rotations are identified;
reflections are retained when they are not cyclically equivalent. The counts
are shown in table~\ref{tab:zerocomparison}. 
\begin{table}[htbp]
\centering
\begin{tabular}{@{}cccc@{}}
\toprule
$n$ & Rank & Old hidden zeros & New zeros \\
\midrule
$6$ & $1$--$2$ & $0$ & $0$ \\
    & $3$       & $1$ & $2$ \\
    & $4$       & $1$ & $9$ \\
    & $5$       & $0$ & $5$ \\
    & $6$       & $0$ & $2$ \\
    & $7$       & $0$ & $0$ \\
    & $8$       & $0$ & $0$ \\
\addlinespace
$7$ & $1$--$3$ & $0$ & $0$ \\
    & $4$       & $1$ & $2$ \\
    & $5$       & $0$ & $14$ \\
    & $6$       & $1$ & $97$ \\
\bottomrule
\end{tabular}
\caption{Regular maximal linear zeros defined using the
fixed-path $P/U/I/Q$ alphabet, with cyclic rotations identified.
The enumeration is exhaustive through rank eight at six points
and through rank six at seven points.}
\label{tab:zerocomparison}
\end{table}
It remains to determine which of these zeros can be obtained
recursively from ordinary zeros and how the finite catalogue
extends to families at arbitrary multiplicity.

\section{Continuous families of zeros}
\label{sec:continuouszeros}

We now explore the possibility of zero loci whose conditions vary
continuously. The five-point
$\operatorname{Tr}(\phi^3)$ amplitude provides an example, which the
exact splitting formula then extends to six points.

\paragraph{A continuous family at five points.}
The five-point amplitude is
\begin{equation}
 A_5=
 \frac{1}{X_{13}X_{14}}
 +\frac{1}{X_{13}X_{35}}
 +\frac{1}{X_{14}X_{24}}
 +\frac{1}{X_{24}X_{25}}
 +\frac{1}{X_{25}X_{35}}.
\end{equation}
Let $\lambda$ be a dimensionless parameter, and consider the
following three linear conditions:
\begin{equation}\label{eq:continuous-five-zero}
 \boxed{
 \begin{gathered}
 X_{24}=-\lambda X_{13},\qquad X_{25}=X_{14},\\
 (2-\lambda)X_{35}=\lambda(X_{13}+X_{14}).
 \end{gathered}}
\end{equation}

For every fixed $\lambda\notin\{0,1,2\}$, these equations
define a regular maximal rank-three zero. Different values of
$\lambda$ give
different zeros, since they fix different ratios
$X_{24}/X_{13}$. For example, $\lambda=-1$ gives
\begin{equation}
 X_{24}=X_{13},\qquad X_{25}=X_{14},\qquad
 3X_{35}+X_{13}+X_{14}=0.
\end{equation}
Maximality can be checked directly from the five-point numerator:
its first-order responses to relaxing the three equations are
independent functions of the free channels $X_{13}$ and $X_{14}$
when $\lambda\notin\{0,1,2\}$. Thus no nontrivial linear combination
of the defining equations can be relaxed while preserving the zero.

\paragraph{Connections to ordinary hidden zeros.}
The exceptional parameter values explain how this family
relates to the ordinary zeros. At $\lambda=0$, the equations
force $X_{24}=X_{35}=0$, so the entire locus lies on poles.
The values $\lambda=1$ and $\lambda=2$ instead give
restrictions of ordinary hidden zeros.

At $\lambda=1$, the conditions become
\begin{equation}
 X_{13}+X_{24}=0,\qquad
 X_{35}=X_{13}+X_{25},\qquad X_{25}=X_{14}.
\end{equation}
The first two equations already define an ordinary
five-point zero. The third is an additional restriction
and does not define a new maximal zero.

At $\lambda=2$, the conditions become
\begin{equation}
 X_{24}=-2X_{13},\qquad X_{14}=X_{25}=-X_{13},
\end{equation}
with $X_{35}$ unrestricted. They lie inside the ordinary
zero
\begin{equation}
 X_{13}+X_{25}=0,\qquad
 X_{13}+X_{24}-X_{14}=0.
\end{equation}
The continuous family therefore approaches restrictions
of two different ordinary zeros. 

\paragraph{A continuous family at six points.}
The five-point family was found by a direct cancellation between
diagrams. It provides another application of the exact splitting formula: a new
type of lower-point zero, with continuously variable coefficients,
can be carried to higher multiplicity by the same construction
used for the discrete families. At six points, keep its three
conditions on the channels of $A[1,2,3,4,5]$ and add two numerator
conditions from the $j=5$ exact splitting formula:
\begin{equation}\label{eq:continuous-six-zero}
 \boxed{\begin{gathered}
 X_{24}=-\lambda X_{13},\qquad X_{25}=X_{14},\\
 (2-\lambda)X_{35}=\lambda(X_{13}+X_{14}),\\
 X_{26}=X_{25}+X_{46},\qquad
 X_{36}=X_{35}+X_{46}.
 \end{gathered}}
\end{equation}
The last two equations remove the terms in rows $a=2,3$
of that exact splitting formula. The remaining term is a product of complete
lower amplitudes:
\begin{equation}\label{eq:continuous-six-product}
 \left.A_6\right|_{\substack{X_{26}=X_{25}+X_{46}\\
                            X_{36}=X_{35}+X_{46}}}
 =A[1,2,3,4,5]A[1,4,5,6].
\end{equation}
The first three equations in
eq.~\eqref{eq:continuous-six-zero} put the five-point factor
on its continuous zero, so the full six-point amplitude
vanishes. The channel $X_{15}$ remains free: it is a boundary
edge of the five-point factor and does not enter its zero
conditions. In the mesh, the two added equations are simply
$c_{25}=c_{35}=0$.

For every fixed $\lambda\notin\{0,1,2\}$, these five
independent conditions define a regular maximal rank-five
zero. The channels $X_{13},X_{14},X_{15},X_{46}$ can be chosen
freely, and all propagators remain nonzero at generic points.
Varying $\lambda$ changes the ratio $X_{24}/X_{13}$ and hence
the zero itself. The free $X_{15}$ pole separates the three
five-point variations from the two new conditions; the pole at
$X_{46}=-X_{35}$ singles out one of these, and the other then remains.
The five responses
are therefore independent, so this lifted zero is also maximal.

The exceptional values relate this family to zeros encountered
earlier. At $\lambda=0$, the conditions force
$X_{24}=X_{35}=0$, placing the locus on propagator poles.
At $\lambda=1$, the rank-five locus lies inside an ordinary
rank-three skinny zero. At $\lambda=2$, it lies inside the
rank-four two-rectangle zero of Family~1 with $(m,k)=(5,3)$.
Thus the generic members are maximal, while these two special
members are further restrictions of larger zeros.

\section{Conclusions}
\label{sec: conclusion}

We have promoted hidden zeros to explicit splitting
representations at generic kinematics. For every rectangular
zero of the cubic theory, the full amplitude is a sum of
complete lower-amplitude products, with one linear prefactor
for each independent zero condition and only physical
propagator poles. This factorized organization remains valid
when all conditions are relaxed simultaneously. Its extensions
to the NLSM and mapped tree wavefunctions preserve the role
of the chosen zero in organizing the lower factors.

The lower factors also turn known zeros into a construction
of further zeros. In $\operatorname{Tr}(\phi^3)$ theory,
this yields five explicit families at arbitrary multiplicity.
We have determined which conditions also give zeros of NLSM
amplitudes and applied the same reasoning to branched tree
functions entering cosmological wavefunctions, where new
non-blob zeros already appear at five-point kinematics.
The finite searches complement these constructions with
further examples whose general patterns remain to be
understood. Extending these examples to arbitrary multiplicity
and establishing maximality of their extensions are natural
next steps.

An open question is which zeros can be reached from the
ordinary ones by successive splitting constructions. Within
a fixed alphabet, can every maximal zero defined by that
alphabet be obtained in this way? Can a splitting adapted to
a new zero reveal further zeros that the original exact splitting formulas do
not expose? 

We have also uncovered continuous families of zeros.
The five-point example shows that the coefficients of the defining linear
conditions can vary continuously while preserving vanishing.
Its extension to six points demonstrates that the exact splitting formula can also
carry these continuous families to higher multiplicity. We leave the origin and more general structure of these zeros for future work.

A separate direction is to extend exact splitting formulas to spinning particles.
Universal and related tree-amplitude expansions express gravity,
Einstein--Yang--Mills and Yang--Mills amplitudes in terms of
same-multiplicity amplitudes with some particles replaced by
lower-spin states
\cite{StiebergerTaylorEYM,FuDuHuangFengEYM,TengFengEYM,FengLiZhouEYM,DongHeHouUniversal,HuZhouYMExpansion}.
These expansions change particle content while keeping multiplicity
fixed; our exact splitting formulas lower multiplicity and make a chosen zero explicit
in each term. It would be interesting to test whether the two
organizations can be combined to give fully relaxed exact splitting formulas for
Yang--Mills, Einstein--Yang--Mills and gravity, with
polarization-dependent prefactors and possibly complete
mixed-theory lower-point amplitudes.

\vspace{0.5cm}
\paragraph{Acknowledgments} YL was supported in part by a Technion fellowship. LR is supported by the National Natural Science Foundation of China General Program No. 12475070 and the Beijing Natural Science Foundation International Scientist Project No. IS24014. We used GPT-5.6 and GPT-6 (OpenAI, 2026) as research and coding assistants for exploratory calculations, figure preparation, and manuscript editing; the authors are responsible for the scientific content.
\appendix
\begingroup
\small
\setlength{\abovedisplayskip}{4pt}
\setlength{\belowdisplayskip}{4pt}
\setlength{\abovedisplayshortskip}{2pt}
\setlength{\belowdisplayshortskip}{2pt}
\setlength{\textfloatsep}{10pt}
\setlength{\floatsep}{8pt}
\section{Proofs and supplementary formulas}
\label{app:proofs}

\subsection{Tree exact splitting formula: proof and example}
\label{app:hosfdetails}
\subsubsection{Cancellation and the maximal-tubing correspondence}
\label{sec:treeproof}

To prove eq.~\eqref{eq:rootedtree}, expand each complete lower
function into channel monomials. Given
$\tau\in\operatorname{MT}(G)$, adjoin all singletons and the full set:
\begin{equation}\label{eq:treehierarchy}
 \overline\tau=\tau\cup\{\{v\}:v\in V\}\cup\{V\}.
\end{equation}
For each nonsingleton set $S$ in this family, its maximal proper subsets
will be called its children. They are disjoint and cover $S$. There are
exactly two: if there were three or more, contract them inside $G[S]$.
Two of the resulting vertices are adjacent, and the union of their
preimages could be inserted as a new compatible proper channel. This
would contradict maximality. Thus a maximal tubing is a binary hierarchy
of connected sets. Conversely, every such hierarchy gives a maximal
tubing upon deleting the singleton and full sets. A hierarchy with $N$
singleton leaves has $N-1$ nonsingleton sets, so it has $N-2$ propagators.
The added singleton and full sets carry no propagators.

For a proper cell, abbreviate $L=L_I$, $R=R_J$, and $U=L\cup R$.
The elementary identity is
\begin{equation}\label{eq:treelocal}
 \frac{s_L+s_R-s_U}{s_Ls_Rs_U}
 =\frac1{s_Rs_U}+\frac1{s_Ls_U}-\frac1{s_Ls_R}.
\end{equation}
Choose independently one maximal tubing inside $L$, one inside $R$,
and one in $G/U$, denoted by $\tau_L,\tau_R,\tau_Q$.
Their common propagator product is
\begin{equation}\label{eq:treecommonweight}
 F=\prod_{S\in\tau_L}\frac1{s_S}
   \prod_{S\in\tau_R}\frac1{s_S}
   \prod_{Q\in\tau_Q}\frac1{s_{\pi_U^{-1}(Q)}}.
\end{equation}
This choice contributes $+F/(s_Rs_U)$, $+F/(s_Ls_U)$ and
$-F/(s_Ls_R)$. For the full cell there are only the two positive
contributions $F/s_R$ and $F/s_L$. We keep track of each contribution
together with its cell and its three chosen tubings, even if several
contributions have the same denominator.

\begin{proof}
\emph{Part I: all incompatible contributions cancel.}

Every negative contribution contains the overlapping channels $L,R$,
with $L\cap R=\{r\}$ and neither containing the other.
We pair it with one positive contribution having exactly the same
propagator labels.

In the binary hierarchy of $G/U$, the contracted leaf $[U]$ has a
unique parent. Let $K$ be the other child of that parent, pulled back to
the original vertices, and set
\begin{equation}\label{eq:treeenlargement}
 U'=U\cup K.
\end{equation}
The set $K$ is connected and avoids $r$, so it lies entirely in one
component of $G-r$. Suppose first that $K\subseteq A$. Then
$L'=L\cup K$ is connected, and $(L',R)$ is another admissible cell.
To construct its lower hierarchy inside $L'$, take $L$ and $K$ as its
two children and retain their descendants. Keep the hierarchy inside
$R$, and contract the old parent $U'$ in the outside hierarchy.
This gives the positive contribution with explicit channels $R,U'$.
If $U'=V$, use the positive full-cell contribution with explicit
channel $R$ instead.

The two contributions have opposite signs and identical denominators.
The channel $L$ moves from the explicit kernel into the lower function
on $L'$. The channel $U'$, when proper, moves from the outside hierarchy
into the explicit kernel. The subtree on $K$ and every other channel
are unchanged. When $U'=V$, neither contribution contains its propagator.
For $K\subseteq B$, the same construction enlarges $R$ and gives the
positive contribution with explicit channels $L,U'$.

The pairing is reversible. Consider a positive contribution with
explicit channels $R,U'$ and inspect the two children of the full set
$L'$ in its lower hierarchy. Denote the child containing $r$ by $L$
and the other child by $K$. If $L\ne\{r\}$, set $U=L\cup R$,
remove $K$ from the left factor, and replace the outside leaf $[U']$
by a parent with children $[U]$ and $K$. Its full preimage is
$U'=U\cup K$. Retaining the descendants of $K$ reconstructs the unique
negative contribution at $(L,R)$. The argument includes $U'=V$, for
which the outside hierarchy initially consists of a single leaf.

A positive contribution of this type is incompatible precisely when
$L\ne\{r\}$: in that case $L$ occurs inside the left factor and
overlaps the explicit channel $R$. If instead $L=\{r\}$, every channel
of that factor avoids $r$ and is disjoint from $R$ and its descendants.
The quotient channels either contain the full cell union or are
disjoint from it, so no other incompatibility can occur. The positive
contributions with explicit left channel are treated symmetrically.
Consequently the construction pairs every negative contribution with
exactly one incompatible positive contribution, exhausting both sets. The one-to-one correspondence between positive and negative incompatible contributions is illustrated in fig.~\ref{fig:part I}.
\begin{figure}[H]
\centering
\begin{tikzpicture}[x=0.94cm,y=0.94cm,
  explicit/.style={black,line width=1.1pt,rounded corners=5pt},
  lower/.style={gray,line width=.7pt,rounded corners=5pt},
  every node/.style={font=\small}]
\begin{scope}
  \node[anchor=south] at (2,1.58) {$-F/(s_Ls_R)$};
  \draw[lower] (-.35,.95)--(2.52,.95)--(2.52,-1.63)
    --(1.53,-1.63)--(1.53,-.49)--(-.35,-.49)--cycle;
  \draw[lower] (2.75,-.39) rectangle (4.25,.43);
  \draw[explicit] (.65,-.39) rectangle (2.33,.43);
  \draw[explicit] (1.69,-1.52) rectangle (2.29,.27);
  \node[gray] at (0,1.20) {$U'$};
  \node at (1.20,.67) {$L$};
  \node[right,inner sep=1pt] at (2.51,-.83) {$R$};
  \node[gray] at (3.5,.67) {$T$};
\end{scope}
\begin{scope}[xshift=5.75cm]
  \node[anchor=south] at (2,1.58) {$+F'/(s_Rs_{U'})$};
  \draw[explicit] (-.35,.95)--(2.52,.95)--(2.52,-1.63)
    --(1.53,-1.63)--(1.53,-.49)--(-.35,-.49)--cycle;
  \draw[lower] (2.75,-.39) rectangle (4.25,.43);
  \draw[lower] (.65,-.39) rectangle (2.33,.43);
  \draw[explicit] (1.69,-1.52) rectangle (2.29,.27);
  \node at (0,1.20) {$U'$};
  \node[gray] at (1.20,.67) {$L$};
  \node[right,inner sep=1pt] at (2.51,-.83) {$R$};
  \node[gray] at (3.5,.67) {$T$};
\end{scope}
\foreach \offset in {0,5.75}{
  \begin{scope}[xshift=\offset cm]
    \draw[gray!65] (0,0)--(4,0) (2,0)--(2,-1.2);
    \foreach \x/\lab in {0/1,1/2,2/3,3/4,4/5}{
      \fill (\x,0) circle (1.7pt);
      \node[below,inner sep=1pt,font=\scriptsize] at (\x,0) {$\lab$};
    }
    \fill (2,-1.2) circle (1.7pt);
    \node[anchor=east,inner sep=1pt,font=\scriptsize]
      at (1.92,-1.2) {$6$};
  \end{scope}
}
\node at (5.0,-.30) {$\longleftrightarrow$};
\end{tikzpicture}

\smallskip
{\footnotesize
\begin{tabular}{c}
$r=3,\quad L=\{2,3\},\quad R=\{3,6\},\quad K=\{1\}$\\
$U'=\{1,2,3,6\},\quad T=\{4,5\}$\\[2pt]
$F=(s_{U'}s_T)^{-1},\qquad F'=(s_Ls_T)^{-1}$\\[2pt]
Both occurrences have denominator $s_Ls_Rs_{U'}s_T$ and opposite signs.
\end{tabular}\par}
\smallskip
{\scriptsize
\tikz[baseline=-.5ex]{\draw[black,line width=1.1pt] (0,0)--(.4,0);}
 explicit kernel channels \qquad
\tikz[baseline=-.5ex]{\draw[gray,line width=.7pt] (0,0)--(.4,0);}
 channels in lower factors\par}
\caption{Pairing that cancels the incompatible contributions.}
\label{fig:part I}
\end{figure}
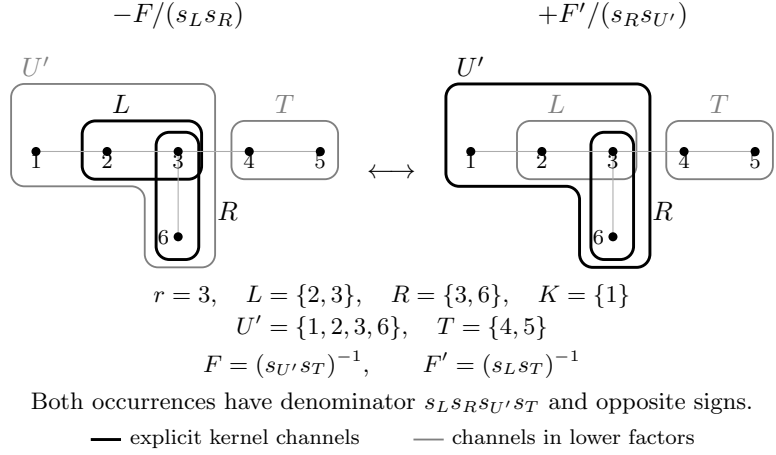

\smallskip
\noindent
\emph{Part II: compatible survivors correspond bijectively to maximal tubings.}

An unpaired positive contribution with explicit $R,U$ has top split
$L=\{r\}\sqcup K$ in its left hierarchy. Its lower left channels
are therefore disjoint from all channels in the right hierarchy.
Together with the quotient channels and the explicit channels $R,U$,
they form a compatible family. For a proper cell its cardinality is
\begin{equation}\label{eq:treesurvivorcount}
 (|L|-2)+(|R|-2)+(N-|U|-1)+2=N-2,
\end{equation}
using $|L|+|R|=|U|+1$. For the full cell the count is
$(|L|-2)+(|R|-2)+1=N-2$. Thus every survivor is a maximal tubing.

Conversely, take any maximal tubing and its augmented binary hierarchy.
Among the sets containing $r$, let $U$ be the smallest containing
vertices from both $A$ and $B$; $U=V$ is allowed. Let $D$ be its child
containing $r$, and $K$ its other child. The connected set $K$ lies
entirely on one side, say $K\subseteq A$. Minimality of $U$ then gives
$\emptyset\ne D\setminus\{r\}\subseteq B$. Set
\begin{equation}\label{eq:treerecovery}
 L=\{r\}\cup K,\qquad R=D.
\end{equation}
The edge joining $K$ to $D$ must end at $r$, since no component of
$G-r$ intersects both $A$ and $B$. Hence $L$ is connected.
Build its hierarchy with children $\{r\}$ and $K$, retaining all
descendants of $K$. Retain the hierarchy on $D$ as the right factor,
and contract $U$ to obtain the outside hierarchy. This is exactly one
unpaired positive contribution, with explicit $R,U$ when $U$ is
proper and explicit $R$ when $U=V$. If $K\subseteq B$, exchange left
and right. Fig.~\ref{fig:part II} illustrates two scenarios where a maximal tubing is mapped to a positive contribution.
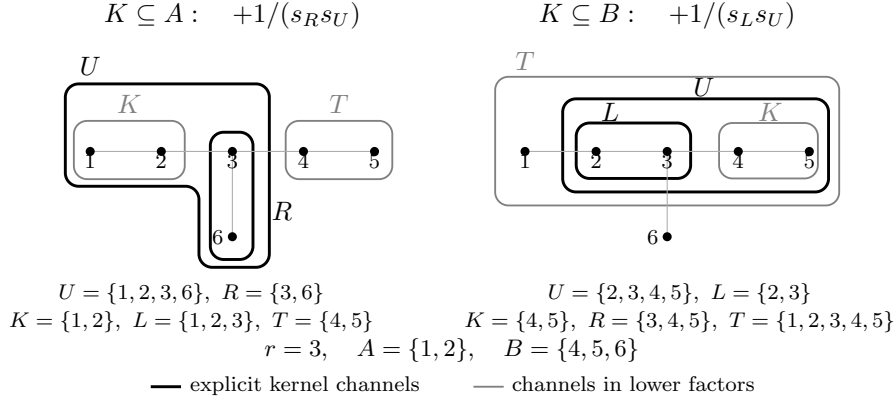
\begin{figure}[H]
\centering
\begin{tikzpicture}[x=0.94cm,y=0.94cm,
  explicit/.style={black,line width=1.1pt,rounded corners=5pt},
  lower/.style={gray,line width=.7pt,rounded corners=5pt},
  every node/.style={font=\small}]
\begin{scope}
  \node[anchor=south] at (2,1.58)
    {$K\subseteq A:\quad +1/(s_Rs_U)$};
  \draw[explicit] (-.35,.95)--(2.52,.95)--(2.52,-1.63)
    --(1.53,-1.63)--(1.53,-.49)--(-.35,-.49)--cycle;
  \draw[explicit] (1.69,-1.52) rectangle (2.29,.27);
  \draw[lower] (-.23,-.39) rectangle (1.33,.43);
  \draw[lower] (2.75,-.39) rectangle (4.25,.43);
  \node at (0,1.20) {$U$};
  \node[gray] at (.55,.67) {$K$};
  \node[right,inner sep=1pt] at (2.51,-.83) {$R$};
  \node[gray] at (3.5,.67) {$T$};
\end{scope}
\begin{scope}[xshift=5.75cm]
  \node[anchor=south] at (2,1.58)
    {$K\subseteq B:\quad +1/(s_Ls_U)$};
  \draw[lower] (-.42,-.76) rectangle (4.42,1.05);
  \draw[explicit] (.53,-.57) rectangle (4.27,.74);
  \draw[explicit] (.72,-.38) rectangle (2.33,.40);
  \draw[lower] (2.73,-.38) rectangle (4.12,.40);
  \node[gray] at (0,1.28) {$T$};
  \node at (2.52,.91) {$U$};
  \node at (1.20,.56) {$L$};
  \node[gray] at (3.45,.56) {$K$};
\end{scope}
\foreach \offset in {0,5.75}{
  \begin{scope}[xshift=\offset cm]
    \draw[gray!65] (0,0)--(4,0) (2,0)--(2,-1.2);
    \foreach \x/\lab in {0/1,1/2,2/3,3/4,4/5}{
      \fill (\x,0) circle (1.7pt);
      \node[below,inner sep=1pt,font=\scriptsize] at (\x,0) {$\lab$};
    }
    \fill (2,-1.2) circle (1.7pt);
    \node[anchor=east,inner sep=1pt,font=\scriptsize]
      at (1.92,-1.2) {$6$};
  \end{scope}
}
\end{tikzpicture}

\smallskip
{\scriptsize
\begin{tabular}{@{}c@{\hspace{1.2cm}}c@{}}
$U=\{1,2,3,6\},\ R=\{3,6\}$ &
$U=\{2,3,4,5\},\ L=\{2,3\}$\\
$K=\{1,2\},\ L=\{1,2,3\},\ T=\{4,5\}$ &
$K=\{4,5\},\ R=\{3,4,5\},\ T=\{1,2,3,4,5\}$
\end{tabular}\par}
{\footnotesize $r=3,\quad A=\{1,2\},\quad B=\{4,5,6\}$\par}
\smallskip
{\scriptsize
\tikz[baseline=-.5ex]{\draw[black,line width=1.1pt] (0,0)--(.4,0);}
 explicit kernel channels \qquad
\tikz[baseline=-.5ex]{\draw[gray,line width=.7pt] (0,0)--(.4,0);}
 channels in lower factors\par}
\caption{Maximal tubings that are mapped to positive cells.}
\label{fig:part II}
\end{figure}

The smallest mixed set $U$ and its children are unique, so this
recovery is a bijection. After Part I, every surviving
monomial occurs once with coefficient $+1$, and their sum is
eq.~\eqref{eq:warmup-tree-sum}. This proves eq.~\eqref{eq:rootedtree};
the rectangle relation then gives eq.~\eqref{eq:treehosfs}.
\end{proof}

The independent choices of the three lower tubings give the
complete lower tree functions in each exact splitting term.

\subsubsection{Two adjacent branching vertices}
\label{app:treecatalogue}
Consider the tree
\begin{equation}\label{eq:treetwobranchgraph}
 E(G)=\{12,23,34,25,36\},\qquad V=\{1,\ldots,6\},\qquad q=7.
\end{equation}
Choose $r=2$, $A=\{1\}$ and $B=\{3,4,5,6\}$.
The rooted right channels are
\[
 \mathcal S=\{23,25,234,235,236,2345,2346,2356,23456\}.
\]
Writing $U_R=\{1\}\cup R$ and
$C_R=\sum_{b\in R\setminus\{2\}}c_{1b}$, the exact splitting formula is
\begin{equation}\label{eq:treeendpointformula}
 \psi_G=\sum_{R\in\mathcal S}C_R\,
 \mathcal J_G[12]\,\mathcal J_G[R]\,
 \mathcal J_G^{\mathrm{out}}[U_R].
\end{equation}
For example, the $R=23$ term is
\begin{equation}
 \frac{c_{13}}{s_{12}s_{23}s_{123}}
 \psi_{[123]-(4,5,6)}^{\langle123\rangle}.
\end{equation}
Its outside factor is a branched tree function:
\begin{align}\label{eq:treebranchedquotient}
 \psi_{[123]-(4,5,6)}^{\langle123\rangle}
 ={}&\frac1{s_{1234}s_{12345}}+\frac1{s_{1235}s_{12345}}
    +\frac1{s_{1234}s_{12346}}\notag\\
 &+\frac1{s_{1236}s_{12346}}+\frac1{s_{1235}s_{12356}}
    +\frac1{s_{1236}s_{12356}}.
\end{align}

\subsection{Splittings adapted to the new zeros}
\label{app:newsplittings}
The term counts below refer to the displayed constructions;
minimum lengths are not considered here.

\paragraph{Inheritance of the linear conditions.}
\label{app:alphabet}\label{lem:inherit}
For an inherited polygon $P=(p_1<\cdots<p_m)$ cut at
$(p_m,p_1)$, replace each physical label $p_u$, $1\leq u\leq m-1$,
by its parent block $[p_u,p_{u+1}-1]$. The last label $p_m$ is
the auxiliary endpoint and is not replaced by a block. This preserves
interval unions and intersections. A subtractive channel that becomes a boundary
edge drops out, so a nonzero $P/U/I/Q$ condition remains a
condition of the same alphabet on the parent path.

\subsubsection{Composition and two-rectangle zeros}
The first row of the cubic exact splitting formula has the common factor
$A[1,\ldots,j]$. Summing its remaining factors by the skinny
exact splitting formula gives
\begin{align}\label{eq:rowgrouped}
 A_n={}&A[1,\ldots,j]A[1,j-1,j,\ldots,n]\notag\\
 &+\sum_{a=2}^{j-2}\sum_{k=j+1}^{n}
 \frac{X_{aj}+X_{j-1,k}-X_{ak}}{X_{aj}X_{j-1,k}X_{ak}}\,
 A[a,\ldots,j]A[j-1,\ldots,k]
 A[k,\ldots,n,1,\ldots,a].
\end{align}
The outside amplitude retains the physical channel $X_{1j}$,
although it is a boundary edge of the first factor.

Let $Y$ be a regular rank-$r$ zero of $A[1,\ldots,j]$,
with $4\leq j<n$, and suppose it has a $K$-term splitting. Substituting that identity into the first product
in eq.~\eqref{eq:rowgrouped} gives
\begin{equation}\label{eq:rowcounts}
 \widehat r=r+(j-3)(n-j),\qquad
 \widehat K=K+(j-3)(n-j).
\end{equation}
The added equations are
\begin{equation}\label{eq:composition-conditions}
 X_{ak}=X_{aj}+X_{j-1,k},\qquad
 2\leq a\leq j-2,\quad j+1\leq k\leq n.
\end{equation}
Together with the child equations, they define a regular zero
$\widehat Y$ of $A_n$. They solve distinct channels outside the
child polygon, while the variables $X_{j-1,k}$ remain free. The inherited conditions remain in
the parent alphabet by the preceding argument. In particular,
a strict child splitting gives a strict parent splitting.
Applying this construction to the ordinary rectangle inside
$A[1,\ldots,m]$ proves the strict splitting for the
two-rectangle family, with the rank in
eq.~\eqref{eq:rectcorerank}.

At six points, use the two-term exact splitting formula of $A[1,2,3,4,5]$
inside the first product. This gives
\begin{equation}\label{eq:two-rectangle-six-splitting}
\begin{aligned}
A_6={}&\frac{X_{14}+X_{35}}{X_{14}X_{35}}
 A[1,2,3,4]A[1,4,5,6]\\
 &+\frac{X_{24}+X_{35}-X_{25}}{X_{24}X_{35}X_{25}}
 A[1,4,5,6]\\
 &+\frac{X_{25}+X_{46}-X_{26}}{X_{25}X_{46}X_{26}}
 A[2,3,4,5]\\
 &+\frac{X_{35}+X_{46}-X_{36}}{X_{35}X_{46}X_{36}}
 A[1,2,3,6].
\end{aligned}
\end{equation}
Its four numerators are the four equations of the six-point
two-rectangle zero. Grouping the first row before substituting
the child identity keeps every numerator linear.

\subsubsection{Three-pair zeros}
\paragraph{Cubic amplitudes.}
At six points, a splitting adapted to the three pair sums is
\begin{equation}\label{eq:three-pair-six-splitting}
\begin{aligned}
 A_6={}&\frac{X_{13}+X_{26}}{X_{13}X_{26}}A[2,3,4,5,6]
 +\frac{X_{24}+X_{35}}{X_{24}X_{35}}A[1,2,3,5,6]\\
 &+\frac{X_{46}+X_{15}}{X_{46}X_{15}}A[1,2,3,4,6]\\
 &-\frac{X_{13}+X_{26}}{X_{13}X_{26}}
 A[1,4,5,6]A[2,3,4,6]\\
 &-\frac{X_{13}+X_{26}}{X_{13}X_{26}}
 A[2,3,4,5]A[2,3,5,6].
\end{aligned}
\end{equation}
To derive it, expand the first three terms into their thirty
reciprocal monomials. After subtracting the fourteen physical
diagrams, the sixteen remaining monomials factor as
\begin{equation}\label{eq:three-pair-surplus}
 \frac{X_{13}+X_{26}}{X_{13}X_{26}}
 \bigl(A[1,4,5,6]A[2,3,4,6]
       +A[2,3,4,5]A[2,3,5,6]\bigr).
\end{equation}
Subtracting these products yields eq.~\eqref{eq:three-pair-six-splitting}.

For $n\geq7$, put $A_{\rm out}=A[1,5,6,\ldots,n]$.
A further regrouping of the six-point amplitude, followed
by eq.~\eqref{eq:rowgrouped} with $j=6$, gives
\begin{align}\label{eq:pairfamilyformula}
 A_n={}&\frac{X_{13}+X_{26}}{X_{13}X_{26}}
 A_{\rm out}A[1,3,4,5,6]
 +\frac{X_{24}+X_{35}}{X_{24}X_{35}}
 A_{\rm out}A[1,2,4,5,6]\notag\\
 &-\frac{X_{46}+X_{15}}{X_{46}X_{15}X_{26}}
 A_{\rm out}A[1,3,4,5]
 -\frac{X_{46}+X_{15}}{X_{35}X_{46}X_{15}X_{14}}
 A_{\rm out}\notag\\
 &+\sum_{a=2}^{4}\sum_{v=7}^{n}
 \frac{X_{a6}+X_{5v}-X_{av}}{X_{a6}X_{5v}X_{av}}\,
 A[a,\ldots,6]A[5,\ldots,v]
 A[v,\ldots,n,1,\ldots,a].
\end{align}
The four terms preceding the sum are $A_{\rm out}$ times
a six-point identity, verified by expanding its complete
lower amplitudes. The formula has $3n-14$ terms, each
containing a complete amplitude with at least four legs.

\paragraph{NLSM amplitudes.}
The six-point NLSM amplitude can likewise be grouped by
the three pair sums:
\begin{align}\label{eq:nlsmpairsplit}
 A_6^{\mathrm{NLSM}}={}&
 -\frac{X_{13}+X_{26}}{X_{36}}
 A^{\mathrm{NLSM}}[3,4,5,6]
 -(X_{13}+X_{26})A^{\mathrm{ext}}[1,5,6]\notag\\
 &-\frac{X_{24}+X_{35}}{X_{25}}
 A^{\mathrm{NLSM}}[1,2,5,6]
 -(X_{24}+X_{35})A^{\mathrm{ext}}[1,5,6]\notag\\
 &-\frac{X_{15}+X_{46}}{X_{14}}
 A^{\mathrm{NLSM}}[1,2,3,4]
 -(X_{15}+X_{46})A^{\mathrm{ext}}[1,5,6].
\end{align}
Here $A^{\mathrm{ext}}[1,5,6]=1$.
For even $n\geq8$, define
\begin{equation}
 \mathcal C^{(6)}_{av}=
 \mathcal J^{\mathrm{NLSM}}[a,\ldots,6]\,
 \mathcal J^{\mathrm{NLSM}}[5,\ldots,v]\,
 \mathcal J^{\mathrm{NLSM}}[v,\ldots,n,1,\ldots,a],
 \quad 2\leq a\leq4,\quad7\leq v\leq n,
\end{equation}
using the projected currents in
eq.~\eqref{eq:nlsm-projected-current}. Projecting the grouped
identity then yields
\begin{equation}\label{eq:nlsmpairfamilyformula}
 A_n^{\mathrm{NLSM}}=
 A^{\mathrm{ext}}[1,5,6,\ldots,n]\,
 A^{\mathrm{NLSM}}[1,\ldots,6]
 +\sum_{a=2}^{4}\sum_{v=7}^{n}
 (X_{a6}+X_{5v}-X_{av})\,\mathcal C^{(6)}_{av}.
\end{equation}
The mixed outside amplitude has scalar legs $1,5,n$.
Substituting eq.~\eqref{eq:nlsmpairsplit} into its six-point
factor gives $3n-12$ terms whose numerators are the equations
of the three-pair zero.

\subsubsection{Balanced-column, row--column and opposite-end zeros}
\label{app:coverformulas}
In each selected exact splitting cell, write its vanishing four-point
factor as $A_4=Q/(uv)$, where $Q=u+v$. Expanding the cell's
cubic kernel gives
\begin{equation}\label{eq:coveratomicexpansion}
 \frac{x+y-z}{xyz}\frac{Q}{uv}\prod_t A[P_t]
 =\left(\frac{Q}{yzuv}+\frac{Q}{xzuv}
              -\frac{Q}{xyuv}\right)\prod_t A[P_t].
\end{equation}
At the intersection of a selected row and column, expand
one of the two four-point factors. Leave the other lower
factors and all unselected cells unchanged. This gives
generic identities whose numerators span the equations of
the balanced-column, row--column and opposite-end zeros.
With $p=j-2$ and $q=n-j$, the numbers of selected cells
are $p$, $p+q-1$ and $2$, respectively. Thus the resulting
formulas have at most
\begin{equation}\label{eq:coverlengths}
 p(q+2),\qquad pq+2(p+q-1),\qquad n+1
\end{equation}
terms, respectively.

\subsection{Additional checks for new zeros}
\subsubsection{Survival under the NLSM projection}
\label{app:nlsm-survival}
Families 2--5 and even-$m$ Family 1 are
preserved by the parent $\delta$-shift:
their cubic amplitudes vanish for every
$\delta$, so the leading NLSM coefficient
also vanishes.

For odd $m$, the child rectangle
$Z_{m;k}\subset A^{\phi^3}[1,\ldots,m]$
has proper-cell numerators
\begin{equation}
 \begin{gathered}
 \rho_{av}=X_{ak}+X_{k-1,v}-X_{av},\\
 1\leq a\leq k-2,\quad k+1\leq v\leq m,\quad (a,v)\ne(1,m),
 \end{gathered}
 \label{eq:nlsm-child-numerators}
\end{equation}
which remain invariant. Its full-cell
numerator $X_{1k}+X_{k-1,m}$ instead
shifts by $-\delta$, because the child
boundary $X_{1m}$ is absent while the
parent channel shifts. Only that cell
survives, giving the exact cubic identity
\begin{equation}\label{eq:nlsm-odd-child-shift}
 A_m^{\phi^3}(X^{(\delta)})\big|_{Z_{m;k}}
 =-\delta\,
 \frac{A^{\phi^3}[1,\ldots,k](X^{(\delta)})}
      {X_{1k}^{(\delta)}}
 \frac{A^{\phi^3}[k-1,\ldots,m](X^{(\delta)})}
      {X_{k-1,m}^{(\delta)}}.
\end{equation}
One lower polygon is even and the other odd.
Projecting their complete amplitudes and
attaching channels yields
\begin{equation}\label{eq:nlsm-odd-child-product}
 A_m^{\mathrm{ext}}[1,\ldots,m]\big|_{Z_{m;k}}
 =\begin{cases}
 \displaystyle
 \frac{A^{\mathrm{NLSM}}[1,\ldots,k]
       A^{\mathrm{ext}}[k-1,\ldots,m]}{X_{1k}},
 &k\ \text{even},\\[7pt]
 \displaystyle
 \frac{A^{\mathrm{ext}}[1,\ldots,k]
       A^{\mathrm{NLSM}}[k-1,\ldots,m]}{X_{k-1,m}},
 &k\ \text{odd}.
 \end{cases}
\end{equation}
The outside polygon gives the pion factor,
so the parent result is
\begin{equation}
 \boxed{
 A_n^{\mathrm{NLSM}}\big|_{\mathcal R_{n;m,k}}
 =A^{\mathrm{NLSM}}[1,m-1,m,\ldots,n]
 \begin{cases}
 \displaystyle
 \frac{A^{\mathrm{NLSM}}[1,\ldots,k]\,
       A^{\mathrm{ext}}[k-1,\ldots,m]}{X_{1k}},
       &k\ \text{even},\\[7pt]
 \displaystyle
 \frac{A^{\mathrm{ext}}[1,\ldots,k]\,
       A^{\mathrm{NLSM}}[k-1,\ldots,m]}{X_{k-1,m}},
       &k\ \text{odd}.
 \end{cases}}
 \label{eq:nlsmoddcore}
\end{equation}
Here $n$ is even, $m$ is odd, and each
mixed amplitude has scalar legs at its
first and last two displayed positions.
The lower internal channels are free,
so the product is generically nonzero.
At $(n,m,k)=(6,5,4)$ it gives
eq.~\eqref{eq:nlsmliftedstrict}.
\subsubsection{Exclusion of blob partitions for the branched trees}
\label{app:tree-zero-nonblob}
Write $t_{ab}=s_{ab}$ for pairs of vertices in the original tree and use
$s_S=\sum_{\{a,b\}\subset S}t_{ab}$. We take the $t_{ab}$ as independent
coordinates after eliminating
\[
 t_{12}=-\sum_{(a,b)\ne(1,2)}t_{ab}.
\]
The coefficient matrix $M_G$ below is that of the corresponding
new-zero equations in the main text.
To include every blob zero of ref.~\cite{CosmologicalWavefunctions}, augment the original
tree by joining its auxiliary vertex $q$ to every leaf. A blob partition
is a partition of the components left by a two-vertex cut into nonempty
sets $A,B$, with equations $c_{ab}=0$ for $a\in A$, $b\in B$.
Here complementary choices of $A$ describe the same partition.
Let $M_G$ be the coefficient matrix of the new zero and $B$ that of these
cross conditions, in the independent coordinates above. Conditions
involving $q$ are evaluated using
$t_{aq}=-\sum_{b\in V\setminus\{a\}}t_{ab}$ and $c_{ab}=-t_{ab}$.

For $222$ and $2222$, the only two-vertex cut is $\{3,q\}$.
For $222$, its three partitions each give
$\operatorname{rank}\binom{M_G}{B}=5>3$.
For $2222$, the four partitions separating one leaf from three give rank
$7$, and the three partitions separating two leaves from two give rank
$8$, each larger than $4$.
For $223$, the complete list is
\[
\begin{array}{c|c|c|c}
\text{cut}&A&B&\operatorname{rank}\binom{M_{223}}{B}\\ \hline
\{3,6\}&\{1\}&\{2,4,5\}&7\\
\{3,6\}&\{2\}&\{1,4,5\}&5\\
\{3,6\}&\{1,2\}&\{4,5\}&7\\
\{4,6\}&\{5\}&\{1,2,3\}&5\\
\{3,5\}&\{4\}&\{1,2,6\}&7
\end{array}
\]
The last partition has nonroot distinguished labels and imposes
$c_{14}=c_{24}=c_{46}=0$. Thus there are five blob partitions for $223$,
including four rooted ones. Every stacked rank exceeds $\operatorname{rank}M_G$,
so the equations of no blob zero are implied by the new locus. None of the
three new zeros is contained in any blob zero.

\subsubsection{Rigidity of ordinary five-point zeros}
Write $(a,b,c,d,e)=(X_{13},X_{14},X_{24},X_{25},X_{35})$.
For the ordinary zero $d=-a$, $b=a+c$, a nearby
rank-two linear locus has first-order changes
$b=a+c+\epsilon B$ and $d=-a+\epsilon D$, with
$B,D$ arbitrary linear forms in $(a,c,e)$.
The zero condition at order $\epsilon$ is
\begin{equation}\label{eq:ordinary-five-rigidity}
 ae\,B+(a+c)(c+e)\,D=0.
\end{equation}
The two quadratic factors are coprime, so the
linear form $D$ must vanish; then $B=0$. Hence the
ordinary five-point zero is isolated in the space of
rank-two linear zeros. Cyclicity covers the others.

\clearpage

\section{Supplementary checks for the five cubic zero families}
\label{app:cubic-regularity}

Here we check regularity and distinguish the five new families
from ordinary rectangular zeros of the ordered
$\operatorname{Tr}(\phi^3)$ amplitude.

\subsection{Conventions and family equations}

We work on the complex vector space $\K_n$ whose independent coordinates are
the physical planar channels
\begin{equation}
 X_{ab},\qquad 1\leq a<b\leq n,\quad b-a\geq2,
 \quad (a,b)\ne(1,n).
 \label{cubic:eq:channels}
\end{equation}
Its dimension is $d_n=n(n-3)/2$. Channel endpoints are unordered,
$X_{ab}=X_{ba}$. Boundary channels are zero, and no dimension-dependent
Gram relations are imposed. With $A_3=1$, the cubic
amplitude is the sum over triangulations of the ordered polygon,
\begin{equation}
 A_n=\sum_{T\in\operatorname{Tri}(1,\ldots,n)}
               \prod_{e\in T}\frac1{X_e}.
 \label{cubic:eq:triangulations}
\end{equation}
For an ordered subpolygon, $A[S]$ uses the parent channel on every internal
diagonal. A boundary edge of $S$ contributes no propagator; its parent
channel is not set to zero.

A linear subspace $Z\subset\K_n$ is a \emph{regular linear zero} if
$A_n|_Z=0$ as a rational identity and no physical channel vanishes
identically on $Z$. Thus $Z$ has a nonempty open subset on which every
propagator is finite. Let $W_Z\subset\K_n^*$ be its space of defining
linear equations and let $r(Z)=\dim W_Z$.

For later use define the exact splitting numerators
\begin{equation}
 R^{(j)}_{av}=X_{aj}+X_{j-1,v}-X_{av},\qquad
 1\leq a\leq j-2,\quad j+1\leq v\leq n.
 \label{cubic:eq:rho}
\end{equation}
In particular $R^{(j)}_{1n}=X_{1j}+X_{j-1,n}$, since $X_{1n}=0$.
For clarity, we restate the channel equations of Families 1--5 in
Section~\ref{sec:branches}. The notation below agrees with the main text.

\paragraph{Two-rectangle zeros $\Rz_{n;m,k}$.}
Let $5\leq m<n$ and $3\leq k\leq m-1$. Define
\begin{align}
 E_{ac}&=X_{ak}+X_{k-1,c}-\bar X_{ac}=0,
 &&1\leq a\leq k-2,\quad k+1\leq c\leq m,\label{cubic:eq:Rchild}\\
 S_{av}&=X_{am}+X_{m-1,v}-X_{av}=0,
 &&2\leq a\leq m-2,\quad m+1\leq v\leq n.\label{cubic:eq:Router}
\end{align}
Here $\bar X_{1m}=0$ and $\bar X_{ac}=X_{ac}$ otherwise. The bar is
essential: \eqref{cubic:eq:Rchild} describes an ordinary rectangular zero of
$A[1,\ldots,m]$, in which $(1,m)$ is a boundary edge. The physical channel
$X_{1m}$ of $A_n$ remains free.

\paragraph{Three-pair zeros $\Pz_n$.}
For $n\geq6$, impose
\begin{align}
 X_{13}+X_{26}=X_{24}+X_{35}=X_{46}+X_{15}&=0,\label{cubic:eq:Pseed}\\
 X_{a6}+X_{5v}-X_{av}&=0,
 &&a=2,3,4,\quad 7\leq v\leq n.\label{cubic:eq:Pextension}
\end{align}
The second line is empty at $n=6$.

\paragraph{Balanced-column zeros $\Cz_{n,j}$.}
For $3\leq j\leq n-3$, impose
\begin{equation}
 X_{j-1,j+1}+X_{j,j+2}=0,\qquad
 R^{(j)}_{av}=0\quad(v\ne j+2),
 \label{cubic:eq:Cfamily}
\end{equation}
with the row and column ranges in \eqref{cubic:eq:rho}.

\paragraph{Row--column zeros $\Xz_{n,j}$.}
For $5\leq j\leq n-3$, impose
\begin{align}
 X_{j-3,j-1}+X_{j-2,j}&=0,\qquad
 X_{j-1,j+1}+X_{j,j+2}=0,\nonumber\\
 R^{(j)}_{av}&=0\quad(a\ne j-3,\ v\ne j+2).
 \label{cubic:eq:Xfamily}
\end{align}

\paragraph{Opposite-end zeros $\Dz_n$.}
For $n\geq8$, impose
\begin{align}
 X_{24}+X_{35}&=0,\qquad X_{n-2,n}+X_{1,n-1}=0,\nonumber\\
 Q_v:=X_{13}+X_{2v}-X_{1v}&=0,
 \quad v\in\{4,\ldots,n\}\setminus\{5,n-2\}.
 \label{cubic:eq:Dfamily}
\end{align}

The ranks of the five families are
\begin{equation}
\begin{array}{c|c}
\text{Family}&r(Z)\\ \hline
\Rz_{n;m,k}&(k-2)(m-k)+(m-3)(n-m)\\
\Pz_n&3n-15\\
\Cz_{n,j}&(j-2)(n-j-1)+1\\
\Xz_{n,j}&(j-3)(n-j-1)+2\\
\Dz_n&n-3.
\end{array}
\label{cubic:eq:ranks}
\end{equation}

\subsection{Vanishing, independence of the equations and regularity}
\label{cubic:sec:regularity}

We recall why these subspaces are regular zeros.
The generic cubic exact splitting formula, eq.~\eqref{eq:rooted}, is
\begin{equation}
 A_n=\sum_{a=1}^{j-2}\sum_{v=j+1}^{n}
 R^{(j)}_{av}\,
 J[a,\ldots,j]\,J[j-1,\ldots,v]\,
 J[v,\ldots,n,1,\ldots,a],
 \label{cubic:eq:HOSF}
\end{equation}
where $J[s_1,\ldots,s_t]=A[s_1,\ldots,s_t]/X_{s_1s_t}$ for
$t\geq3$ and the full-cell outside factor is $J[n,1]=1$.
All lower amplitudes use the inherited-channel convention above.

For $\Rz_{n;m,k}$, take $j=m$. Equations \eqref{cubic:eq:Rchild} put the
first row's common factor $A[1,\ldots,m]$ on its ordinary rectangular
zero; \eqref{cubic:eq:Router} kills every other row. At six points the
$j=3$ formula proves \eqref{cubic:eq:Pseed}: the full-cell numerator is
$X_{13}+X_{26}$, and the other two terms contain the four-point factors
with numerators $X_{15}+X_{46}$ and $X_{24}+X_{35}$. For $n>6$, use
$j=6$: the first row contains this vanishing six-point amplitude, and
\eqref{cubic:eq:Pextension} kills the remaining rows. For $\Cz$ and $\Xz$,
the pair equations set the selected four-point column or row factors
to zero; every unselected term has $R^{(j)}_{av}=0$. Finally, for
$\Dz$ use $j=3$. Its two pair equations kill the terms at $v=5$ and
$v=n-2$, and the remaining terms have $Q_v=0$.
These are rational identities, not conclusions drawn from numerical
sampling.

The displayed equations are independent, and their solutions contain
regular points. We give the details because the residue proof will also
use the free coordinates.

\paragraph{Two-rectangle and three-pair families.}
For the child rectangle in \eqref{cubic:eq:Rchild}, set
$x_a=X_{ak}$ and $y_c=X_{k-1,c}$. Its full-cell equation sets
$y_m=-x_1$. Every other child equation solves a distinct cross-channel,
\begin{equation}
 X_{ac}=x_a+y_c,\qquad (a,c)\ne(1,m).
 \label{cubic:eq:child-param}
\end{equation}
In the last column this is $X_{am}=x_a-x_1$ for $a>1$.
Choose the remaining $x_a,y_c$ and the channels internal to the two
child polygons freely. No physical child channel is the zero linear
form. The extension equations then solve distinct new channels
$X_{av}=X_{am}+X_{m-1,v}$, with every $X_{m-1,v}$ free.
Thus they add exactly $(m-3)(n-m)$ independent equations and force no
parent channel to vanish. This proves the first rank in \eqref{cubic:eq:ranks}
and regularity.

For $\Pz_6$, solve \eqref{cubic:eq:Pseed} for
$X_{26}=-X_{13}$, $X_{35}=-X_{24}$ and $X_{46}=-X_{15}$.
All other six-point channels are free. For $n>6$, choose every
$X_{5v}$ freely and solve \eqref{cubic:eq:Pextension} for the distinct
$X_{av}$. This gives rank $3+3(n-6)$ and regularity.

\paragraph{Balanced-column and row--column families.}
Set $x_a=X_{aj}$ and $y_v=X_{j-1,v}$. In both families the full cell
is constrained, so $y_n=-x_1$. Every other constrained cell solves
$X_{av}=x_a+y_v$. The right pair equation solves
$X_{j,j+2}=-y_{j+1}$; the left pair equation, when present, solves
$X_{j-3,j-1}=-x_{j-2}$. These are different channels from the solved
cross-channels. The unconstrained cells and all remaining channels are
free. This proves the two ranks in \eqref{cubic:eq:ranks}. All solved
physical channels are nonzero linear forms in the free coordinates,
including $x_a-x_1$ in the last column with $a>1$.

\paragraph{Opposite-end family.}
Put $t=X_{13}$ and $f_v=X_{1v}$, with $f_n=0$. The $Q_v$ equations
solve the distinct channels $X_{2v}=f_v-t$. The pair equations solve
\begin{equation}
 X_{35}=t-f_4,\qquad X_{n-2,n}=-f_{n-1}.
 \label{cubic:eq:D-param}
\end{equation}
All remaining channels are free. These are $n-3$ independent
equations, and again none makes a physical channel the zero linear
form. In each family, a generic choice of the free coordinates avoids
the finite union of channel hyperplanes, proving regularity.

\subsection{Distinction from ordinary rectangular zeros}
\label{cubic:sec:newness}
A contraction is \emph{forced} on $Z$ if it vanishes identically on $Z$.
We denote the set of such nonadjacent cyclic contractions by
$\operatorname{Forced}(Z)$, counting each unordered pair once.

\paragraph{Two-rectangle and three-pair families.}
An ordinary rank-$r$ rectangle sets $r$ distinct nonadjacent contractions
$c_{ij}$ to zero. This remains true after any cyclic relabeling. We show
that a two-rectangle zero forces exactly $r-1$ such contractions to vanish,
whereas a three-pair zero forces exactly $r-3$. They therefore cannot equal
an ordinary rectangle.

To count all cyclic contractions, use the independent coordinates
\[
 F_\ell=X_{1\ell}\quad(3\leq\ell\leq n-1),\qquad
 c_{ij}\quad(1\leq i<j\leq n-1,\ j-i\geq2),
 \qquad F_2=F_n=0.
\]
These are a linear coordinate system on the abstract planar kinematic
space: repeated use of the rectangle relation gives the inverse
\begin{equation}\label{eq:novelty-fan-inverse}
 X_{ij}=F_j-F_{i+1}
       +\sum_{a=1}^{i-1}\sum_{b=i+1}^{j-1}c_{ab}.
\end{equation}
The nonadjacent contractions missing from this coordinate list are exactly
$c_{i,n}$ for $2\leq i\leq n-2$. Their expressions are
\begin{equation}\label{eq:novelty-wrap}
 c_{i,n}=F_{i+2}-F_i
       +\sum_{a=1}^{i-1}c_{a,i+1}
       -\sum_{b=i+2}^{n-1}c_{i,b}.
\end{equation}
Thus checking this list also checks every cyclic image of every ordinary
rectangle.

For $\mathcal R_{n;m,k}$, finite differences of its two sets of cumulative
conditions give exactly
\begin{align}
 c_{uv}&=0,
 &&1\leq u\leq k-2,\quad k\leq v\leq m-1,
      \quad(u,v)\ne(1,m-1),\nonumber\\
 c_{1,m-1}&=-F_m,\nonumber\\
 c_{uv}&=0,
 &&2\leq u\leq m-2,\quad m\leq v\leq n-1.
 \label{eq:novelty-R-mesh}
\end{align}
The exceptional equation is due to the child boundary convention
$\bar X_{1m}=0$: its child plaquette equals the parent plaquette plus
$X_{1m}$. Conversely, summing these plaquette equations recovers every
original cumulative condition. The two displayed zero blocks are disjoint
and contain
\[
 (k-2)(m-k)-1+(m-3)(n-m)=r(\mathcal R_{n;m,k})-1
\]
entries. Every unlisted coordinate $c_{ij}$ and every fan coordinate is
free. The exceptional $c_{1,m-1}=-F_m$ is not identically zero, since
$m<n$.
In \eqref{eq:novelty-wrap}, the only constrained plaquette that contributes
a fan term is $c_{1,m-1}$. Consequently the fan part of $c_{i,n}$ is
\[
 F_{i+2}-F_i-\delta_{i,m-2}F_m.
\]
For $i=m-2$ this is $-F_{m-2}$; for every other allowed $i$ it is
$F_{i+2}-F_i$. Each is a nonzero linear form in the independent fan
coordinates, including the endpoints $F_2=F_n=0$. Hence no wrapping
contraction vanishes identically, proving the exact count.

For $\mathcal P_n$, its extension conditions are equivalent to
\begin{equation}\label{eq:novelty-P-strip}
 c_{uv}=0,\qquad u=2,3,4,\quad 6\leq v\leq n-1.
\end{equation}
The three seed equations, expressed in the same coordinates, solve
\begin{equation}\label{eq:novelty-P-pivots}
 c_{15}=-F_6-c_{13}-c_{14},\qquad
 c_{24}=F_3-F_5-c_{13}-c_{14},\qquad
 c_{35}=c_{13}+c_{14}-c_{25}.
\end{equation}
Here $F_6=0$ when $n=6$. All fan coordinates and all other nonwrapping
plaquettes outside \eqref{eq:novelty-P-strip} are free. In particular
$c_{13},c_{14},c_{25}$ are independent, so none of the three expressions
in \eqref{eq:novelty-P-pivots} is identically zero. Thus the strip gives
exactly $3(n-6)=r(\mathcal P_n)-3$ vanishing nonwrapping contractions.
For the wrapping contractions, \eqref{eq:novelty-wrap} reduces to
\begin{align}
 c_{2,n}&=F_4-F_3+F_5+2c_{13}+c_{14}-c_{25},\nonumber\\
 c_{3,n}&=-2c_{13}-c_{14}+c_{25},\nonumber\\
 c_{4,n}&=-F_4.
 \label{eq:novelty-P-wrap-small}
\end{align}
All three are nonzero linear forms in free coordinates. If
$5\leq i\leq n-2$, no pivot plaquette in
\eqref{eq:novelty-P-pivots} occurs in \eqref{eq:novelty-wrap}; its fan
part is therefore $F_{i+2}-F_i$, which is again nonzero. This proves the
claimed all-cyclic count and the distinction from every ordinary
rectangle.

\paragraph{Balanced-column, row--column and opposite-end families.}
We now count the forced contractions for the three families in
\eqref{cubic:eq:Cfamily}--\eqref{cubic:eq:Dfamily}. As above, each
unordered nonadjacent cyclic pair is counted once; in particular the
count includes $c_{u,n}$ for $2\leq u\leq n-2$.

\emph{The generating rectangle of the cover families.}
Put $p=j-2$, $t=j+2$, and, for the row--column family only, $s=j-3$.
Use
\[
 x_a=X_{aj}\quad(1\leq a\leq p),\qquad
 y_v=X_{j-1,v}\quad(j+1\leq v\leq n),\qquad y_n=-x_1.
\]
Extend $x_{j-1}=0$, $y_j=0$, and write
\[
 \rho_{av}=x_a+y_v-X_{av},\qquad
 \rho_{a,j}=\rho_{j-1,v}=0.
\]
In the balanced-column family, the free residuals are precisely
$\rho_{a,t}$, $1\leq a\leq p$; all other $\rho$ vanish.
In the row--column family, the free residuals are those in row $s$
or column $t$, with their intersection counted once. All other
$\rho$ vanish. These residuals may replace the corresponding free
cross-channels as independent coordinates. The pair equations merely set
\[
 X_{j,j+2}=-y_{j+1},\qquad
 X_{s,j-1}=-x_{s+1}\quad\text{in the row--column family};
\]
neither imposes a relation among the free residuals.

For $1\leq u\leq p$, $j\leq v\leq n-1$,
\[
 c_{uv}=-\rho_{uv}-\rho_{u+1,v+1}
             +\rho_{u,v+1}+\rho_{u+1,v}.
\]
For the balanced-column family this vanishes exactly when
$v\notin\{j+1,j+2\}$. At the two released columns its values are
$\rho_{u,t}-\rho_{u+1,t}$ and its negative; these are nonzero
linear forms. For the row--column family it vanishes exactly when
\[
 u\notin\{j-4,j-3\},\qquad v\notin\{j+1,j+2\}.
\]
Indeed, every other plaquette contains at least one independent free
residual, each with coefficient $+1$ or $-1$ and occurring only once.

\emph{No further cover-family contraction vanishes.}
The following argument checks the full cyclic set, not just the displayed
rectangle. Partition the labels as
$A=\{1,\ldots,j-2\}$, $\{j-1\}$,
$B=\{j,\ldots,n-1\}$, $\{n\}$.
For a nonadjacent pair within $A$, the free channel $X_{uv}$ occurs
once with coefficient $+1$ in $c_{uv}$. For a pair within $B$, the
free channel $X_{u+1,v+1}$ does so. Neither coordinate is eliminated
by either pair equation.

For $u\in A$ paired with $j-1$, put $Z_u=X_{u,j-1}$.
Here $u\leq j-3$ and
\[
 c_{u,j-1}=Z_u-Z_{u+1}+x_{u+1}-x_u.
\]
The coordinate $Z_u$ is free unless this is the row--column family
and $u=s$; in that exceptional case the expression is $-x_s$,
because $Z_s=-x_{s+1}$ and $Z_{s+1}=0$.
For $j-1$ paired with $v\in B$, necessarily $v\geq j+1$, and
\[
 c_{j-1,v}=y_v-y_{v+1}+X_{j,v+1}-X_{jv}.
\]
The channel $X_{j,v+1}$ is free unless $v=j+1$; in that case the
expression is $-y_{j+2}$. The pair of distinguished labels gives
\[
 c_{j-1,n}=-X_{1,j-1}-X_{jn},
\]
which contains the free channel $X_{jn}$.

It remains to check the cyclic entries containing $n$. Put
$F_a=X_{1a}$. For $u\in A\setminus\{1\}$,
\[
 c_{u,n}=x_u-x_{u+1}+F_{u+1}-F_u
              -\rho_{u,n}+\rho_{u+1,n}.
\]
In the balanced-column family the residual terms vanish, and the
independent coordinate $x_u$ has coefficient $+1$. In the row--column
family the same argument applies except when $u=s-1$ or $s$; then
the independent residual $\rho_{s,n}$ occurs with coefficient $+1$
or $-1$. Finally, for $u\in B\setminus\{n-1\}$, the expression
\[
 c_{u,n}=X_{un}-X_{u+1,n}+F_{u+1}-F_u
\]
contains $X_{un}$ with coefficient $+1$. This channel is free and
does not occur in the reduced fan expressions $F_u,F_{u+1}$.
These cases exhaust every nonadjacent cyclic pair outside $A\times B$.
Thus the numbers of forced entries are exactly
\[
 \#\operatorname{Forced}(\mathcal C_{n,j})=(j-2)(n-j-2),\qquad
 \#\operatorname{Forced}(\mathcal X_{n,j})=(j-4)(n-j-2).
\]

\emph{Opposite-end family.}
Put $f_v=X_{1v}$, $f_n=0$, $t=X_{13}$, and
$Q_v=t+X_{2v}-f_v$, with $Q_3=Q_n=0$.
The only nonzero free $Q$'s are $\alpha=Q_5$ and
$\beta=Q_{n-2}$. They are independent because they replace the two
unconstrained channels $X_{25},X_{2,n-2}$. The pair equations set
\[
 X_{35}=t-f_4,\qquad X_{n-2,n}=-f_{n-1}.
\]
For the skinny generating rectangle,
\[
 c_{1v}=Q_{v+1}-Q_v,\qquad 3\leq v\leq n-1.
\]
At $n\geq9$, its four released entries have values
$\alpha,-\alpha,\beta,-\beta$, at indices
$v=4,5,n-3,n-2$, respectively. At $n=8$, the three released entries
are $\alpha,\beta-\alpha,-\beta$, at $v=4,5,6$.
All are nonzero linear forms.

For completeness, no cyclic entry outside that row vanishes. For
$3\leq u<v\leq n-1$, the free channel $X_{uv}$ is a witness,
except at $(u,v)=(3,5)$, where $X_{46}$ is a free witness.
For $c_{2v}$ with $4\leq v\leq n-1$, the free channel
$X_{3,v+1}$ is a witness except at $v=4$, where
$c_{24}=-X_{25}$. The remaining distinguished pair has
$c_{2n}=-X_{3n}$. For $3\leq u\leq n-3$, the cyclic entry
$c_{u,n}$ contains the free channel $X_{un}$ with coefficient $+1$;
the last one is
\[
 c_{n-2,n}=-f_{n-2},
\]
also nonzero. These witnesses remain independent after both pair
equations, since the eliminated channels are only $X_{35}$ and
$X_{n-2,n}$, whose replacement expressions use fan channels.
Therefore the complete forced set is exactly
\[
 c_{1v}=0,\qquad 3\leq v\leq n-1,\qquad
 v\notin\{4,5,n-3,n-2\},
\]
of cardinality $2$ at $n=8$ and $n-7$ for $n\geq9$.

\emph{Comparison with ordinary rectangles.}
An ordinary rectangular zero of rank $r$ forces $r$ distinct
nonadjacent cyclic entries $c_{uv}$ to vanish. In the three families,
using the ranks in \eqref{cubic:eq:ranks}, the difference between the rank and the complete forced-entry count is
\[
 j-1\quad(\mathcal C),\qquad n-3\quad(\mathcal X),\qquad
 3\ (n=8),\ 4\ (n\geq9)\quad(\mathcal D).
\]
Each is positive. Therefore none of these zeros equals an ordinary
rectangle of the same rank, and equality with a rectangle of another
rank is impossible by dimension.

\endgroup

\clearpage
\providecommand{\href}[2]{#2}\begingroup\raggedright\endgroup
\end{document}